\documentclass[pre,twocolumn,showpacs,amsmath,amssymb,noshowpacs]{revtex4}
\usepackage{amsmath,amssymb,amsthm}
\usepackage{dcolumn}
\usepackage{enumerate}
\usepackage{braket}
\usepackage{graphicx,xcolor}
\usepackage{wrapfig}
\usepackage{mathtools}
\usepackage{bm}
\usepackage{multirow}

\usepackage[normalem]{ulem}

\theoremstyle{definition}
\newtheorem{thm}{Theorem}[section]
\newtheorem{lem}[thm]{Lemma}

\theoremstyle{definition}

\newtheorem{rem}[thm]{Remark}

\def\R{{\mathbb R}}
\def\Z{{\mathbb Z}}
\def\C{{\mathbb C}}

\def\id{\mathop{\widehat{\mathrm{id}}}\nolimits}

\newcommand{\kket}[1]{| #1 )}

\definecolor{Forest}{HTML}{288c66}

\newcommand{\rd}[1]{#1}

\newcommand{\bv}[1]{{\boldsymbol #1}}

\newcommand{\bka}[1]{\left ( {#1} \right)}
\newcommand{\bkb}[1]{\left\{ {#1} \right\}}
\newcommand{\bkc}[1]{\left[ {#1} \right]}

\newcommand{\abs}[1]{\left| {#1} \right|}
\newcommand{\dif}{\mathrm{d}}
\newcommand{\bvec}[1]{\boldsymbol{#1}}

\newcommand{\pr}{\widehat{\mathrm{pr}}}
\newcommand{\alit}[1]{\begin{align}\begin{split} #1 \end{split}\end{align}}

\newcommand{\Shift}[2]{f_{#1}^{#2}}
\newcommand{\ShiftOp}[2]{\widehat{\Delta}_{#1}^{#2}}	
\newcommand{\PermOp}[1]{\widehat{\Pi}_{#1}}						
\newcommand{\CheckOp}[2]{\widehat{\Xi}_{#1}^{#2}}			
\newcommand{\StrictCheckOp}[2]{\widehat{\Upsilon}_{#1}^{#2}}			
\newcommand{\TransOp}[2]{\widehat{T}_{#1}^{#2}}										
\newcommand{\BetaTerm}[1]{\widehat{U}_{#1}}										
\newcommand{\SubOp}[1]{\widehat{P}_{#1}}											
\newcommand{\AbstOp}{\widehat{A}}															

\newcommand{\DiffEnergy}[2]{\mathcal{D}_{#1}^{#2}E}		
\newcommand{\dP}{+}				
\newcommand{\dM}{-}				
\newcommand{\dPM}{\pm}		

\newcommand{\sctZero}[1]{\subsection{#1}}							
\newcommand{\sctOne}[1]{\subsubsection{#1}}						

\newcommand{\BOp}[1]{\widehat{B}_{#1}}
\newcommand{\GOp}{\widehat{G}_0}

\newcommand{\Pcan}{P_{\mathrm{can}}}

\makeatletter
\begin{document}
\title{Emergent centrality in rank-based supplanting process}
\author{Kenji\ Shimomura$^1$, Yasuhiro\ Ishitsuka$^2$, and Hiroki\ Ohta$^3$}

\affiliation{
$^1$Center for Gravitational Physics and Quantum Information,\\
Yukawa Institute for Theoretical Physics, Kyoto University, Kyoto, 606–8502, Japan
\\
$^2$Institute of Mathematics for Industry, Kyushu University, Fukuoka, 819–0395, Japan \\
$^3$Department of Human Sciences, Obihiro University of Agriculture and Veterinary Medicine, Hokkaido, 080-8555, Japan}
\date{\today}


\begin{abstract}
  We propose a stochastic process of interacting many agents, which is inspired by rank-based supplanting dynamics commonly observed in a group of Japanese macaques. 
  In order to characterize the breaking of permutation symmetry with respect to agents' rank in the stochastic process, we introduce a rank-dependent quantity, {\em overlap centrality}, which quantifies how often a given agent overlaps with the other agents.
  We give a sufficient condition in a wide class of the models such that overlap centrality shows perfect correlation in terms of the agents' rank in zero-supplanting limit.
  We also discuss a singularity of the correlation in the case of interaction induced by a Potts energy.
\end{abstract}

\maketitle

\section{Introduction}
One of promising candidates for going a step further in studying a many-body system is to construct a lattice model of interacting elements or agents which describes a many-body system. 
This strategy is not only applied to equilibrium systems \cite{Baxter} but also 
nonequilibrium systems \cite{Odor}. 
One of such attentions has been put on nonequilibrium lattice models such as driven lattice gas \cite{DLG},
ASEP \cite{ASEP}, ABC model \cite{ABC}, zero range process \cite{ZRP1,ZRP2}, etc., where emergent macroscopic property such as phase transition is one of the topics to be elucidated.
Recent development on active matter has focused on experimentally realizable systems such as colloidal or biological systems showing various phase transition  
such as flocking transition, lane formation, or motility-induced phase separation \cite{Vicsek,Lane1,Bacteria}. In the framework of statistical physics, it is of interest to look for an analytically tractable and minimal model for such phenomena \cite{V-Solon,Lane2,Hydro}.

Apart from such model-based studies, in the context of network theory,
the concept of centrality plays one of important roles in studying a given network induced by a many-body system consisting of inhomogeneous agents. 
Centrality has been particularly used in the literature of social network analysis to characterize 
which element on a network is the most influential. 
Depending on the purpose of network analysis, 
various measures of centrality such as degree centrality, closeness centrality, PageRank, eigenvector centrality, etc., \ have been proposed and found to be useful to characterize network structures \cite{Freeman, Bonacich, Barrat,Newman}.

As an example of inhomogeneous agents, 
primate species often live through interacting with members in a group \cite{primatebook}.
It has been reported in a primate species that the individuals, which we call agents, with high rank in social dominance tend to have high rank also in the eigenvector centrality of adjacency matrix for a graph composed from agents positions \cite{pecentrality1,pecentrality2}. 
In particular, Japanese macaques form a group living together and each agent in a group has its rank along the linear social dominance in the group,
leading to rank-dependent repulsion between two agents \cite{supplanting,Nakagawa}. 
This is so-called {\it supplanting} phenomenon which we mainly focus on.

In this paper, \rd{we focus on one of the intersections between lattice models and the network theory from the viewpoint of inhomogeneous agents.} We propose a new type of nonequilibrium lattice model, which 
is inspired by a supplanting phenomenon occurring between two agents in a group of Japanese macaques. 
The main objective of this paper is to show that a new type of macroscopic correlation appears when supplanting process, which is \rd{a rank-based} interaction with broken detailed balance, is added to an equilibrium system. It turns out that this problem can be mapped 
to computing a type of centrality, which we call {\em overlap centrality}, for a complete graph derived 
naturally from the correlations of agents' positions.

This paper consists of five sections. 
In Section \ref{Model}, 
we introduce a class of models breaking permutation symmetry that we study in this paper.
In Section \ref{Potts}, 
focusing on the case of the model where the interaction induced by the Potts energy is assumed, we provide a brief review of the equilibrium properties, introduce overlap centrality, and compute it by exact diagonalization of transition matrix. 
In Section \ref{Rig}, 
we provide the proof for the main result that  overlap centrality characterizing how often a given agent overlaps with the other agent shows perfect correlation with respect to the ranking of the given agent in zero-supplanting limit. 
This result holds rather generally, which is not limited to the case where the Potts energy is assumed as the source of an equilibrium interaction. Further, a conjecture about the existence of a singularity of the correlation is discussed for the case of the Potts energy.
In Section \ref{CR} as concluding remarks, 
we summarize the results and some subjects of future considerations.

\section{Model}\label{Model}
Let $N \ge 2$ be the number of agents, and $L \ge 3$ be the length of the one-dimensional lattice $X \coloneqq \Z/L\Z = \bkb{0,1,\ldots,L-1}$.
Let us denote by $i\in\{1,2,\cdots,N\}$ an agent, 
and by integer $x_i\in X$ the position of agent $i$.
We also regard the number $i$ identifying an agent as \textit{rank} of that agent. We say that 
rank $i$ is \textit{higher} (\textit{lower}) than $j$ if $i<j$ ($i>j$) ; for example, rank $i$ is higher than rank $i+1$.
Let us also write the collection of elements $a_i$ labelled by $1 \le i \le n$ as $(a_i)_{i=1}^{n}$ or the bold symbol $\bv{a}$. In particular we write a set of positions of agents as $\bvec{x} = (x_i)_{i=1}^{N}$. Hereafter, we call $\bv{x} = (x_i)_{i=1}^{N}$ a \textit{configuration}. We consider a hopping map $f_i^{\pm}$ such that $f_i^{\pm}\bvec{x}\coloneqq (x_j\pm\delta(i,j))_{j=1}^N$, where 
$\delta(i,j)$ is the Kronecker delta.
Note that the periodic boundary condition in terms of positions is automatically assumed by definition of $X$.

\subsection{Equilibrium dynamics}\label{sct:Equilbrium_dynamics}
We consider a general class of energy function 
$E(\bm{x}) = E(x_1, x_2, \dots, x_N)$
which is permutation symmetric in the following sense:
\begin{equation}\label{sym}
    E(x_1, x_2, \dots, x_N) = E(x_{\sigma(1)}, x_{\sigma(2)}, \dots, x_{\sigma(N)}),
\end{equation}
for any permutation $\sigma \in \mathfrak{S}_N$ of $N$ elements, where $\mathfrak{S}_N$ is the symmetric group of order $N$.
In addition, let $\beta$ be a parameter determining the magnitude of the energy including the sign $\pm$.

As an example of the models belonging to the above class, one can consider the following $L$-state  Potts energy $E(\bv{x})$ on the complete graph
where each agent connects with all agents \cite{Potts}:
\begin{align}\label{eq:PottsEnergy}
E(\bv{x}) = -\frac{2(L-1)\log (L-1)}{L-2}\dfrac{1}{2N}\sum_{i=1}^N\sum_{j=1}^N\delta(x_i,x_j).
\end{align}
This case means that an agent interacts with the other agents only if they have overlaps. In this sense, this Potts model on the complete graph is equivalent to the agents with an on-site interaction in one dimension, which may be a simple model to describe interacting agents. The coefficient of \eqref{eq:PottsEnergy} is adjusted so that the phase transition point $\beta_c$ in the equilibrium state is equal to 1, which we will discuss in more detail in Section \ref{Computation of partition function at equilibrium}.

Let us consider a Markov process with discrete time $t$, where during one time step between $t$ and $t+1$, only one of the following possible transitions may occur. The transition probability $T_0(\bv{x}\to f_i^{\pm}\bv{x})$ from each configuration $\bv{x}$  to the configuration $f_i^{\pm}\bv{x}$ for any agent $i$ is
\begin{align} \label{eq0:T0coeff}
T_0(\bv{x}\to f_i^{\pm}\bv{x}) =
\dfrac{1}{2N}\dfrac{1}{1+\exp \left(\beta \mathcal{D}_i^{\pm}E(\bv{x})\right) },
\end{align}
where 
\begin{align}
  &\mathcal{D}_i^{\pm}E(\bv{x}) := E(f_i^{\pm }\bv{x})-E(\bv{x}).
\end{align}

This leads to that the joint probability $P_t(\bv{x})$ of configuration $\bv{x}$ at time $t$ 
satisfies the following master equation:
\begin{align}
  P_{t+1}(\bv{x}) = &\sum_{\bv{x}'\neq\bv{x}}P_t(\bv{x}')T_0(\bv{x}'\to\bv{x})\nonumber\\
  &+ P_t(\bv{x}) \Big( 1-\sum_{\bv{x}'\neq \bv{x}}T_0(\bv{x}\to\bv{x}') \Big),
  \label{mas}
\end{align} 
where the summation over $\bv{x'}$ is done
for all of the possible configurations such that $T_0(\bv{x}\to\bv{x}')$ and $T_0(\bv{x}'\to\bv{x})$ are defined above.

The Gibbs distribution
\begin{align}
  P_{\mathrm{can}}(\bv{x}):=\dfrac{1}{Z_{N}(\beta)}\exp\left( -\beta E(\bv{x}) \right), \label{eq:GibbsDistr}
\end{align}
where $Z_{N}(\beta):=\sum_{\bv{x}}\exp(-\beta E(\bv{x}))$,
is the stationary solution $P_{\mathrm{st}}(\bv{x})$ of the master equation (\ref{mas}), satisfying
\begin{align}
 \sum_{\bv{x}'}P_{\mathrm{st}}(\bv{x}')T_0(\bv{x}'\to\bv{x})
=P_{\mathrm{st}}(\bv{x})\sum_{\bv{x}'\neq \bv{x}}T_0(\bv{x}\to\bv{x}'),
\end{align} 
because the Gibbs distribution satisfies the detailed balance condition:
\begin{align}
P_{\mathrm{can}}(\bv{x})T_0(\bv{x}\to\bv{x'})
=P_{\mathrm{can}}(\bv{x'})T_0(\bv{x'}\to\bv{x}),
\end{align} for any pair $\bv{x},\bv{x'}$ realized by the above dynamics.

\begin{figure}
\includegraphics[width=6cm,clip]{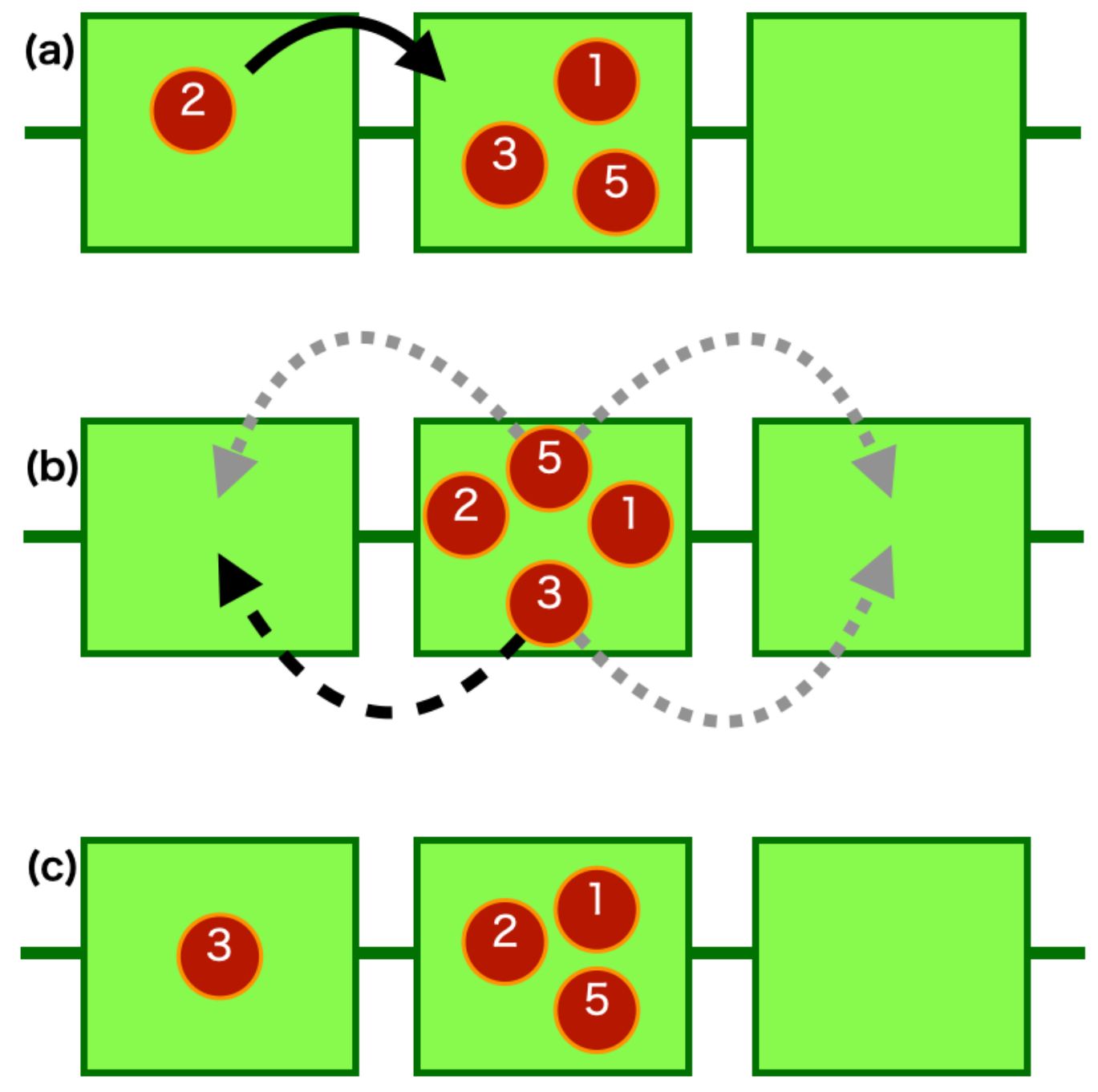}
\caption{(Color online) Schematic illustration of a transition step from a configuration described by (a) to a configuration described by (c) in the model. From (a) to (b), agent $2$ hops to the right site, and from (b) to (c), agent $3$ among two possibly supplanted agents, which are $3$ and $5$, is supplanted by agent $2$ and hops to the left site.  Four arrows in (b) means that, in the above transition, two agents $3$ and $5$ in $S(\bm{x}, 2, +)$ could be supplanted, and the direction of supplanting process could be either to the left or to the right.}
\label{pic}
\end{figure}

\subsection{Broken detailed balance by supplanting}\label{sct:BDB_by_supplanting}
Next, we consider to add a supplanting process to the equilibrium dynamics introduced above, which breaks the detailed balance condition.

Let us imagine a process in which an agent $i$ hops to a position $y_i=x_i\pm 1$ from position $x_i$ in accordance with the equilibrium transition probability $T_0$, and then an agent at $y_i$, which has a rank $j$ such that $i<j$, is stochastically forced to hop to the position $y_i +1$ for $d=+$, or $y_i -1$ for $d=-$.
In this process, an agent with a higher rank $i$ supplants another agent with a lower rank $j$. This is why we call such a process the \textit{supplanting process}. Note that the configuration $\bvec{x}$ turns to be the configuration $f_j^d f_i^\pm \bvec{x}$ through the whole supplanting process.
See Fig.\ \ref{pic} for a graphical reference of the process. 
For convenience, let us introduce the following set 
\begin{equation}\label{eq:DefOfS}
    S(\bm{x}, i, \dPM) \coloneqq \{ i < j \le N \mid
    x_j = x_i \dPM 1 \},
\end{equation} which is a set of every agent whose rank is lower than $i$ at position $x_i \dPM 1$. That is, those agents could be supplanted by the agent $i$ when the agent $i$ hops to position $x_i\pm 1$.

Suppose that every agent in $S(\bvec{x},i,\pm)$ has the same chance to be chosen as the supplanted one, and that the direction $d$ of hopping by supplanting is determined with equal probability $1/2$.
Explicitly, we define the transition probability from $\bv{x}$ to $f_j^{d}f_i^{\pm}\bv{x}$ for $d\in\{+,-\}$ and $j\in S(\bm{x}, i, \dPM)$:
\begin{align} \label{eq0:Tcoeff}
  T(\bv{x}\to f_j^{d}f_i^{\pm }\bv{x})=
  \dfrac{p}{2}\frac{1}{1+p\# S(\bm{x}, i, \dPM)}
  T_0(\bv{x}\to f_i^{\pm}\bv{x}),
\end{align} where $p\in \mathbb{R}_+\coloneqq[0,\infty)$ is the parameter of supplanting rate and $\#S$ is the number of all the elements of a set $S$.
When $p\to 0$, supplanting rarely occurs, and when $p\to\infty$, supplanting almost always occurs.
Note that at most one agent is supplanted in a single transition regardless of the value of $p$.

On the other hand, for the case of $j\notin S(\bm{x}, i, \dPM)$, it holds that
\begin{align}
  T(\bv{x}\to f_j^{d}f_i^{\pm}\bv{x})=0.
 \end{align}
Further, the probability of transition $T(\bv{x}\to f_i^{\pm}\bv{x})$ is modified from $T_0$ as
\begin{align} \label{eq:Tnonmove}
  T(\bv{x}\to f_i^{\pm}\bv{x}) = \frac{1}{1+p\# S(\bm{x}, i, \dPM)} T_0(\bv{x}\to f_i^{\pm}\bv{x}).
\end{align}
Totally, the following holds:
\begin{align}\label{eq:T0andT}
  &T_0(\bv{x}\to f_i^{\pm}\bv{x}) \nonumber \\
  &= T(\bv{x}\to f_i^{\pm}\bv{x})  +\sum_{\substack{j\in S(\bm{x}, i, \dPM) \\ d = \pm}}T(\bv{x}\to f_j^{d}f_i^{\pm}\bv{x}). 
\end{align}

Then, the master equation for the joint probability $P_t(\bv{x})$
governing the above stochastic process is as follows:
\begin{align}\label{eq:mas2}
   P_{t+1}(\bv{x})= &\sum_{\bv{x}'}P_t(\bv{x}')T(\bv{x}'\to\bv{x})\nonumber\\
   &+P_t(\bv{x})(1-\sum_{\bv{x}'\neq \bv{x}}T(\bv{x}\to\bv{x}')),
\end{align} where the summation over $\bv{x'}$ is done
for all of the possible configurations such that $T(\bv{x}'\to\bv{x})$ and 
$T(\bv{x}\to\bv{x}')$ are defined.
Since $T(\bm{x} \to f_i^{\pm} \bm{x})$ is positive for any $\bm{x}$ and $i$ with finite $\beta, p, N, L$, and $E$, any state can reach any state in this stochastic process; 
that is, the stochastic process defined above is an irreducible Markov process. 

In the presence of $p>0$, the supplanting process does not hold the detailed balance condition because of the asymmetric property in terms of agents rank; when a supplanting occurs in one step, i.e., an agent supplants another agent, the reverse process never occurs in any single step.
Thus, obviously the stationary solution $P_{\mathrm{st}}(\bv{x})$ is no longer the Gibbs distribution of a given energy function.
Note that in the limit of $p\to 0$, 
the detailed balance condition in terms of a 
given energy function is recovered.
Thus, we can also regard $p$ as the strength of violation of detailed balance condition.

\section{The case of the Potts energy}\label{Potts}
In this section, we focus on the 
case of the Potts energy defined by (\ref{eq:PottsEnergy}). 
We briefly review the known equilibrium properties and compute nonequilibrium stationary distribution by exact diagonalization of transition matrix corresponding to the master equation (\ref{eq:mas2}). Further, we introduce 
overlap centrality and its correlation
 with agents' rank, which are calculated using the computed stationary distribution.

\subsection{Computation of partition function at equilibrium}\label{Computation of partition function at equilibrium}
As a preliminary, we consider the equilibrium case with $p=0$.
The equilibrium ferromagnetic Potts model on the complete graph has two phases; one is the ordered phase for stronger interaction, and another is the disordered phase for weaker interaction, which are separated by a first-order transition point if $L\ge 3$ \cite{Potts}. 
In the context of this paper, 
the ordered state can be regarded as the condensate state of agents, that is, the state where all agents are located at the same position.

Let us look into more detail for the computation of the above results. When $p=0$ and interaction strength $\beta$ is positive ($\beta>0$), corresponding to the case of attractive interaction, 
the partition function $Z_N(\beta):=\sum_{\bv{x}}
\exp\big({-}\beta E(\bv{x}) \big)$ for $N\to\infty$ can be explicitly expressed, by which the equilibrium transition point is computed exactly.
Concretely, by performing the Stratonovich--Hubbard transformation \cite{mean-field_Potts2013, mean-field_Potts2020} with
\begin{align}
    &\exp\bkc{\frac{K\beta}{2N}\sum_{x\in X}\bka{\sum_{i=1}^N \delta(x_i,x)}^2}\notag\\
    &= \prod_{x\in X}\sqrt{\frac{NK\beta}{2\pi}}\int_\R\dif q\exp\bkc{-\frac{NK\beta}{2}q^2+K\beta q\sum_{i=1}^N\delta(x_i,x)},
\end{align}
we obtain 
\begin{align}
  &Z_N(\beta)= \bka{\frac{NK\beta}{2\pi}}^{L/2}\int_{\mathbb{R}^{L}}\dif^{L}q 
  \exp \left(-N\beta f_\beta(\bv{q}) \right),\\
  &\rd{f_\beta(\bv{q})=\frac{K}{2}\sum_{x\in X}q_{x}^{2}-\beta^{-1}\log\bka{\sum_{x\in X} \exp\bka{K\beta q_{x}}},}
\end{align}
where
\begin{align}
    K
    = \frac{2(L-1)\log(L-1)}{L-2}.
\end{align}
For $N\gg 1$, the minimal value of $f_\beta(\bvec{q})$ as a function of the order parameter $\bvec{q}$ behaves effectively as the free energy density of the Potts model as follows: 
\begin{align}
    \lim_{N\to\infty}\frac{-\beta^{-1}\log Z_N(\beta)}{N}
    = \min_{\bvec{q}\in\R^L}f_\beta(\bvec{q}).
\end{align}

Taking 
\rd{$\frac{\partial f_\beta}{\partial q_{x}}(\bvec{q})=0$} to minimize $f_\beta(\bvec{q})$, 
we obtain the stationary condition 
\rd{\begin{align}\label{eq:stationary_cond2}
    q_y\exp(-K\beta q_y)
    = \Big(\sum_{x\in X}\exp(K\beta q_x)\Big)^{-1}
\end{align}
for each $y \in X$.} 
From \eqref{eq:stationary_cond2}, we see that
\rd{\begin{align}\label{eq:ab_constraint}
    q_x\exp(-K\beta q_x)
    =q_y \exp(-K\beta q_y),
\end{align}
for any $x, y \in X$.
This indicates that there is a constant $c$ such that $q_x \exp(-K \beta q_x) = c$,
which has at most two real solutions $a, b$ as an equation for $q_x$.
Then, we have $q_x \in \{a, b\}$ for each $x \in X$.}
From \eqref{eq:stationary_cond2}, we also have
\rd{\begin{align}\label{eq:q_constraint}
    \sum_{x\in X}q_x = 1.
\end{align}}

Thus, a necessary condition for the order parameter $\bvec{q}$ to minimize $f_\beta(\bvec{q})$ is described below. Keeping with \eqref{eq:ab_constraint} and \eqref{eq:q_constraint}, one of the following conditions (i) and (ii) is satisfied:
\begin{enumerate}[(i)] 
\item It holds that $\bm{q} = \tilde{\bvec{q}}^{(0)} \coloneqq \dfrac{1}{L} (1, 1, \dots, 1).$
\item There exist an integer $n\in\bkb{1,\dots,L-1}$ and two distinct real numbers $a_n = a_n(\beta), b_n = b_n(\beta)$ satisfying 
\begin{gather} \label{eq:eq_of_ab}
    a_n \exp(-K\beta a_n) = b_n \exp(-K\beta b_n), \\
    na_n+(L-n)b_n = 1. \label{eq:eq_of_ab2}
\end{gather}
Moreover, $n$ components of $\bvec{q}$ are $a_n$ and remaining $(L-n)$ components of $\bvec{q}$ are $b_n$. 
\end{enumerate}
For example, if $n=1$, the solutions are described by \rd{$\bvec{q} = \tilde{\bvec{q}}^{(\alpha)}$ for $1 \le \alpha \le L$,} where
\rd{\begin{align}
    \tilde{q}^{(\alpha)}_x \coloneqq \begin{dcases}
        a_1(\beta) & (\text{if $x \equiv \alpha \mod L$}) \\
        b_1(\beta) & (\text{if $x \not\equiv \alpha \mod L$})
    \end{dcases}.
\end{align}}

In Ref.\ \cite{Ellis}, it has been shown that the set of the global minimum points $\bvec{q}$ of $f_\beta(\bvec{q})$ corresponds to the case of $\bvec{q} = \tilde{\bvec{q}}^{(0)}$ or $n=1$, depending on $\beta$. Concretely, the set is described as
\begin{align}
    \begin{dcases}
        \bkb{\tilde{\bvec{q}}^{(1)}(\beta),\tilde{\bvec{q}}^{(2)}(\beta),\ldots,\tilde{\bvec{q}}^{(L)}(\beta)} & (\text{if}\ 0<\beta<1) \\
        \bkb{\tilde{\bvec{q}}^{(0)},\tilde{\bvec{q}}^{(1)}(1),\ldots,\tilde{\bvec{q}}^{(L)}(1)} & (\text{if}\ \beta=1) \\
        \bkb{\tilde{\bvec{q}}^{(0)}} & (\text{if}\ \beta>1). \\
    \end{dcases}
\end{align}
Note that the equations \eqref{eq:eq_of_ab} and \eqref{eq:eq_of_ab2} with $n=1$ determine the value $a_1 \neq b_1$ uniquely, and the resulting functions $a_1(\beta), b_1(\beta)$ are differentiable in the region $0 < \beta < 1$.

The expectation value of energy density is also expressed as
\begin{align}\label{eq:energy_density}
    \lim_{N\to\infty}\frac{\braket{E}_{\text{can}}}{N}
    = 
    \begin{dcases}
        \frac{\partial}{\partial\beta}\beta f_\beta(\tilde{\bvec{q}}^{(1)}(\beta)) & (\text{if}\ 0<\beta<1) \\
        \frac{\partial}{\partial\beta}\beta f_\beta(\tilde{\bvec{q}}^{(0)}) & (\text{if}\ \beta>1).
    \end{dcases}
\end{align}
Thus, one can show that the energy density \eqref{eq:energy_density} exhibits a discontinuous jump at $\beta=1$, which is the phase transition point of the Potts model.

\subsection{State vector description}
Let us move onto the model with general $p \ge 0$. In this case, we need to explicitly consider the dynamics
in order to compute the stationary distribution of the model.
We would like to describe
the stochastic process by transition matrices with some basic linear operators.
For more detailed description and derivation,
see Appendix \ref{sct:TransMatr}.

Let $H_X$ be the one-agent state space,
\rd{which is of dimension $L$.}
It is considered as a complex vector space with inner product $\braket{\cdot | \cdot}$,
and has an orthonormal basis $\{ |x \rangle \mid x \in X \}$ over
the set of complex numbers $\mathbb{C}$.
Then the $N$-times self-tensored space $H_X^{\otimes N}$ can be 
identified to the $N$-agent state spaces,
\rd{which is of dimension $L^N$.}
For a configuration of agents $\bm{x} = (x_1, x_2, \dots, x_N) \in X^N$,
the corresponding state vector is $|\bm{x} \rangle = |x_1\rangle \otimes
|x_2 \rangle \otimes \dots \otimes |x_N\rangle$.
The space $H_X^{\otimes N}$ has a natural inner product induced by $\langle \cdot | \cdot \rangle$, and the set of the state vectors $\{|\bvec{x} \rangle\}_{\bvec{x} \in X^N}$ is an orthonormal basis. We use the same symbol $\langle \cdot | \cdot \rangle$ to write the inner product on $H_X^{\otimes N}$.

We associate the probability $P(\bvec{x})$ for agents' configuration $\bvec{x}$ with a state $\ket{P}\in H_X^{\otimes N}$ as follows:
\alit{
    \braket{\bvec{x}|P}
    = P(\bvec{x}),
}
or
\alit{
    \ket{P}
    = \sum_{\bvec{x}\in X^N}P(\bvec{x})\ket{\bvec{x}}.
}
For a given state $\ket{P_t}$ at time $t$, the time evolution of the state is described as follows:
\begin{align}
    \ket{P_{t+1}}=\TransOp{}{}\ket{P_t}, \label{eq:Transition}
\end{align}
where $\TransOp{}{}$ is a transition matrix on $H_X^{\otimes N}$ such that \eqref{eq:Transition}
is equivalent to the master equation \eqref{eq:mas2} for the joint probability. 
Similarly, $\TransOp{0}{}$ is the transition matrix $\TransOp{}{}$ when $p=0$.

We introduce some basic operators.
A hopping map $f_i^\pm$ to the right (resp.\ the left) corresponds to
the operator defined by $\ShiftOp{i}{\pm}$;
for $\bm{x} = (x_1, x_2, \dots, x_N) \in X^N$, 
\begin{align}
    \ShiftOp{i}{\dP} | \bm{x} \rangle
    &= |\Shift{i}{\dP} \bm{x} \rangle \notag \\
    &= \ket{x_1} \otimes \ket{x_2} \otimes \dots \otimes \ket{x_{i} + 1} \otimes \dots \otimes \ket{x_N}, \\
    \ShiftOp{i}{\dM} | \bm{x} \rangle
    &= |\Shift{i}{\dM} \bm{x} \rangle \notag \\
    &= \ket{x_1} \otimes \ket{x_2} \otimes \dots \otimes \ket{x_{i} - 1} \otimes \dots \otimes \ket{x_N}.
\end{align}
Next we define projection operators.
For a site $y \in X$ and an agent $1 \le i \le N$, 
we define $\CheckOp{i}{y}$ by
\begin{equation}
    \CheckOp{i}{y} |\bm{x}\rangle =
    \begin{dcases}
        |\bm{x} \rangle & \text{(if $x_i = y$)} \\
        0 & \text{(if $x_i \neq y$)},
    \end{dcases}
\end{equation}
for $\bm{x} = (x_1, x_2, \dots, x_N) \in X^N$.
Then for a configuration $\bm{y} \in X^N$,
we define the operator $\CheckOp{}{\bm{y}}$ as 
$\prod_{1 \le i \le N} \CheckOp{i}{y_i}$.
Finally, we denote $\id_H$ (resp.\ $\id_H^{\otimes N}$) as
the identity operator on $H_X$ (resp.\ $H_X^{\otimes N}$).

Using these notions, we can describe the transition matrices 
$\TransOp{0}{}$ and $\TransOp{}{}$. First, $\TransOp{0}{}$ is
\begin{align}
 \TransOp{0}{} &= \sum_{\substack{1 \le i \le N \\ d = \pm}}
    \sum_{\bm{x} \in X^N}
    \Big( T_0(\bm{x} \to f_i^{d} \bm{x}) \ShiftOp{i}{d} \nonumber \\
    &\qquad + \left( \frac{1}{2N}-T_0(\bm{x} \to \Shift{i}{d} \bm{x}) \right) \id_H^{\otimes N} \Big)
    \CheckOp{}{\bm{x}}.
\end{align}
Then $\TransOp{}{}$ is written as
\begin{widetext}
\begin{align}
    \TransOp{}{} &= 
        \sum_{\substack{1 \le i \le N \\ d = \pm}}
        \sum_{\bm{x} \in X^N}
         \Bigg( \sum_{\substack{j \in S(\bm{x}, i, d) \\ d' = \pm}}
        T(\bm{x} \to \Shift{j}{d'}\Shift{i}{d} \bm{x}) \ShiftOp{j}{d'}\ShiftOp{i}{d}
        + T(\bm{x} \to \Shift{i}{d} \bm{x}) 
        \ShiftOp{i}{d} + \Big(\frac{1}{2N}-T_0(\bm{x} \to \Shift{i}{d} \bm{x})\Big) \id_H^{\otimes N} \Bigg)
        \CheckOp{}{\bm{x}}
    \label{eq0:Trep} \\
    &= \TransOp{0}{} + \sum_{\substack{1 \le i \le N \\ d = \pm}}
    \sum_{\bm{x} \in X^N}
    \sum_{\substack{j \in S(\bm{x}, i, d) \\ d' = \pm}}
    T(\bm{x} \to \Shift{j}{d'}\Shift{i}{d} \bm{x}) (\ShiftOp{j}{d'}
    - \id_H^{\otimes N}) \ShiftOp{i}{d} \CheckOp{}{\bm{x}},
    \label{eq1:Trep}
\end{align}
\end{widetext}
where we used \eqref{eq:T0andT}.
For the definition of coefficients, 
see \eqref{eq0:T0coeff}, \eqref{eq0:Tcoeff}, 
and \eqref{eq:Tnonmove}.

On this settings, the transition matrix $\TransOp{}{}=\TransOp{}{}(\beta,p)$ is naturally regarded as a linear operator on $H_X^{\otimes N}$. Let $\ket{P(\beta,p)}$ be the unique stationary state of $\TransOp{}{}(\beta,p)$ satisfying
\alit{\label{eq:stationary state for T}
    \TransOp{}{}(\beta,p)\ket{P(\beta,p)}
    = \ket{P(\beta,p)}.
}
One can show that, since $\TransOp{}{}$ is irreducible, $\ket{P(\beta,p)}$ exists and is uniquely determined by Perron--Frobenius theorem.

For the latter discussion, let us consider the symmetry of the transition matrices $\TransOp{0}{}$ and $\TransOp{}{}$.
We introduce permutation operators
$\PermOp{\sigma}$ on $H_X^{\otimes N}$. 
For a given element $\sigma \in \mathfrak{S}_N$ of 
the symmetry group $\mathfrak{S}_N$ of agents,
we define
\begin{equation}\label{eq:PermOpDefinition}
    \PermOp{\sigma}\ket{\bvec{x}}
    \coloneqq \ket{\sigma^{-1}(\bvec{x})},
\end{equation}
where $\sigma^{-1}(\bvec{x})\coloneqq\bka{x_{\sigma^{-1}(j)}}_{j=1}^N$.
Then, we have
\begin{gather}\label{eq:commuT0}
    \PermOp{\sigma}^\dag\TransOp{0}{}\PermOp{\sigma}
    = \TransOp{0}{},\\
    \label{noncommuT}
    \PermOp{\sigma}^\dag\TransOp{}{}\PermOp{\sigma}
    \neq \TransOp{}{},
\end{gather}
where $\PermOp{\sigma}^\dag$ is the Hermitian conjugate of $\PermOp{\sigma}$ (see \eqref{eq:TotalT0Permutation} and \eqref{eq:NotCommutesTransOp}). Note that $\PermOp{\sigma}$ is unitary: $\PermOp{\sigma}^\dag=\PermOp{\sigma}^{-1}=\PermOp{\sigma^{-1}}$.
In the sense of relation \eqref{eq:commuT0}, the equilibrium dynamics described by $\TransOp{0}{}$ holds permutation symmetry. 
In contrast, the whole dynamics by $\TransOp{}{}$ breaks the permutation symmetry, as described by \eqref{noncommuT}.

\subsection{Exact diagonalization of transition matrices}
\begin{figure}
  \includegraphics[width=8cm,clip]{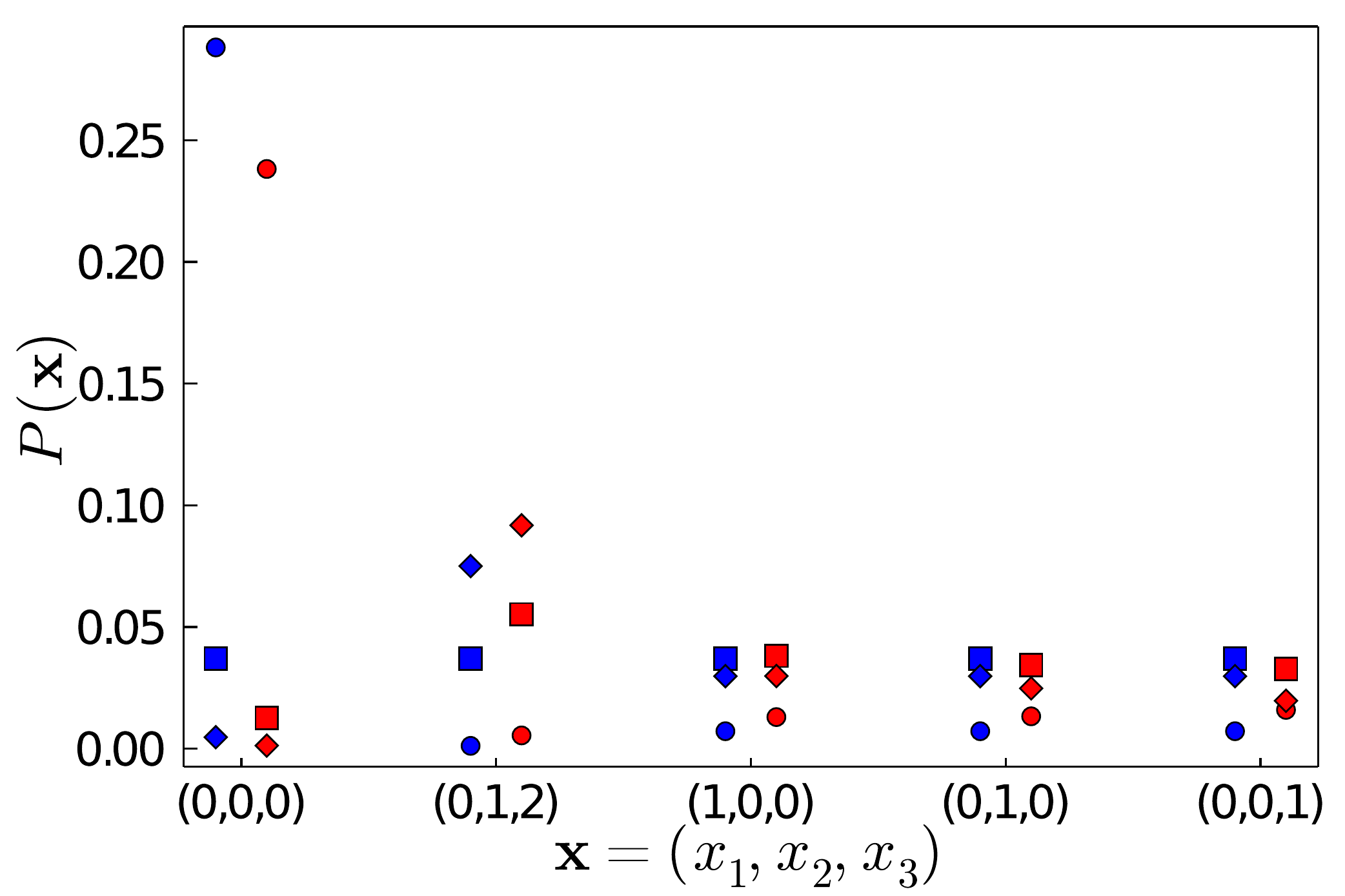}
\caption{(Color online) Probabilities of the configurations determined by stationary distribution for $\beta=2$ (circles), $\beta=0$ (rectangles), and $\beta=-1$ (diamonds); $p=0$ (blue or the left at each column) and $p=10$ (red or the right); $N=L=3$. The state of $(0,0,0)$ means that three agents are located at the same site. 
The state of $(0,1,2)$ means that 
each of three agents is located at the different position, respectively. 
The other three states mean that two agents are located at the same site and another agent is located at one of the different sites.}
\label{ff}
\end{figure}

We perform exact diagonalization for the transition matrix $\TransOp{}{}$
to obtain the eigenvalues and their corresponding eigenvectors.
Thus, the stationary distribution corresponds to the eigenvector with the maximum real part, which is 1, of the eigenvalue. Note that the number of the states is $L^N$,
which gets exponentially large as a function of $N$.

As shown in Fig.\ \ref{ff} with $p=0$, 
at $\beta=0$, the joint probability of each 
configuration shows the same value. 
As $\beta$ is increased from $0$, 
the joint probability of condensate configuration $(0,0,0)$ is much higher 
than that of the other configurations. 
Conversely, as $\beta$ is decreased from $0$, the joint probability of the configuration $(0,1,2)$, where all the agents are separated, is much higher than that of the other configurations. 
On the other hand, the joint probability of configuration with a pair of two agents overlapping at the same site and the other located at the different site such as $(1,0,0), (0,1,0), (0,0,1)$ does not depend on the pair for $p=0$. 

When $p=10$, the joint probabilities of configurations $(1,0,0)$, $(0,1,0)$, $(0,0,1)$ are distinct. 
Concretely, those probabilities with $\beta=2$ and $\beta=-1$ increase and decrease, respectively, in the order of $(1,0,0)$, $(0,1,0)$, and $(0,0,1)$. This means that higher-ranked agents (resp.\ lower-ranked agents) tend to overlap more frequently for $\beta=2$ (resp.\ for $\beta=-1$).
This can be interpreted as a typical consequence of supplanting process.

In order to discuss how the configuration is condensed, let us introduce the normalized expectation value of the Potts energy in terms of a probability distribution $P(\bvec{x})$ as follows:
\begin{align}\label{eq:expectation_value_of_Potts_energy}
   M:= \frac{1}{N^2}
   \sum_{\bv{x}}\sum_{i=1}^N\sum_{j=1}^N\delta(x_i,x_j)P(\bv{x}).
\end{align} 
Note that by definition, $M$ takes $1$ as the maximum value 
in the case of $P(\bvec{x})=\prod_{k=1}^N\delta(x,x_k)$. 
In Fig.\ \ref{order}, using the computation of the stationary distribution by the exact diagonalization, $M$ is shown as a function of $p$ and $\beta$.

Relating to Section \ref{Computation of partition function at equilibrium},
in the case of the equilibrium distribution corresponding to the case with $p=0$, $M$ is rewritten by using $\braket{E}_{\text{can}}$ as
\begin{align}
    M
    = -\frac{2}{NK}\sum_{\bv{x}}E(\bvec{x})\frac{e^{-\beta E(\bv{x})}}{Z_N(\beta)}
    = -\frac{2}{K}\frac{\braket{E}_{\text{can}}}{N},
\end{align}
which means that $M$ is also discontinuous at the equilibrium phase transition point $\beta=1$ in the thermodynamic limit $N\to\infty$.

\begin{figure}
  \includegraphics[width=8cm,clip]{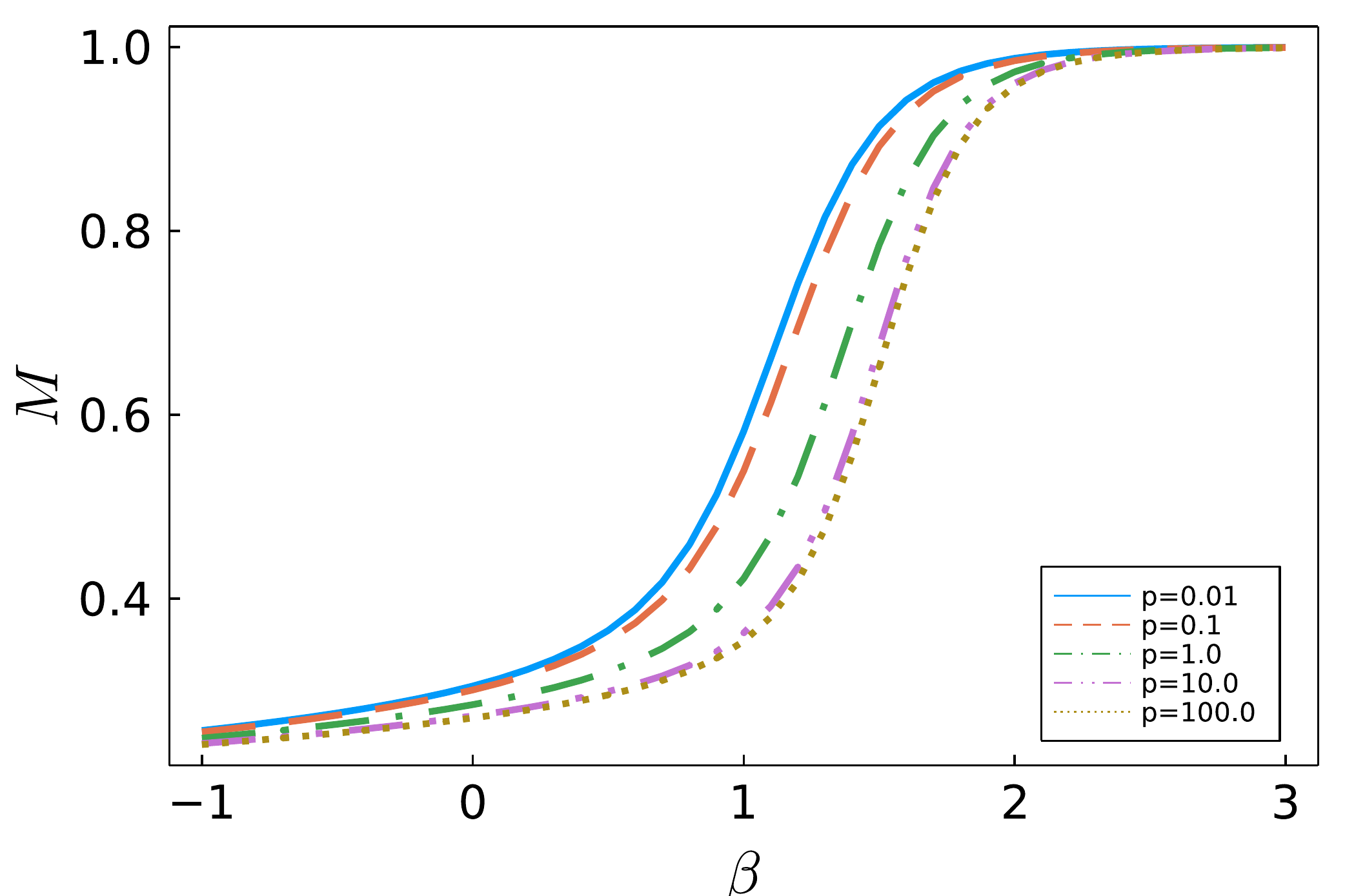}
  \caption{(Color online) The $\beta$-dependence of the normalized expectation value of the energy $M=\dfrac{1}{N^2}\sum_{x}\sum_{i,j}\delta(x_i,x_j)\braket{\bvec{x}|P(\beta,p)}$ 
  for various values of $p$ with $L=6, N=6$.}
\label{order}
\end{figure}

\subsection{Overlap centrality and its correlation coefficient}

\begin{figure}
  \includegraphics[width=7cm,clip]{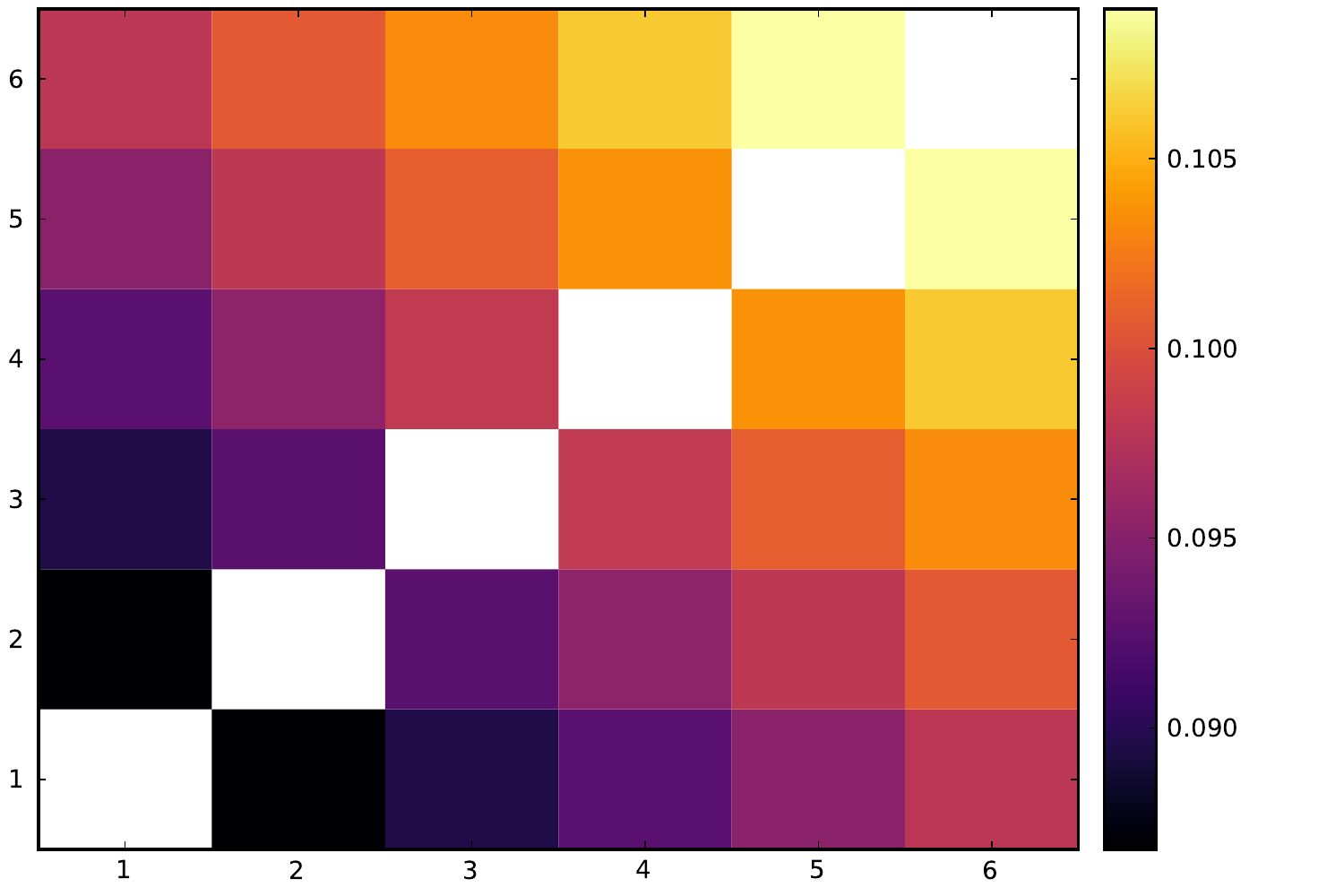}
  \includegraphics[width=7cm,clip]{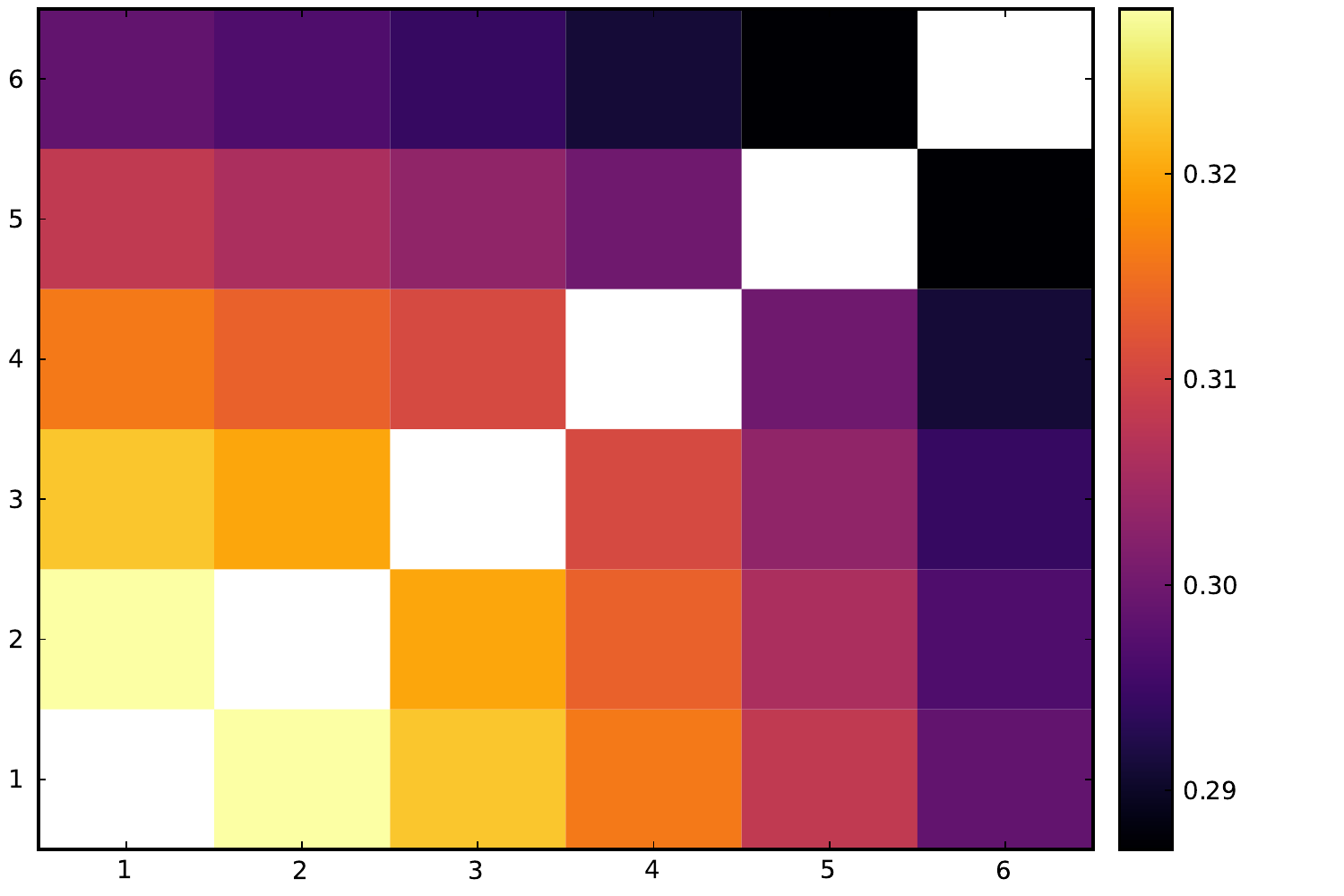}
\caption{(Color online) Heatmap of neighbor matrix $\mathcal{R}$ determined by the stationary distribution with the diagonal components left out.
The color corresponds to $r_{ij}$ for pair of two agents $(i,j)$.  Parameters: (top) $\beta=-1, p=1, L=6, N=6$.
  (bottom) $\beta=1, p=1, L=6, N=6$.}
\label{nei}
\end{figure}

\begin{figure*}
\begin{tabular}{cc}
\begin{minipage}[t]{0.45\linewidth}
    \includegraphics[width=7cm,clip]{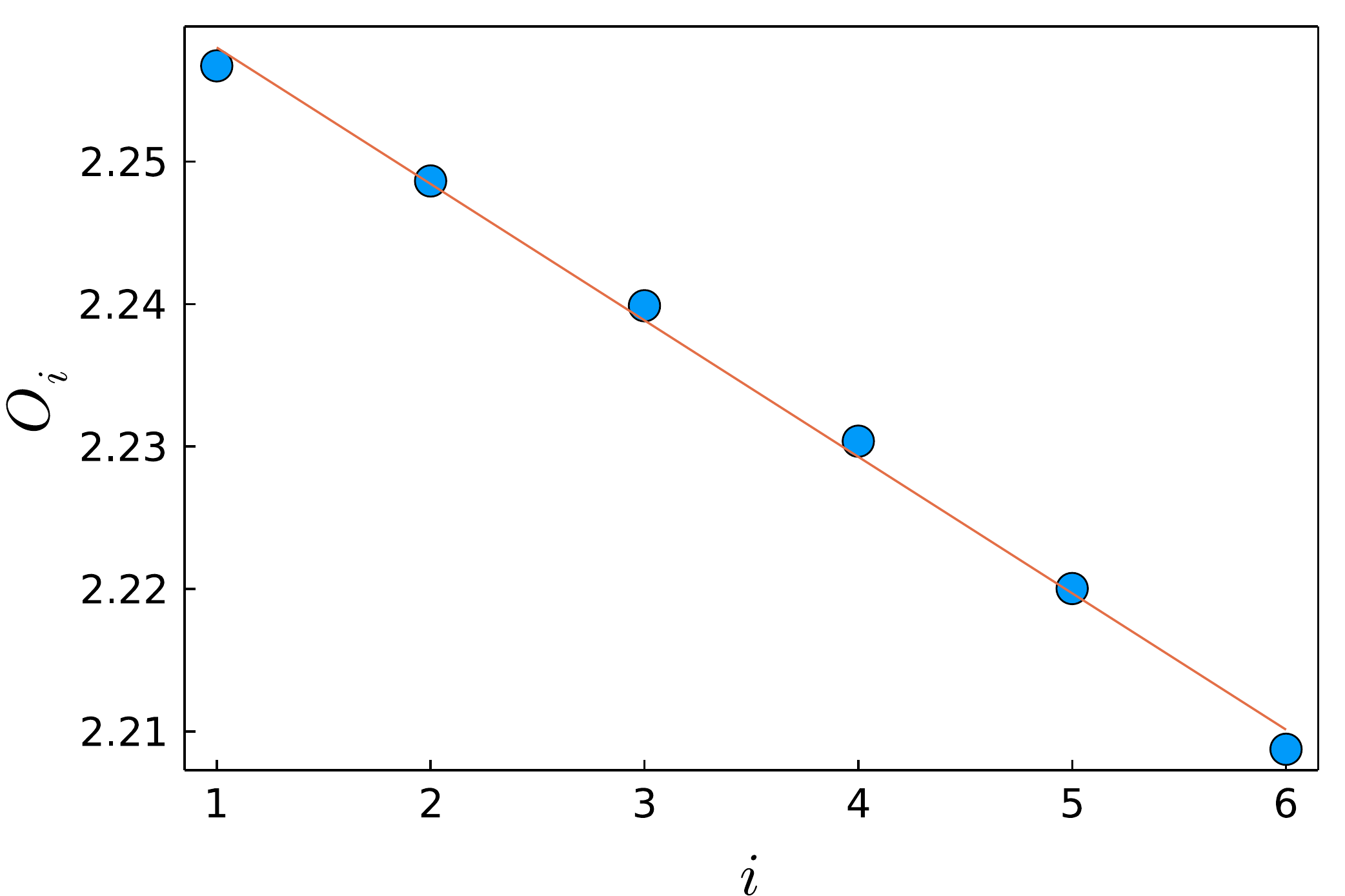}
    \includegraphics[width=7cm,clip]{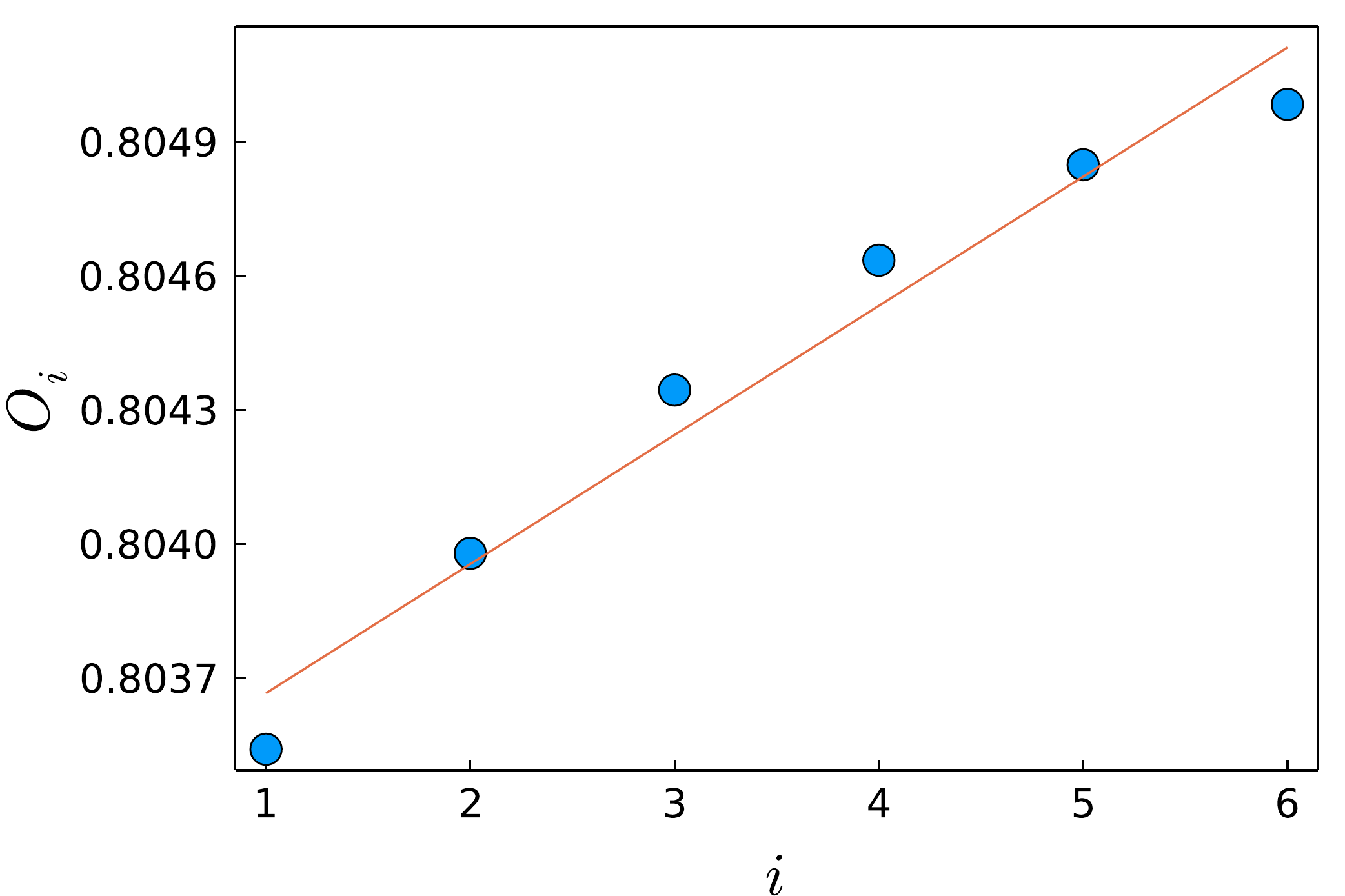}
    \includegraphics[width=7cm,clip]{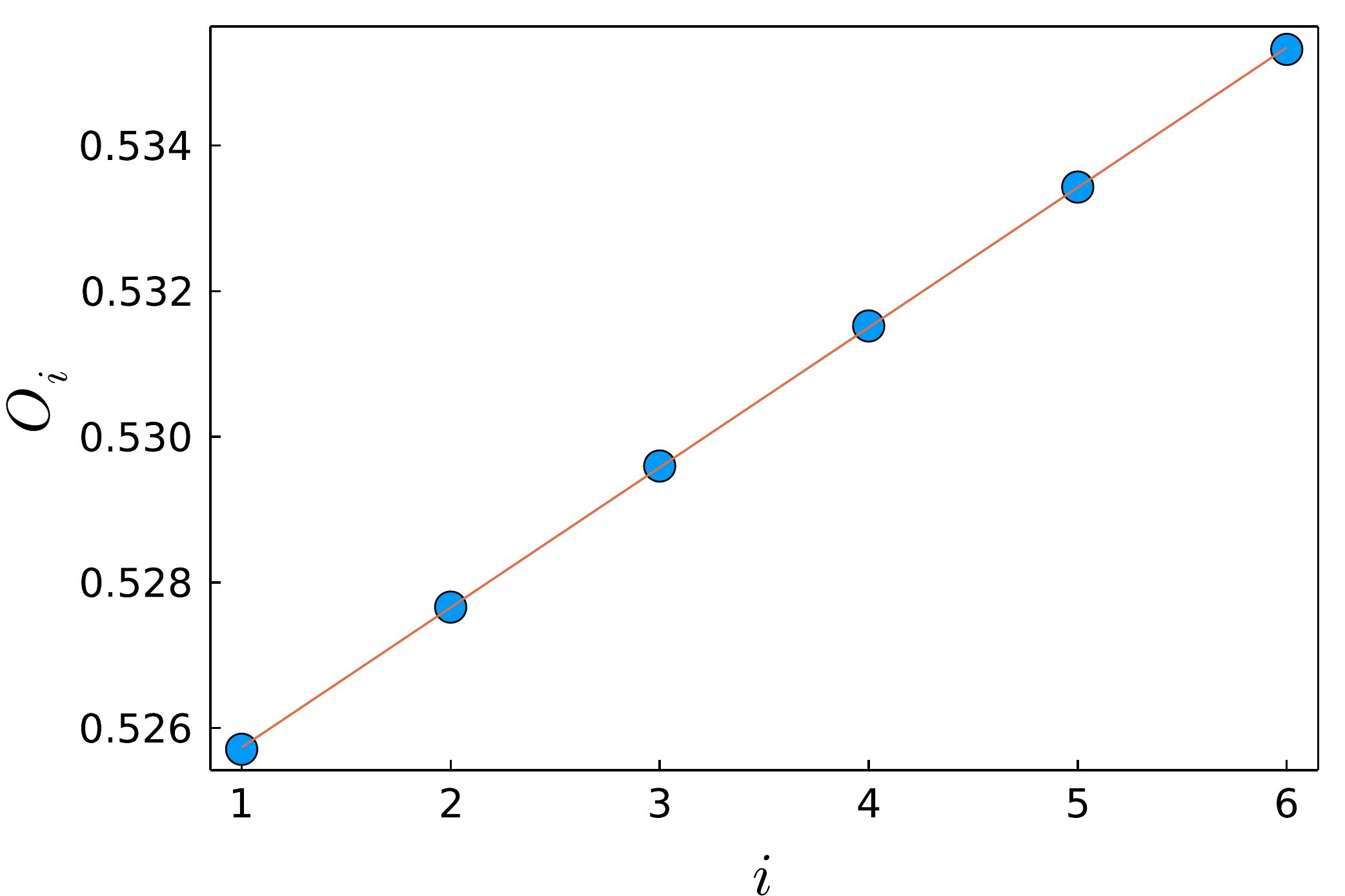}
\end{minipage}
\begin{minipage}[t]{0.45\linewidth}
    \includegraphics[width=7cm,clip]{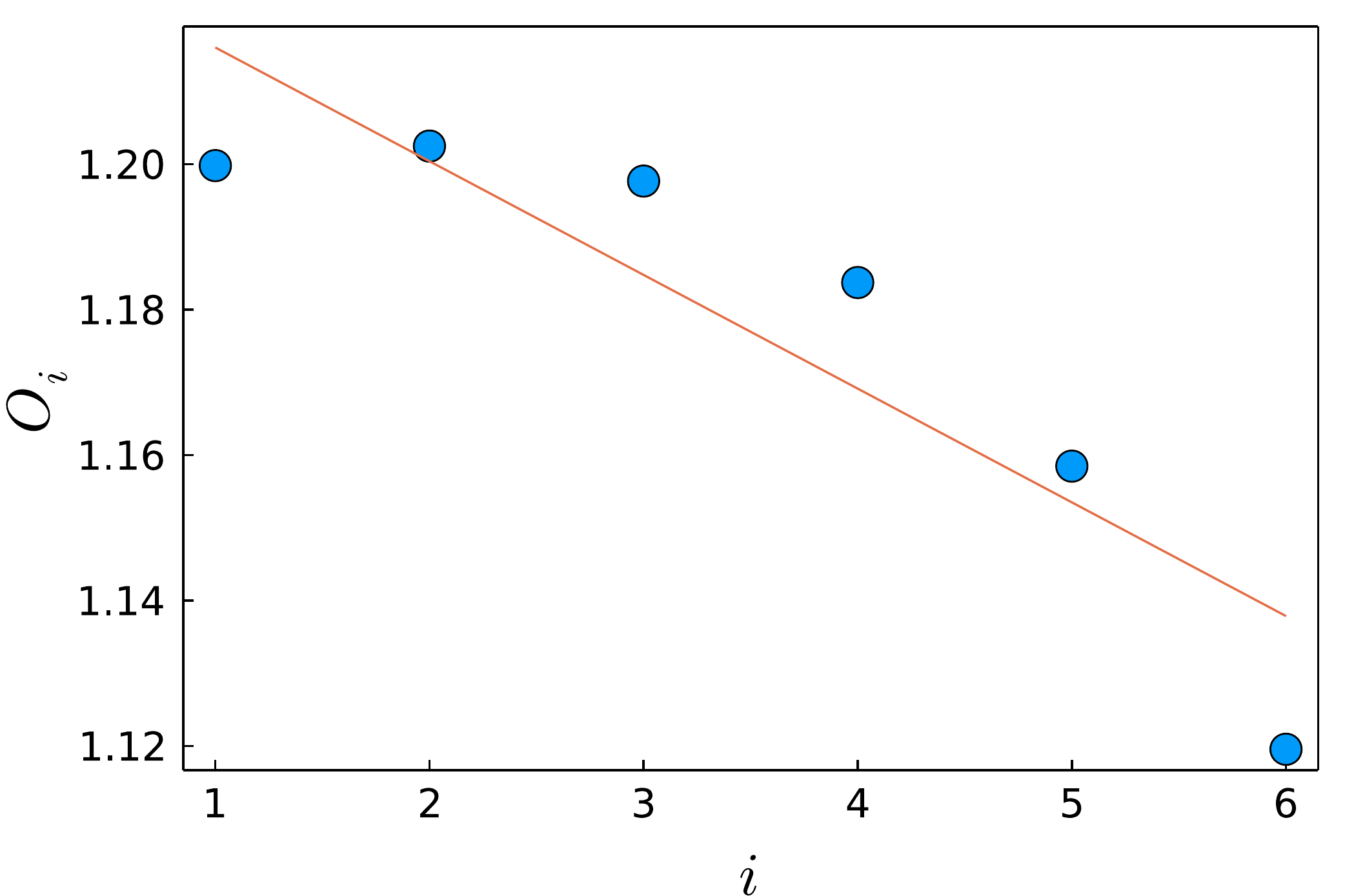}
    \includegraphics[width=7cm,clip]{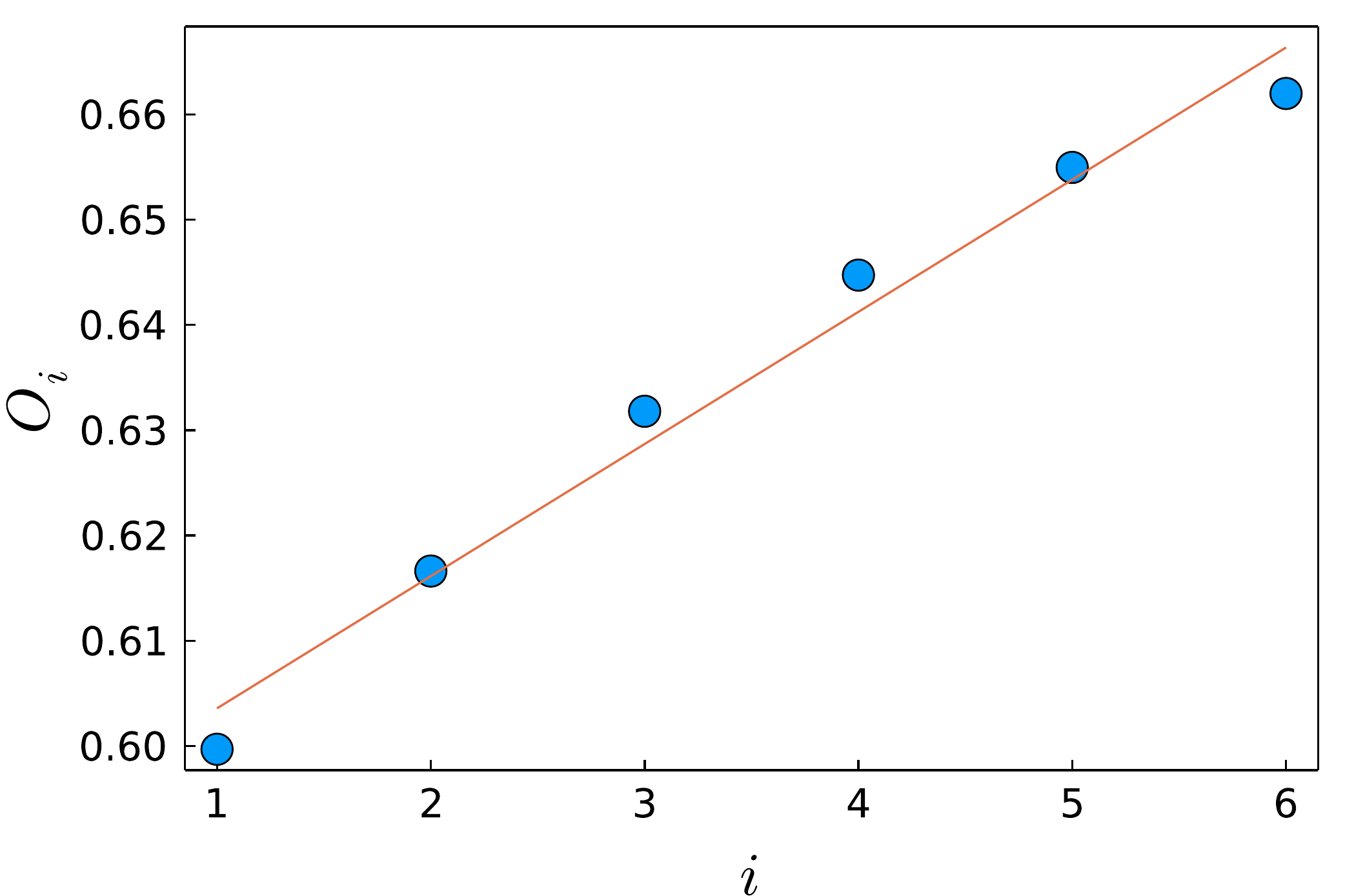}
    \includegraphics[width=7cm,clip]{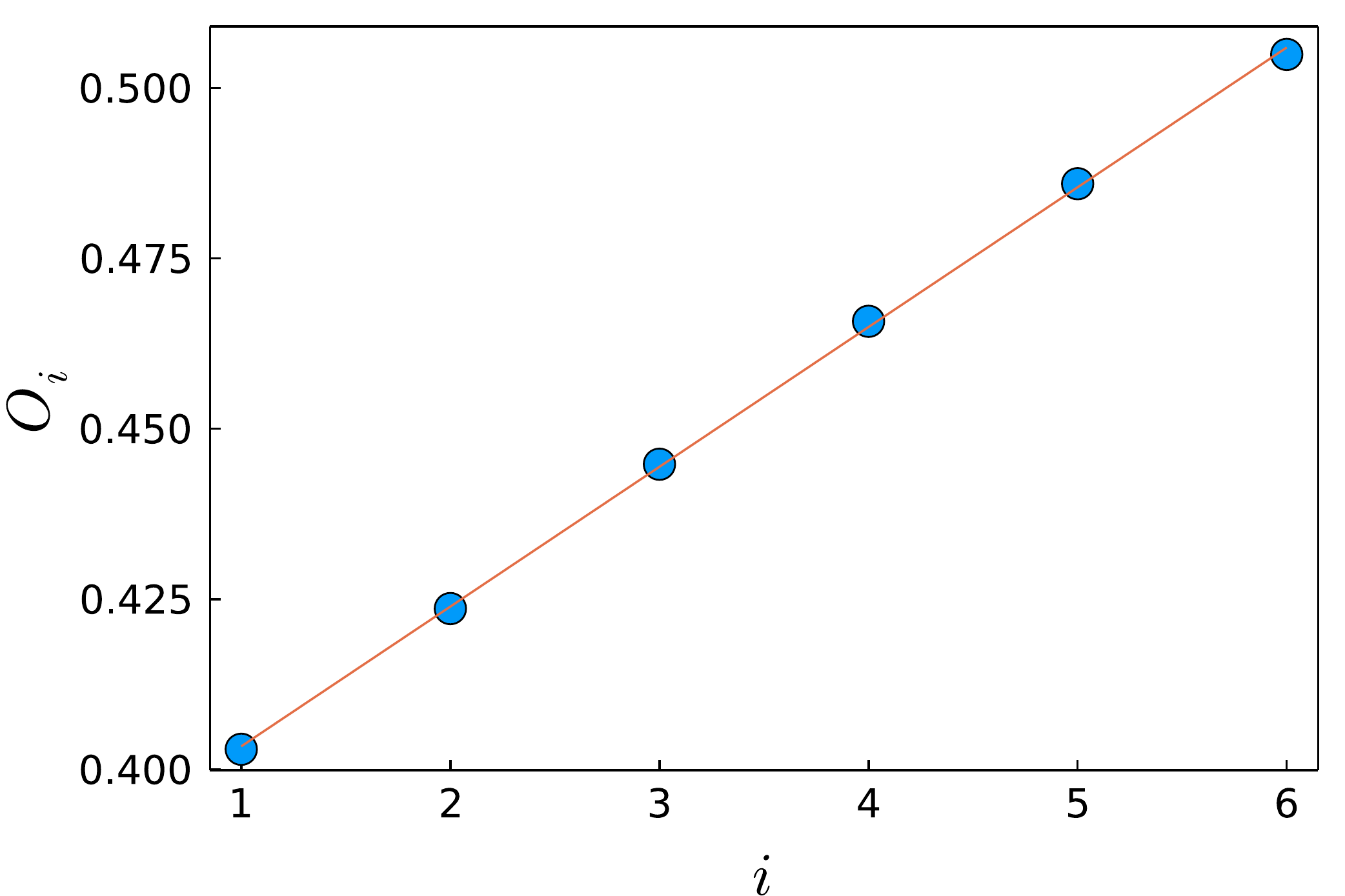}
\end{minipage}
\end{tabular}
    \caption{\rd{(Color online) The overlap centrality $O_i$ determined by the stationary distribution obtained by the exact diagonalization as a function of agent $i$ for $L=N=6$. The solid lines are linear regressions for $O_i$ and $i$. 
    \rd{The parameters are set as follows; $p=0.1$ (left) and $p=10.0$ (right); $\beta=1$ (top), $\beta=0$ (center), and $\beta=-1$ (bottom), respectively.}
    For negative $\beta$, $\phi$ is close to $1$ regardless of the value of $p$.}}
\label{neigen}
\end{figure*}

In order to characterize the correlation among agents, 
we may consider the neighbor matrix $\mathcal{R}:= (r_{ij})_{i,j}$ defined as 
\begin{align}\label{eq:definition_of_neighbor_matrix}
r_{ij}:= \sum_{\bv{x}}\delta(x_i,x_j)P(\bv{x}).
\end{align}
For example, if $P(\bvec{x})$ is the uniform distribution then $r_{ij}=1/L$, and if $P(\bvec{x})=\prod_{k=1}^N\delta(x,x_k)$ then $r_{ij}=1$. The latter gives the maximum value of $r_{ij}$. Note that $r_{jj}=1$ for any $j$.

The entry $r_{ij}$ means how often agent $i$ and $j$ are located at the same site under the distribution $P(\bvec{x})$.
In Fig.\ \ref{nei}, we show heatmaps of neighbor matrices computed from the stationary distribution.
It demonstrates that at $\beta=1$, the agents with higher rank have more overlaps with the other agents, and conversely at $\beta=-1$, the agents with lower rank have more overlaps with the other agents.

In order to quantify how often a given agent overlaps with the other agents in total, we introduce the overlap centrality as a function of rank $i$ using the entries of the neighbor matrix:
\begin{align}\label{eq:definition_of_overlap_centrality}
  O_i:=\sum_{\substack{1\le j\le N\\j\neq i}}r_{ij}.
\end{align} 
That is, we regard the agent $i$ having larger value of the overlap centrality 
as more influential one compared to the other agents having lower values of the overlap centrality. 
When the probability distribution $P(\bvec{x})$ is permutation symmetric, i.e., $P(\sigma(\bvec{x}))=P(\bvec{x})$ for any $\sigma\in\mathfrak{S}_N$, the overlap centrality does not depend on rank $i$.
Note that 
\begin{align}\label{MOrel}
M=\displaystyle \dfrac{1}{N^2}\sum_{i=1}^N O_i+\dfrac{1}{N}
\end{align} holds by the definition.
As shown in Fig.\ \ref{neigen}, we compute the overlap centrality $O_i$
computed from the stationary distribution, showing that
the overlap centrality has a plus slope at attractive interaction of $\beta=1$ and has a minus slope at repulsive interaction of $\beta=-1$.

In order to quantify the class of the overlap centrality in terms of the slope,
we measure the correlation coefficient $\phi$ of the overlap centrality with respect to agents' rank. This is defined as
\begin{align}\label{eq:correlation_coefficient}
    \phi:=\frac{1}{N}\sum_{i=1}^N\dfrac{(O_i-\frac{1}{N}\sum_{j=1}^NO_j)(i-\frac{1}{N}\sum_{j=1}^Nj)}{s_Os_I},
\end{align} where we define
\begin{align}
    s_O^2 &:= \frac{1}{N}\sum_{i}(O_i-\frac{1}{N}\sum_{j=1}^NO_j)^2, \\
    s_I^2 &:= \frac{1}{N}\sum_{i}(i-\frac{1}{N}\sum_{j=1}^Nj)^2
    = \frac{(N-1)(N+1)}{12}.
\end{align}
By definition, when $|\phi|=1$, 
$O_i$ is a linear function of $i$. 
We check the condition when this quantity $\phi$ is not defined. Since we set $N \ge 2$, the denominator of $\phi$ is zero exactly when $s_O^2 = 0$.  This corresponds to the case when $O_i$ is constant as a function of $i$.
In this case, we say that the quantity $\phi$ is {\em singular}.

In Fig.\ \ref{cc}, we show the $\beta$-dependence of the correlation coefficient $\phi$.
In the weak supplanting condition with $p\ll 1$ such as $p=0.1$ or $p=0.01$,
$\phi$ is close to $+1$ for negative $\beta$. As $\beta$ increases, $\phi$ sharply changes its sign around $\beta=0$, and turns out to be $-1$ for positive $\beta$.
We will discuss this behavior in a more general setting in Section \ref{Rig}.

\begin{figure}
  \includegraphics[width=7cm,clip]{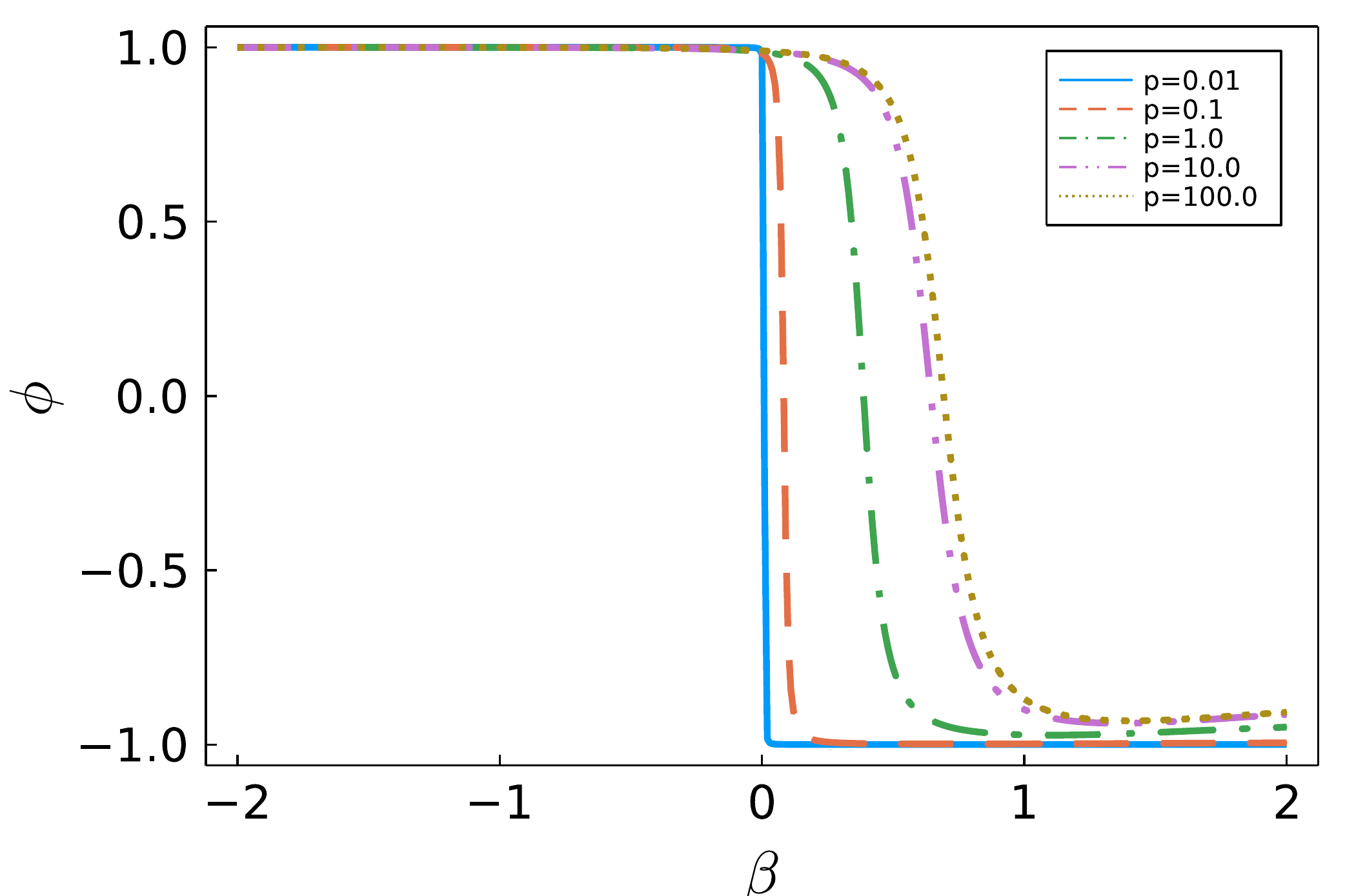}
  \caption{(Color online) Correlation coefficient $\phi$ as a function of $\beta$ for $L=N=6$.}
\label{cc}
\end{figure}

\section{Analytic results for overlap centrality}\label{Rig}
Let us discuss the weak-supplanting limit of $p\ll 1$, where general results in terms of overlap centrality are available. Although we have focused on only the Potts energy as $E(\bvec{x})$ in Section \ref{Potts}, hereafter, 
we consider all of the models which belong to the general class of energy functions satisfying the permutation symmetry condition \eqref{sym}. The neighbor matrix $\mathcal{R}$ \eqref{eq:definition_of_neighbor_matrix}, the overlap centrality $O_i \; (1 \le i \le N)$ \eqref{eq:definition_of_overlap_centrality}, and the correlation coefficient $\phi$ of the overlap centrality with agents’ rank \eqref{eq:correlation_coefficient} can also be defined for the general cases in the same manner. These are indeed the main subject in this section and also this paper. 
Hereafter, we use the state vector description, fix parameters 
$\beta, L, N$ as arbitrary values, and consider 
$p$-dependence of the dynamics 
unless otherwise specified.

The main goal in this section is 
to derive perfect correlation, which means that the correlation coefficient \eqref{eq:correlation_coefficient} satisfies $\phi=\pm 1$,
in the weak-supplanting limit of $p\to 0$ as long as $\phi$ is not singular in the 
sense mentioned after the definition of $\phi$.
In order to take a step forward, we start with introducing auxiliary stochastic processes {\eqref{tn0matrix}}, by which the transition matrix of the model can be completely reconstructed. 
Then, using this decomposition property {\eqref{eq:supplanting_decomposition}}, 
we construct another decomposition form {\eqref{eq:BetaDecomp0}} of the transition matrix, which we call {\it beta decomposition}, where the asymptotic behaviors in terms of $p\to 0$ can be rigorously estimated.
Note that the assumption of \eqref{sym} is essential in the derivation of key properties \eqref{eq:PermCanVec} and \eqref{eq:CompBOp}.

\subsection{Decompositions of transition matrix}\label{sct:Decomp}

In this subsection, we would like to introduce \textit{beta decomposition} \eqref{eq:BetaDecomp0} of the operator $\TransOp{}{}$. This decomposition enables us to investigate an asymptotic behavior of $\TransOp{}{}$ for $p \ll 1$ because of an asymptotic property \eqref{eq:BetaEstimate1}.
To describe it, first we introduce another decomposition, called \textit{supplanting decomposition}. This is relatively easy to describe, and simplifies the description of beta decomposition.
For details, see Appendix \ref{sct:SDecomp} and \ref{sct:BetaDecomp}.

Let us define the partial sum in \eqref{eq0:Trep}. 
For an integer $1 \le n \le N-1$, we define
\begin{equation}\label{tn0matrix}
\begin{aligned}
    \TransOp{n}{} \coloneqq &
    \sum_{\substack{1 \le i \le N \\ d = \dPM}}
    \sum_{\substack{\bm{x} \in X^N}} \delta(\#S(\bm{x}, i, d) , n)\\
    &\Bigg[\sum_{\substack{j \in S(\bm{x}, i, d) \\ d' = \dPM}}
    T(\bm{x} \to \Shift{j}{d'}\Shift{i}{d} \bm{x}) (\ShiftOp{j}{d'}
    - \id_H^{\otimes N}) \Bigg] \ShiftOp{i}{d} \CheckOp{}{\bm{x}}.
\end{aligned}\end{equation}
This operator $\TransOp{n}{}$ is the second term in \eqref{eq1:Trep} restricting indices $\bm{x}, i, d$ to those with $\#S(\bm{x}, i, d) = n$.
For any subset $\mathcal{S} \subseteq \{1, 2, \dots, N-1\}$, the matrix $\TransOp{0}{} + \sum_{i \in \mathcal{S}}\TransOp{i}{}$ is also a stochastic matrix. For example, the matrix $\TransOp{0}{} + \TransOp{n}{}$ represents the supplanting process only when $\#S(\bm{x}, i, d) = n$.
By definition, we have
\begin{equation}\label{eq:supplanting_decomposition}
    \TransOp{}{} = \TransOp{0}{} + \TransOp{1}{} 
    + \TransOp{2}{} + \dots + \TransOp{N-1}{}.
\end{equation}
This is a decomposition of the operator $\TransOp{}{}$, which we call 
the \textit{supplanting decomposition}.
Note that the coefficients $\langle \bm{y} | \TransOp{n}{} | \bm{x} \rangle$ of $n$-th term $\TransOp{n}{}$ for $n \ge 1$ are estimated as
\begin{equation}\label{eq:STermEstimate}
    | \langle \bm{y} | \TransOp{n}{} | \bm{x} \rangle | \le  \frac{p}{2}.
\end{equation}
For details, see Appendix \ref{sct:nthCoeff}.
In particular, the coefficients of $\TransOp{n}{}$ are estimated by $\mathcal{O}(p)$ when $p \to +0$.
Moreover, since at most one of $\langle \bm{y} | \TransOp{n}{} | \bm{x} \rangle \; (1 \le n \le N-1)$ is non-zero for fixed $\bm{x}$ and $\bm{y}$, we also have
\begin{equation} \label{eq:DifferenceOpEstimate}
    | \langle \bm{y} | (\TransOp{}{} - \TransOp{0}{} )| \bm{x} \rangle | \le  \frac{p}{2}.
\end{equation}

Though we explicitly estimate the coefficients of $\TransOp{n}{}$ in (\ref{eq:STermEstimate}), the supplanting decomposition does not give effective truncation of $\TransOp{}{}$ 
in terms of small $p$. 
Thus, we look for another decomposition of $\TransOp{}{}$
\begin{equation} \label{eq:BetaDecomp0}
     \TransOp{}{} = \TransOp{0}{} 
    + \BetaTerm{1} + \dots + \BetaTerm{N-1}
\end{equation}
satisfying
\begin{equation}\label{eq:BetaEstimate0}
    \langle \bm{y} | \BetaTerm{m} | \bm{x} \rangle 
    = \mathcal{O}(p^{m}) \text{ as } p \to +0.
\end{equation}
If we find such an expansion,
we have
\begin{equation}\label{eq:BetaEstimate1}
    \TransOp{}{} = \TransOp{0}{} 
    + \BetaTerm{1} + \dots + \BetaTerm{m} + \mathcal{O}(p^{m+1})
\end{equation}
for $1 \le m \le N-1$.

It would not be straightforward to practically find such an expansion. Nevertheless, in this case, through a rather tricky procedure as shown in Appendix \ref{sct:PrfBetaDecomp}, one can prove that \eqref{eq:BetaDecomp0} and \eqref{eq:BetaEstimate0} are satisfied by the following definition of $\BetaTerm{m}$:
\begin{align}\label{ummatrix}
    \BetaTerm{m} \coloneqq &\frac{(-1)^{m+1} B(m, 1+1/p)}{p} \nonumber \\
    &\quad \sum_{m \le n \le N} \binom{n-1}{m-1}(1+np) \TransOp{n}{},
\end{align}
for $1 \le m \le N-1$.
Here, $B(a,b)$ is the beta function 
and $\displaystyle\binom{n-1}{m-1}$ is the binomial coefficient. See \eqref{eq:BetaDecomp} and \eqref{eq:BetaDecompEstimate} of Appendix \ref{MATRIX} for the detail of the derivation; see also Remark \ref{rem:Addendum} for the motivation of this decomposition. We call this expansion \textit{beta decomposition}. Note that, for an integer $m \ge 1$,
\begin{equation}\label{betafunction}
    B \left( m, 1+\frac{1}{p} \right) =
    \frac{(m-1)!p^{m}}{(1+p)(1+2p)\dots(1+mp)}.
\end{equation}

By substituting (\ref{tn0matrix}) and (\ref{betafunction}) into (\ref{ummatrix}) with certain sets of transformations in Appendix \ref{sct:AnotherDescr} and \ref{sct:AnotherBetaDecomp},
we reach another representation of $\BetaTerm{m}$
as in \eqref{eq:SubOp} and \eqref{eq:UmDescr}:
\begin{equation}\label{u1}
\begin{aligned}
    \BetaTerm{m} = &\frac{(-1)^{m-1}(m-1)!p^m}{2(1+p)(1+2p)\dots(1+mp)} \\
    &\times \sum_{\substack{1 \le i < i_1 < \dots < i_{m} \le N}}
    \left[ 
    \sum_{1 \le k \le m}
    (\ShiftOp{i_k}{+} + \ShiftOp{i_k}{-} - 2\id^{\otimes N}_H)
    \right] \\
    &\qquad \times\CheckOp{i, i_1, \dots, i_m}{}
    \TransOp{0,\text{move}}{i},
\end{aligned}
\end{equation}
where we define
\begin{align}\label{tmove}
    \TransOp{0,\text{move}}{i} &\coloneqq
    \sum_{d = \pm}
    \sum_{\bm{x} \in X^N} T_0(\bm{x} \to \Shift{i}{d} \bm{x}) \ShiftOp{i}{d}
    \CheckOp{}{\bm{x}}, \\
    \CheckOp{i, i_1, \dots, i_m}{} &\coloneqq
    \sum_{x \in X} \CheckOp{i}{x} \CheckOp{i_1}{x} 
    \dots \CheckOp{i_m}{x}.
\end{align}
This representation of $\BetaTerm{m}$
is suitable for the further calculation related to permutation symmetry in Section \ref{sct:Perfect_correlation}. 
In the following sections, $\BetaTerm{1}$ is the one 
we mainly consider in the weak supplanting limit $p \to +0$.

Recall that $\PermOp{\sigma}$ defined in \eqref{eq:PermOpDefinition} denotes the permutation operator corresponding to a permutation $\sigma \in \mathfrak{S}_N$.
In Appendix \ref{sct:PermOperators} and \ref{sct:TransFirst}, 
we obtain commutation relations 
between permutation operators and other operators:
\begin{align}
    \PermOp{\sigma}\ShiftOp{i_1}{\dPM} 
    &= \ShiftOp{\sigma(i_1)}{\dPM}\PermOp{\sigma},
    \label{eq:CommPerm1} \\
    \PermOp{\sigma}\CheckOp{i_0,i_1}{} 
    &= \CheckOp{\sigma(i_0),\sigma(i_1)}{}\PermOp{\sigma}, 
    \label{eq:CommPerm2} \\
    \PermOp{\sigma}\TransOp{0,\text{move}}{i_0} 
    &= \TransOp{0,\text{move}}{\sigma(i_0)}\PermOp{\sigma}.
    \label{eq:CommPerm3} 
\end{align}
These relations will be used in \eqref{eq:CompBOp}. 

\subsection{Existence of perfect correlation}\label{sct:Perfect_correlation}
By using beta decomposition of the transition matrix obtained above, 
we are going to show that the correlation coefficient $\phi$ exhibits perfect correlation $|\phi|=1$ 
in the limit of $p \to 0$.



One can use the Brillouin--Wigner type perturbation theory \cite{Brilloin-Wigner-type perturbation theory} to rewrite the stationary state $\ket{P(\beta,p)}$ of $\TransOp{}{}(\beta, p)$ as a perturbation expansion from the stationary state of $\TransOp{0}{}$. For that purpose, we introduce some symbols. Let $\ket{\Pcan(\beta)}$ be the stationary state of $\TransOp{0}{}(\beta)$, i.e., the state corresponding to the Gibbs distribution \eqref{eq:GibbsDistr}:
\alit{
    \ket{\Pcan(\beta)}
    \coloneqq \sum_{\bvec{x}}\frac{\exp(-\beta E(\bvec{x}))}{Z(\beta)}\ket{\bvec{x}},
}
which satisfies
\alit{
    \TransOp{0}{}(\beta)\ket{\Pcan(\beta)}
    = \ket{\Pcan(\beta)},
}
and
\begin{align}
    \sum_{\bvec{x}}\braket{\bvec{x}|\Pcan(\beta)}
    = 1.
\end{align}
Note that $\ket{\Pcan(\beta)}$ is invariant under a permutation, that is,  it holds that 
\begin{align}\label{eq:Pcan_sym}
    \PermOp{\sigma}\ket{\Pcan(\beta)}=\ket{\Pcan(\beta)}
\end{align}
for any $\sigma\in\mathfrak{S}_N$, because of the permutation symmetry condition \eqref{sym} for energy function.
Let $\pr(\beta)$ be a projection operator on $H_X^{\otimes N}$ to the orthogonal complement of the subspace $\C\ket{\Pcan(\beta)}$:
\alit{\label{eq:projection op pi}
    \pr(\beta)
    \coloneqq \id_H^{\otimes N}-\frac{\ket{\Pcan(\beta)}\bra{\Pcan(\beta)}}{\braket{\Pcan(\beta)|\Pcan(\beta)}},
}
and $\GOp(\beta)$ be a linear operator from $H_X^{\otimes N}$ to itself:
\alit{
\GOp(\beta)
    \coloneqq \bka{\id_H^{\otimes N}-\TransOp{0}{}(\beta)}^{-1}\pr(\beta).
}
Here, 
we need to set the coefficient of the term $\id_H^{\otimes N}$ in $\GOp(\beta)$ as the eigenvalue $1$ corresponding to
the eigenvector $\ket{P(\beta, p)}$ of $\TransOp{}{}(\beta, p)$ (see \eqref{eq:stationary state for T}).
With the above notations, $\ket{P(\beta,p)}$ can be written as follows:
\alit{\label{ptheory}
    &\ket{P(\beta,p)}\\
    &= C(\beta,p)\bkc{
        \ket{\Pcan(\beta)}+\sum_{n=1}^\infty \bka{\GOp (\TransOp{}{}-\TransOp{0}{})}^n\ket{\Pcan(\beta)}
    },
}
where $C(\beta,p)$ is the positive normalization factor of $\ket{P(\beta,p)}$ such that $\sum_{\bvec{x}\in X^N}\braket{\bvec{x}|P(\beta,p)}=1$. 
Since $\TransOp{}{}-\TransOp{0}{}=\mathcal{O}(p)$, one can estimate that $C(\beta,p)=1+\mathcal{O}(p)$.

As obtained in (\ref{eq:BetaDecomp0}) and (\ref{u1}), within the asymptotic regime of small $p$, $\TransOp{}{}-\TransOp{0}{}$ is can be expanded with powers of $p$, and then we have
\alit{
    \TransOp{}{}-\TransOp{0}{}
    = \BetaTerm{1} + \mathcal{O}(p^2),
}
where
\alit{\label{u10}
    \BetaTerm{1}
    = \frac{p/2}{1+p}\sum_{1\le i_0<i_1\le N}\bka{\ShiftOp{i_1}{\dP}+\ShiftOp{i_1}{\dM}-2\id_{H}^{\otimes N}}\CheckOp{i_0,i_1}{}\TransOp{0,\text{move}}{i_0}.
}
By substituting (\ref{u10}) into (\ref{ptheory}), we have
\alit{\label{eq:perturbative_decomposition}
    &\ket{P(\beta,p)}\\
    &= C(\beta,p)\bkc{
        \ket{\Pcan(\beta)} + \GOp\BetaTerm{1}\ket{\Pcan(\beta)} + \mathcal{O}(p^2)
    }\\
    &= C(\beta,p)\bkc{
        \ket{\Pcan(\beta)} + \frac{p}{1+p}\sum_{1\le i_0<i_1\le N}\BOp{i_0,i_1}\ket{\Pcan(\beta)}}\\
    &\hspace{70mm} + \mathcal{O}(p^2),
}
where we define:
\alit{\label{eq:B_i0i1}
    \BOp{i_0,i_1}
    \coloneqq \frac{1}{2}\GOp\bka{\ShiftOp{i_1}{\dP}+\ShiftOp{i_1}{\dM}-2\id_{H}^{\otimes N}}\CheckOp{i_0,i_1}{}\TransOp{0,\text{move}}{i_0}.
}
This operator $\BOp{i_0,i_1}$ is dependent on $\beta$ but independent of $p$.
Thus, using $r_{ii}=1$ for any $i$, the overlap centrality can be written as follows: 


\alit{\label{eq:expression_overlap_centrality_by_A}
    &O_i
    = \sum_{1\le j\le N}\sum_{\bvec{x}\in X^N}\delta\bka{x_i,x_j}\braket{\bvec{x}|P(\beta,p)} -1 \\
    &= C(\beta,p)\bkc{
        \mathcal{A}_0(i)
        +\frac{p}{1+p} \sum_{1\le i_0<i_1\le N}\mathcal{A}_1(i,i_0,i_1)}
        -1\\
        &\hspace{70mm} +\mathcal{O}(p^2),
}
where 
\begin{align}
 &\mathcal{A}_0(i)=\sum_{1\le j\le N}\sum_{\bvec{x}\in X^N}\delta\bka{x_i,x_j}
    \braket{\bvec{x}|\Pcan(\beta)},\\
    \label{eq:definition_of_A1}
  &\mathcal{A}_1(i,i_0,i_1)=     
        \sum_{1\le j\le N}\sum_{\bvec{x}\in X^N}\delta\bka{x_i,x_j}\bra{\bvec{x}}\BOp{i_0,i_1}\ket{\Pcan(\beta)}.
\end{align}
Here $\mathcal{A}_0(i)$ and $\mathcal{A}_1(i, i_0, i_1)$ are constant as a function of $p$.


Recall that $\PermOp{\sigma}$ is the permutation operator corresponding to a permutation $\sigma \in \mathfrak{S}_N$ (see \eqref{eq:PermOpDefinition} for definition).
 By using \eqref{eq:Pcan_sym}, we have
\begin{align} \label{eq:PermCanVec}
    &\mathcal{A}_0(i)\notag\\
    &= \sum_{1\le j\le N}\sum_{\bvec{x}\in X^N}\delta\bka{x_i,x_j}\braket{\bvec{x}|\PermOp{\sigma}|\Pcan(\beta)}\notag\\
    &= \sum_{1\le j\le N}\sum_{\bvec{x}\in X^N}\delta\bka{x_{\sigma(i)},x_{\sigma(j)}}\braket{\bvec{x}|\Pcan(\beta)}\notag\\
    &= \sum_{1\le j\le N}\sum_{\bvec{x}\in X^N}\delta\bka{x_{\sigma(i)},x_{j}}\braket{\bvec{x}|\Pcan(\beta)} \notag\\
    &= \mathcal{A}_0(\sigma(i)).
\end{align}
This means that $\mathcal{A}_0(i)$ does not depend on $i$:
\begin{equation}
\mathcal{A}_0(i) = \mathcal{A}_0(1) \eqqcolon \mathcal{B}_0.
\end{equation}

Moreover, by using 
\eqref{eq:CommPerm1}, \eqref{eq:CommPerm2},  \eqref{eq:CommPerm3}, and
\eqref{eq:Pcan_sym},
we have $\PermOp{\sigma} \BOp{i_0,i_1} = \BOp{\sigma(i_0),\sigma(i_1)}\PermOp{\sigma}$. 
Therefore, it follows that
\begin{align} \label{eq:CompBOp}
    &\mathcal{A}_1(i,i_0,i_1)\notag\\
    &= \sum_{1\le j\le N}\sum_{\bvec{x}\in X^N}\delta\bka{x_i,x_j}\bra{\bvec{x}}\PermOp{\sigma}^{-1}\BOp{\sigma(i_0),\sigma(i_1)}\PermOp{\sigma}\ket{\Pcan(\beta)}\notag\\
    &= \sum_{1\le j\le N}\sum_{\bvec{x}\in X^N}\delta\bka{x_{\sigma(i)},x_{\sigma(j)}}\bra{\bvec{x}}\BOp{\sigma(i_0),\sigma(i_1)}\ket{\Pcan(\beta)}\notag\\
    &= \sum_{1\le j\le N}\sum_{\bvec{x}\in X^N}\delta\bka{x_{\sigma(i)},x_j}\bra{\bvec{x}}\BOp{\sigma(i_0),\sigma(i_1)}\ket{\Pcan(\beta)}\notag\\
    &= \mathcal{A}_1(\sigma(i),\sigma(i_0),\sigma(i_1))
\end{align}
for any $\sigma\in\mathfrak{S}_N$. In this sense, $\mathcal{A}_{1}$ preserves permutation symmetry in spite of the permutation-symmetry breaking of $\TransOp{}{}$.
By this equation, we obtain 
\begin{align}\label{eq:definition_of_B123}
    \mathcal{A}_1(i,i_0,i_1)
    =
    \begin{dcases}
        \mathcal{A}_1(1, 1, 2) \eqqcolon \mathcal{B}_1  & (\text{if}\ i=i_0)\\
        \mathcal{A}_1(2, 1, 2) \eqqcolon \mathcal{B}_2  & (\text{if}\ i=i_1)\\
        \mathcal{A}_1(3, 1, 2) \eqqcolon \mathcal{B}_3  & (\text{if}\ i \neq i_0, i_1)
    \end{dcases}
\end{align}
for any $i$ and any pair $i_0<i_1$.
These quantities $\mathcal{B}_1, \mathcal{B}_2$, and $\mathcal{B}_3$ are independent of $p$.

Note that for a given agent $i$, the number of pairs $(i_0,i_1)$ which satisfies each of the above conditions corresponding with $\mathcal{B}_1$, $\mathcal{B}_2$, and $\mathcal{B}_3$ is $N-i$, $i-1$, and $(N-1)(N-2)/2$, respectively. 
Therefore, we obtain
\alit{\label{eq:sum_of_A1}
    &\sum_{1\le i_0<i_1\le N}\mathcal{A}_1(i,i_0,i_1)\\
    &= (N-i)\mathcal{B}_1+(i-1)\mathcal{B}_2+\frac{(N-1)(N-2)}{2}\mathcal{B}_3.
}
Substituting \eqref{eq:sum_of_A1} into \eqref{eq:expression_overlap_centrality_by_A}, we find that
\begin{align} \label{eq:overlap_centrality_estimate}
    O_i
    = \frac{p}{1+p}C(\beta,p)(\mathcal{B}_2-\mathcal{B}_1)i+c_0+\mathcal{O}(p^2),
\end{align}
where $c_0$ is a real number independent of rank $i$:
\begin{align}
    c_0
    &= \frac{p}{1+p}C(\beta,p)\bka{N\mathcal{B}_1-\mathcal{B}_2+\frac{(N-1)(N-2)}{2}\mathcal{B}_3} \notag\\
    &\qquad +C(\beta,p)\mathcal{B}_0-1.
    \label{eq:Compute_c0}
\end{align}
With the estimation of $C(\beta,p)=1+\mathcal{O}(p)$, we can rewrite \eqref{eq:overlap_centrality_estimate} as
\begin{align}
\label{eq:overlap_centrality_estimate2}
    O_i
    = \frac{p}{1+p}(\mathcal{B}_2-\mathcal{B}_1)i+c_0+\mathcal{O}(p^2).
\end{align}
Note that, in \eqref{eq:expression_overlap_centrality_by_A}, \eqref{eq:overlap_centrality_estimate}, and \eqref{eq:overlap_centrality_estimate2}, the terms in $\mathcal{O}(p^2)$ could depend on $i$.
\rd{As shown in Fig.\ \ref{tilt_of_Oi}, we have confirmed that 
$p$-dependence in the first term of \eqref{eq:overlap_centrality_estimate2} 
is consistent to that computed by the exact diagonalization. Moreover, as shown in Fig.\ \ref{B2minusB1}, we can see good agreement of the value $\mathcal{B}_2 - \mathcal{B}_1$ by two distinct methods: one method is by definition \eqref{eq:definition_of_B123}. The other is by comparing $O_i$ computed by the exact diagonalization with rank-dependence in \eqref{eq:overlap_centrality_estimate2}}.

\begin{figure*}
\begin{tabular}{cc}
\begin{minipage}[t]{0.45\linewidth}
    \includegraphics[width=8cm,clip]{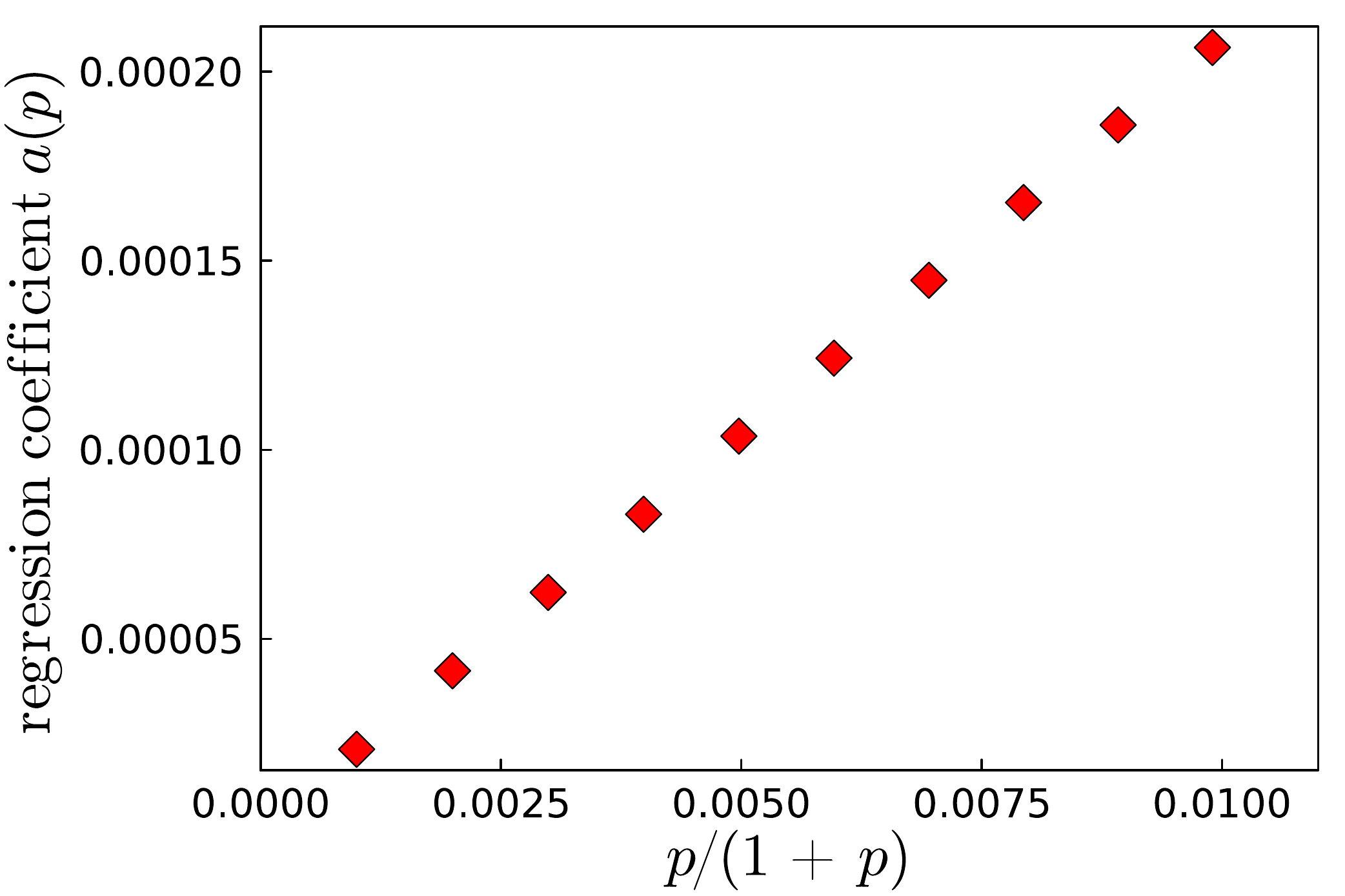}
\end{minipage}
\begin{minipage}[t]{0.45\linewidth}
    \includegraphics[width=8cm,clip]{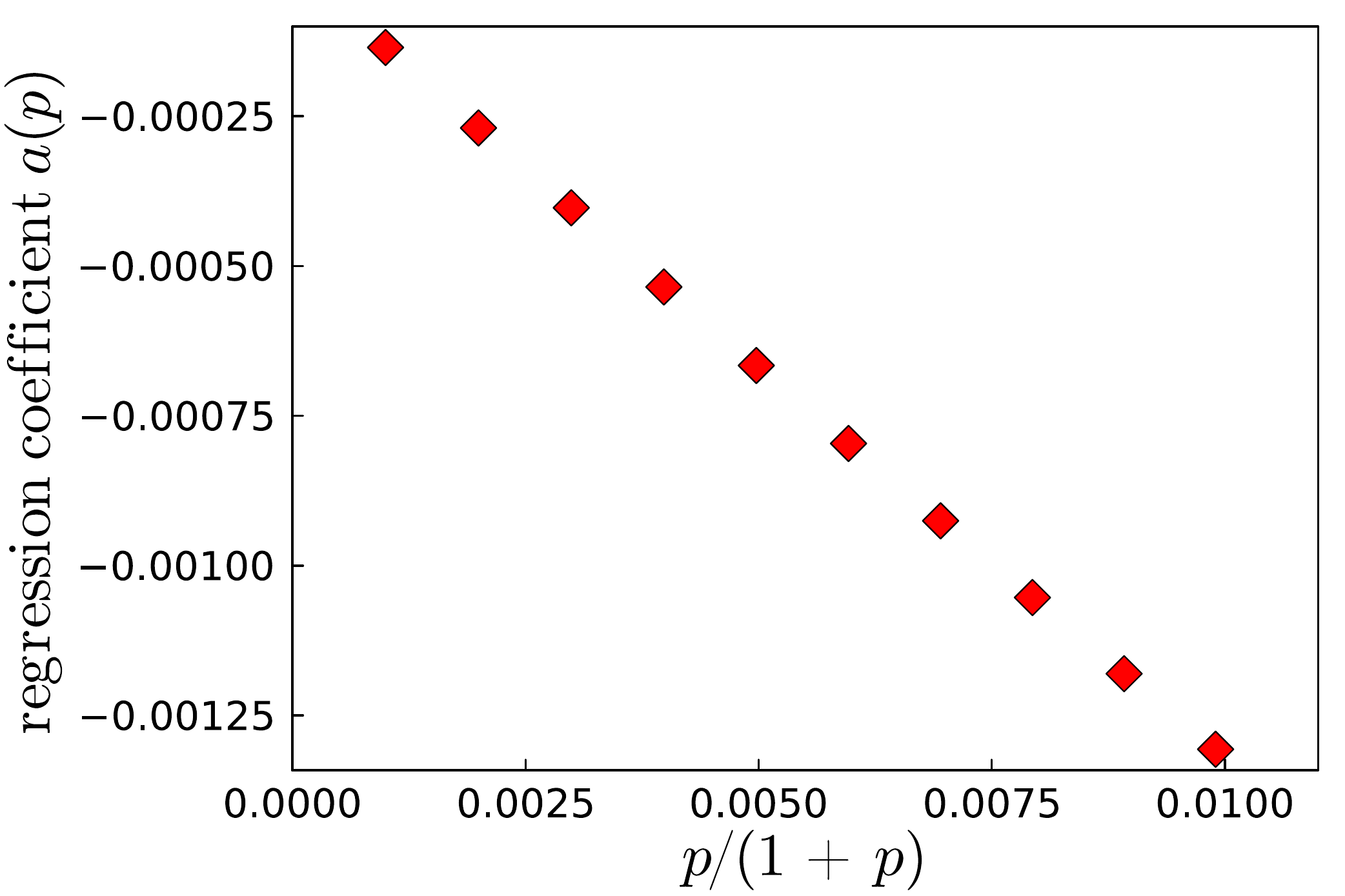}
\end{minipage}
\end{tabular}
    \caption{\rd{(Color online) 
    The overlap centrality $O_j$ computed by the exact diagonalization for $L=N=6$ is expressed by $O_j\simeq a(p)j+b(p)$, where coefficient $a(p)$ and $b(p)$ are determined by linear regression.
    These two figures plot the regression coefficient $a(p)$ versus $p/(1+p)$ for $p=0.001, 0.002, \dots, 0.01$;
     $\beta = -1$ (left) and $\beta = 1$ (right).
    The correlation coefficient between $a(p)$ and $p/(1+p)$ is $0.9999998...$ for $\beta=-1$ and $-0.99997...$ for $\beta=1$. 
    Note that $p$-dependence of the first term in \eqref{eq:overlap_centrality_estimate2} corresponds to that the correlation coefficient is equal to $1$ for $\beta=-1$ and $-1$ for $\beta=1$, respectively.}}
\label{tilt_of_Oi}
\end{figure*}

\begin{figure}
    \includegraphics[width=8cm,clip]{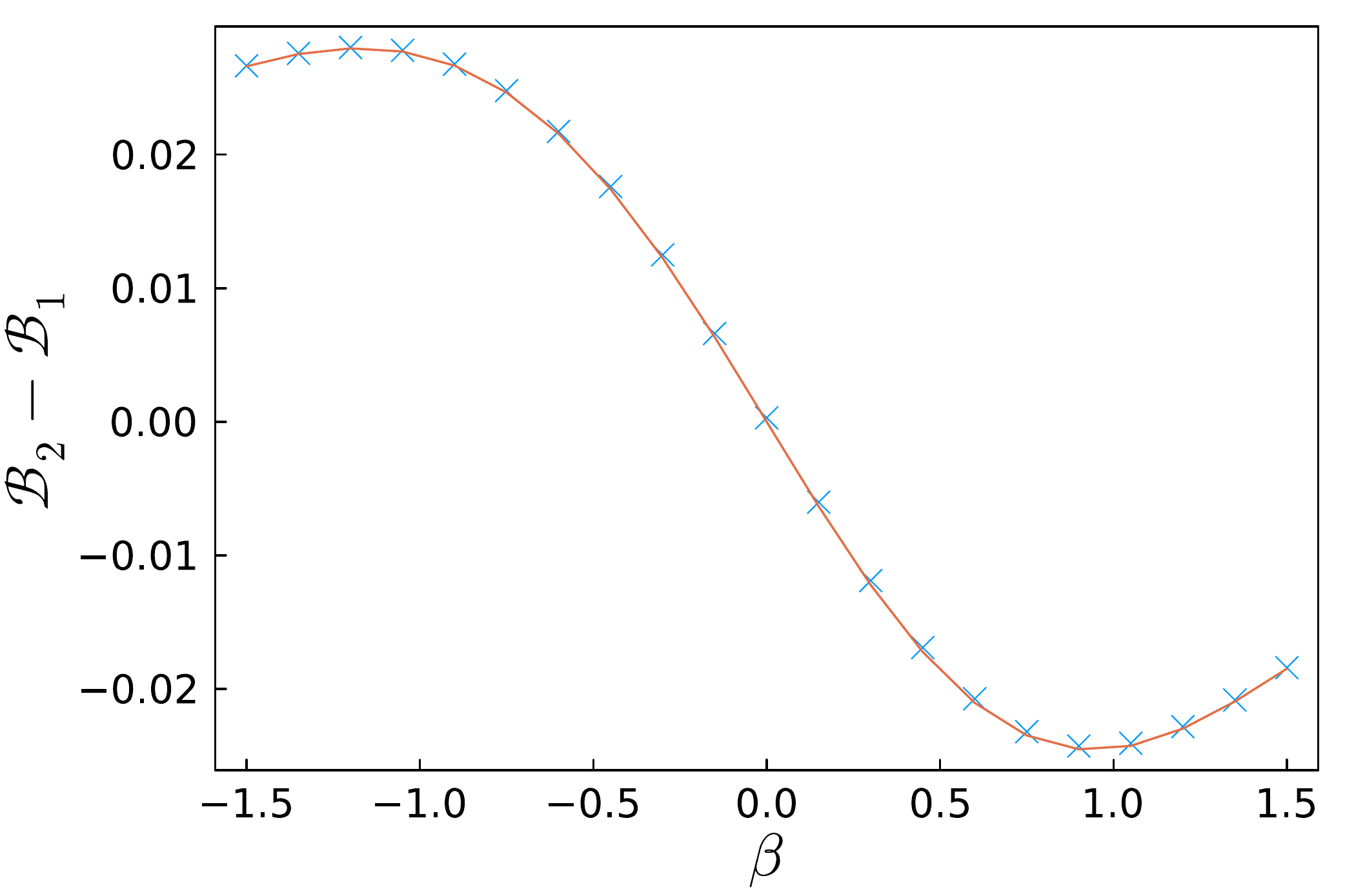}
    \caption{\rd{(Color online) 
    $\mathcal{B}_2-\mathcal{B}_1$ as a function of $\beta$ for $L=N=3$. The line corresponds to the computation through its definition \eqref{eq:definition_of_B123}. 
    The blue crosses correspond to the estimation through \eqref{eq:overlap_centrality_estimate2}
    where ${O}_i$ is calculated by the exact diagonalization. 
    Practically, the estimated value is calculated as $\tilde{a}(\beta)$ where $a(p) \simeq \tilde{a}(\beta)\frac{p}{1+p}+\tilde{b}(\beta)$ through linear regression.  }}
\label{B2minusB1}
\end{figure}

We write the remaining term in \eqref{eq:overlap_centrality_estimate2} as
\begin{equation}
    \varepsilon_i(\beta,p) \coloneqq O_i
    - \frac{p}{1+p}(\mathcal{B}_2-\mathcal{B}_1)i - c_0
    = \mathcal{O}(p^2).
\end{equation}

Let us suppose $\mathcal{B}_2 \neq \mathcal{B}_1$.
The remainder term $\varepsilon_i$ can be ignored when $p \ll 1$ compared to the term $\dfrac{p}{1+p}(\mathcal{B}_2-\mathcal{B}_1)i.$
After ignoring $\varepsilon_i$, the overlap centrality $O_i$ is a linear function with respect to rank $i$. 
Therefore the perfect correlation holds:
\begin{align}\label{eq:B_1-B_2_and_phi}
    \phi \to
    \begin{dcases}
        +1 & (\text{if}\ \mathcal{B}_1<\mathcal{B}_2) \\
        -1 & (\text{if}\ \mathcal{B}_1>\mathcal{B}_2)
    \end{dcases} \quad
    \text{as}\ p\to +0.
\end{align}
This is consistent with the observation in Fig.\ \ref{cc} as long as $p$ is small 
because the value of $\phi$ approaches $+1$ or $-1$ very closely.

The above discussion gives another indication. Ignoring $\varepsilon_i$, the linear dependency of the overlap centrality $O_i$ with respect to rank $i$ comes from permutation symmetry \eqref{eq:CompBOp} in $\mathcal{A}_1$.
Note that, as mentioned by \eqref{noncommuT}, the permutation symmetry is broken in the transition matrix if $p \neq 0$, but this symmetry is partially recovered in the quantity $\mathcal{A}_1$. 
Note that 
one can also derive 
\begin{align}
M &= \frac{p}{1+p}\bka{\dfrac{N-1}{2N}(\mathcal{B}_1+\mathcal{B}_2) + \dfrac{(N-1)(N-2)}{2N}\mathcal{B}_3} \notag \\
&\qquad+C(\beta, p)\dfrac{\mathcal{B}_0}{N}+\mathcal{O} (p^2)
\end{align}
by using (\ref{MOrel}).

Evaluating the sign of $\phi$ requires concrete calculation of both $\mathcal{B}_1$ and $\mathcal{B}_2$, 
but it gets complicated to obtain their analytic expressions for a given energy function such as the Potts energy \eqref{eq:PottsEnergy}. 
Nevertheless, it is still feasible to perform such a calculation in the case of $\beta=0$ as discussed in the next subsection. It is because at $\beta=0$, transition matrix $\TransOp{0}{}(\beta)$ gets independent of the form of energy function, resulting in getting close to that of the free random walk process.


\subsection{Singularity in $\phi$ at $\beta=0$}\label{sec:Singularity}
One can calculate the overlap centrality at $\beta=0$ for the asymptotic regime of small $p$ concretely. As a result, we show that at $\beta=0$, $\mathcal{B}_1=\mathcal{B}_2$ holds, which means that the correlation coefficient $\phi$ is singular at $\beta=0$.

From the definitions of $\mathcal{B}_1$, $\mathcal{B}_2$ in \eqref{eq:definition_of_B123} and $\mathcal{A}_1$ in \eqref{eq:definition_of_A1}, we have
\begin{gather}\label{eq:B1def}
    \mathcal{B}_1
    = \sum_{1\le j\le N}\sum_{\bvec{x}\in X^N}\delta(x_1,x_j)\braket{\bvec{x}|\BOp{1,2}|\Pcan(\beta)},\\
    \label{eq:B2def}
    \mathcal{B}_2
    = \sum_{1\le j\le N}\sum_{\bvec{x}\in X^N}\delta(x_2,x_j)\braket{\bvec{x}|\BOp{1,2}|\Pcan(\beta)}.
\end{gather}
In order to calculate $\mathcal{B}_1$ and $\mathcal{B}_2$ at $\beta=0$, we need derive an explicit expression of $\braket{\bvec{x}|\BOp{1,2}|\Pcan(0)}$. 
First, let us define 
\begin{align}
    \kket{k}
    \coloneqq \sum_{x\in X}\frac{e^{ikx}}{\sqrt{L}}\ket{x} \in H_X,
\end{align}
for $k=2\pi n/L$ and $n\in\mathbb{Z}/L\mathbb{Z}$. The orthonormal system $\bkb{\kket{k}}_k$ spans the whole space $H_X$ over $\mathbb{C}$. 
In particular 
\begin{align}
    \ket{\Pcan(0)}=\bka{\dfrac{\kket{0}}{\sqrt{L}}}^{\otimes N}
\end{align}
holds. Note that, for $\bm{k} = (k_1, k_2, \dots, k_N)$ with $k_i=2\pi n_i/L$ and $n_i\in\mathbb{Z}/L\mathbb{Z}$ with $1 \le i \le N$,
\begin{align}
    \kket{\bvec{k}}\coloneqq\kket{k_1}\otimes\cdots\otimes\kket{k_N}\in H_X^{\otimes N}
\end{align} is an eigenvector of $\TransOp{0}{}(\beta = 0)$ with an eigenvalue 
\begin{align}
\dfrac{1}{N}\bka{\cos^2\dfrac{k_1}{2}+\cdots+\cos^2\dfrac{k_N}{2}}.
\end{align}
The vector $\kket{\bvec{k}}$ is also an eigenvector of $\GOp(\beta=0)$, and the eigenvalue is
\begin{align}\label{eq:eigenvalue_of_GOp}
    N\bka{\sin^2\dfrac{k_1}{2}+\cdots+\sin^2\dfrac{k_N}{2}}^{-1}
\end{align}
for any  $\bvec{k}\neq \bvec{0} = (0, 0, \dots, 0)$.

Next, substituting the definition \eqref{eq:B_i0i1} of $\BOp{1,2}$ into the term $\braket{\bvec{x}|\BOp{1,2}|\Pcan(0)}$, 
we obtain
\begin{align}\label{eq:xB12Pcan0}
    &\braket{\bvec{x}|\BOp{1,2}|\Pcan(0)}\notag\\
    &= \frac{1}{2}\bra{\bvec{x}}\GOp(0)(\ShiftOp{2}{\dP}+\ShiftOp{2}{\dM}-2\id_H^{\otimes N})\notag\\
    &\quad\quad\quad \times\CheckOp{1,2}{}\TransOp{0,\mathrm{move}}{1}{(\beta=0)}\ket{\Pcan(0)}.
\end{align}
In order to obtain more explicit expression, let us multiply $\ket{\Pcan(0)}$ by $\TransOp{0,\mathrm{move}}{1}$, $\CheckOp{1,2}{}$, $(\ShiftOp{2}{\dP}+\ShiftOp{2}{\dM}-2\id_H^{\otimes N})$, and $\GOp(0)$ from the left, successively. Reminding of the definition (\ref{tmove}) of $\TransOp{0,\mathrm{move}}{1}$, we have
\begin{align}\label{eq:T0move_acting_on_P0}
    \TransOp{0,\mathrm{move}}{1}(0)\ket{\Pcan(0)}
    = \frac{1}{2N}\ket{\Pcan(0)},
\end{align}
and can show that
\begin{align}\label{eq:CheckOp_acting_on_P0}
    \CheckOp{1,2}{}\ket{\Pcan(0)}
    &= \sum_{x\in X}\frac{\ket{x}}{L}\otimes\frac{\ket{x}}{L}\otimes\bka{\frac{\kket{0}}{\sqrt{L}}}^{\otimes(N-2)}\notag\\
    &= \frac{1}{L^2}\sum_{k_1}\kket{k_1}\otimes\kket{{-}k_1}\otimes\bka{\frac{\kket{0}}{\sqrt{L}}}^{\otimes(N-2)}.
\end{align}
By using \eqref{eq:T0move_acting_on_P0} and \eqref{eq:CheckOp_acting_on_P0}, we obtain
\begin{align}\label{eq:DeltaXiTP}
    &(\ShiftOp{2}{\dP}+\ShiftOp{2}{\dM}-2\id_H^{\otimes N})\CheckOp{1,2}{}\TransOp{0,\mathrm{move}}{1}(0)\ket{\Pcan(0)}\notag\\
    &= -\frac{1}{2N}\frac{1}{L^2}\sum_{k_1}4\sin^2\frac{k_1}{2}\kket{k_1}\otimes\kket{{-}k_1}\otimes\bka{\frac{\kket{0}}{\sqrt{L}}}^{\otimes(N-2)}.
\end{align}
Substituting \eqref{eq:DeltaXiTP} into \eqref{eq:xB12Pcan0}, and using \eqref{eq:eigenvalue_of_GOp}, we find
\begin{align}\label{eq:xB12Pcan}
    & \braket{\bvec{x}|\BOp{1,2}|\Pcan(0)} \notag \\
    &= -\frac{1}{2L^2}\sum_{k_1 \neq 0}\frac{e^{ik_1 x_1}}{\sqrt{L}}\frac{e^{-ik_1 x_2}}{\sqrt{L}}\bka{\frac{1}{L}}^{N-2}\notag\\
    &= -\frac{1}{2L^{N+1}}\bka{L\delta(x_1,x_2)-1}.
\end{align}
Recalling \eqref{eq:B1def} and \eqref{eq:B2def} with
\begin{align} \label{eq:CombinatorialSum}
    &\sum_{\bm{x} \in X^N} \delta(x_i,x_j) (L\delta(x_1, x_2) - 1) \notag \\
    &= \begin{dcases}
        L^{N-1}(L-1) & (\text{if}\ (i,j) = (1,2) \text{ or } (2,1)) \\
        0 & (\text{otherwise}),
    \end{dcases}
\end{align}
one can calculate
\begin{align}
    \mathcal{B}_1
    = \mathcal{B}_2
    = -\frac{L-1}{2L^2}.
\end{align}
Thus, it turns out that $\phi$ is singular at $\beta=0$.

One can also calculate $\mathcal{B}_0=N/L$, and using \eqref{eq:CombinatorialSum}, $\mathcal{B}_3 = 0$.
By substituting the value of $\mathcal{B}_\ell$ ($\ell=0,1,2,3$) into \eqref{eq:Compute_c0} and \eqref{eq:overlap_centrality_estimate2}, the overlap centrality $O_i$ is $\mathcal{O}(N/L)$.
This is consistent with the uniform distribution corresponding to the case of $\beta=0$.
Therefore, it is reasonable that 
$\mathcal{B}_0=\mathcal{O}(N/L)$ and $\mathcal{B}_1 =\mathcal{B}_2 =\mathcal{O}(L^{-1})$ as functions of $N$ and $L$. 
Note that one can also evaluate 
\begin{equation}
    M=\dfrac{1}{L}\bka{C(0,p) - \dfrac{p}{2(1+p)}\dfrac{N-1}{N}\dfrac{L-1}{L}} +\mathcal{O}(p^2).
\end{equation}

\subsection{Comparison between exact diagonalization and analytic result}
In this subsection, we shall illustrate the behavior of the correlation coefficient $\phi$ computed by exact diagonalization in comparison with analytic discussion in Section \ref{Rig}.

In order to consider $\beta$-dependence of $\phi$, let us fix $p$ as a small but non-zero value and change the value of $\beta$.
When the parameter $\beta$ varies with satisfying the condition
\begin{equation} \label{eq:no_longer_dominant}
    \mathcal{B}_2(\beta,p)-\mathcal{B}_1(\beta,p)= \mathcal{O}(\varepsilon_i(\beta, p)),
\end{equation}
the linear term $(\mathcal{B}_2 - \mathcal{B}_1)i$ in $O_i$ is no longer dominant in the rank-dependence. 
As a result, $\phi$ could change continuously from $\phi=-1$ to $\phi=1$. 
In this case, \eqref{eq:B_1-B_2_and_phi} does not necessarily hold. Exact diagonalization indicates that the range of $\beta$, where $\phi$ takes a value close to $\pm 1$, is wider as $p$ is smaller.

Let us remind of the observation in the exact diagonalization in the case of the Potts energy 
that the value of $\phi$ sharply changes around $\beta=0$ for small $p$ as shown in Fig.\ \ref{cc}. 
Assuming that this observation is 
universal for sufficiently small $p$, 
by combining the existence of perfect correlation and the singularity of $\phi$ at $\beta=0$,  
it may be a reasonable conjecture that, 
at least, in the case of the Potts energy,
$\phi$ becomes discontinuous at $\beta=0$
as a function of $\beta$ 
in the limit of $p\to +0$.

\subsection{Correspondence between Overlap centrality and Eigenvector centrality}\label{relation_between_centralities}
Let us discuss the relation between the overlap centrality $\bv{O}=(O_i)_{i=1}^N$ and the other existing ways to define centrality.

First, the overlap centrality has a connection to another existing centrality in the following sense. Let us consider a weighted complete graph, where each agent is regarded as a vertex, and an element $r_{ij}$ of the neighbor matrix $\mathcal{R}$ for the pair of agents $i,j$ is regarded as the weight of the edge $(i,j)$. 
Then, the overlap centrality defined above is equivalent to the strength centrality of the complete graph constructed above, which has been introduced in the field of network science \cite{Freeman, Barrat}.

Second, 
one can also define the eigenvector centrality of the complete graph mentioned above as the eigenvector $\bv{V}=(V_i)_{1\le i\le N}$ of $\mathcal{R}$
with the maximum eigenvalue.
Indeed, one may show that the eigenvector centrality is 
directly related to the overlap centrality when $p \ll 1$ in the following manner:
\begin{align}\label{ov-centrality}
    \bv{V}
    \propto \frac{1}{N^{3/2}c}\bv{O}-\gamma\times(1,1,\ldots,1)^T + \mathcal{O}(p^2),
\end{align}
where $c$ and $\gamma$ are constants with $\mathcal{O}(1)$. 
In particular, combining with \eqref{eq:overlap_centrality_estimate2}, we have
\begin{align}
    V_i &\propto \frac{p}{1+p}\frac{\mathcal{B}_2 - \mathcal{B}_1}{N^{3/2}c} i+ \left(\frac{c_0}{N^{3/2}c} -\gamma\right) +  \mathcal{O}(p^2)
\end{align}
where the coefficient of proportionality is independent of $i$. 
Thus, the eigenvector centrality as well as the overlap centrality depends on the rank $i$ linearly if the term of $\mathcal{O}(p^2)$ is ignored in \eqref{ov-centrality}.
Remarkably, that relation \eqref{ov-centrality} holds for general probability distribution which breaks the permutation symmetry of agents weakly. See Appendix \ref{ECENTER} 
for the explicit two conditions to hold the relation \eqref{ov-centrality} in a 
more general form.
\rd{As shown in Fig. \ref{OCandEC}, 
one can see better agreement of both centralities for smaller $p$.}

\begin{figure*}
\begin{tabular}{cc}
\begin{minipage}[t]{0.45\linewidth}
    \includegraphics[width=7cm,clip]{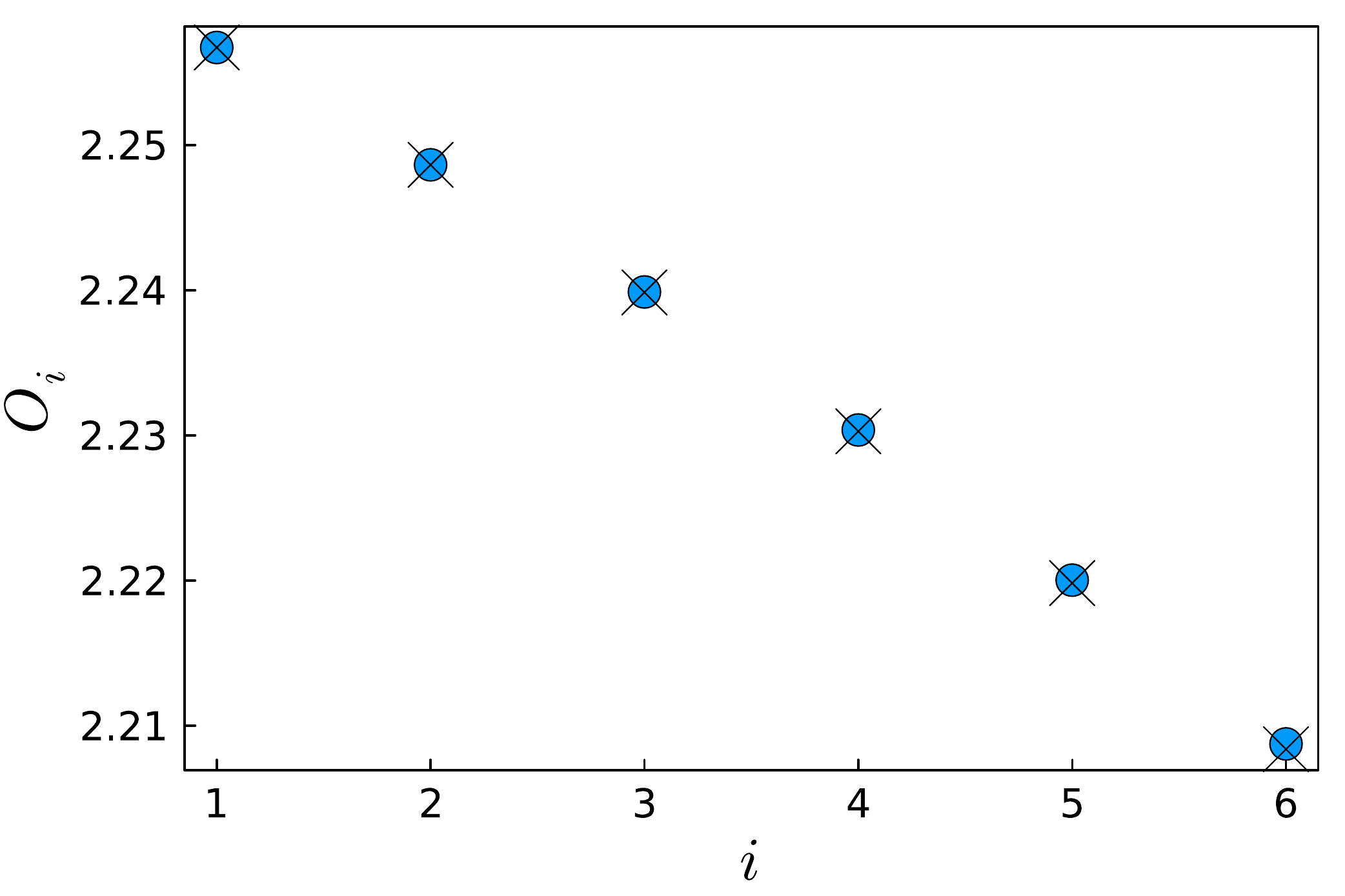}
    \includegraphics[width=7cm,clip]{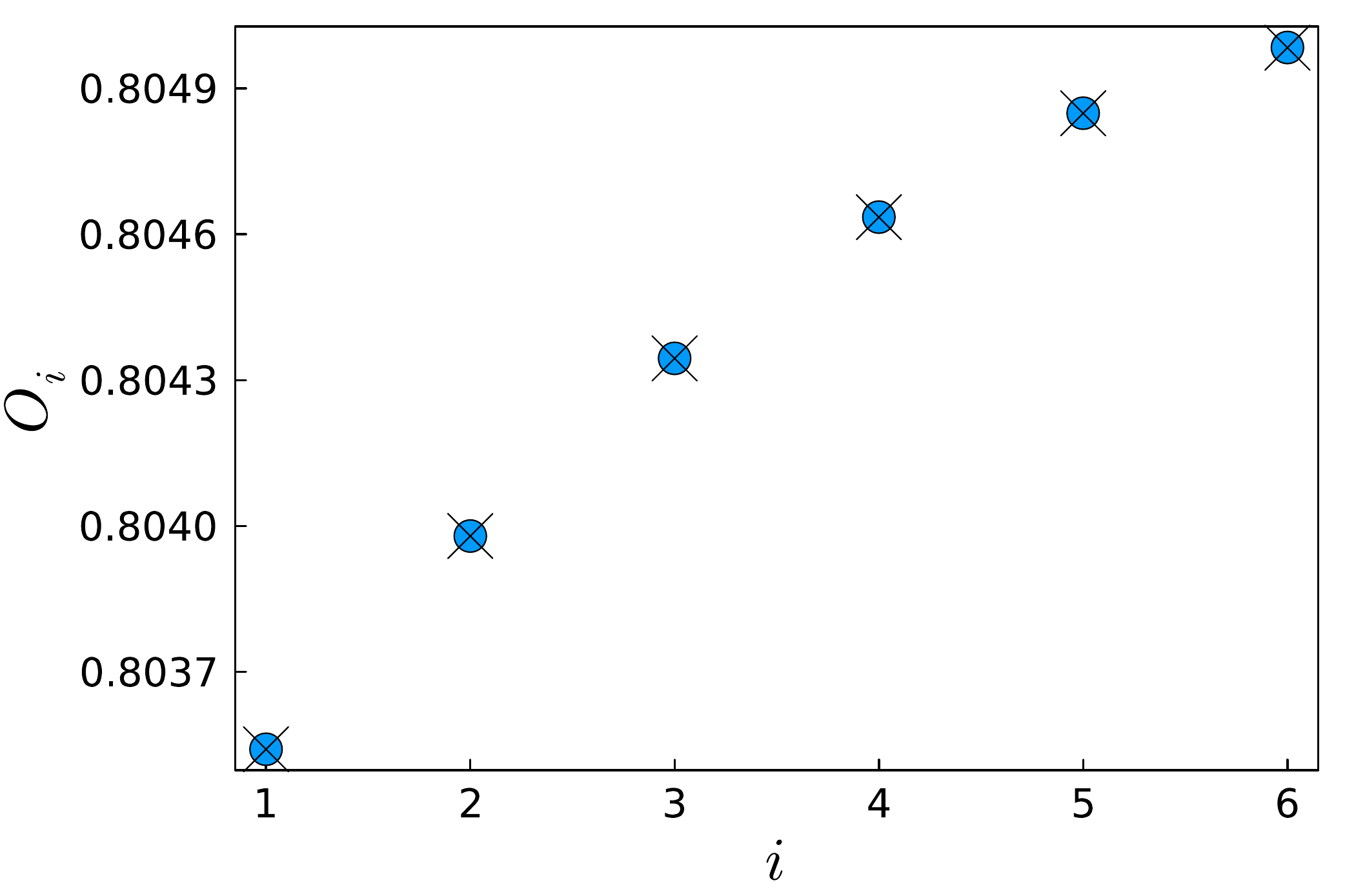}
    \includegraphics[width=7cm,clip]{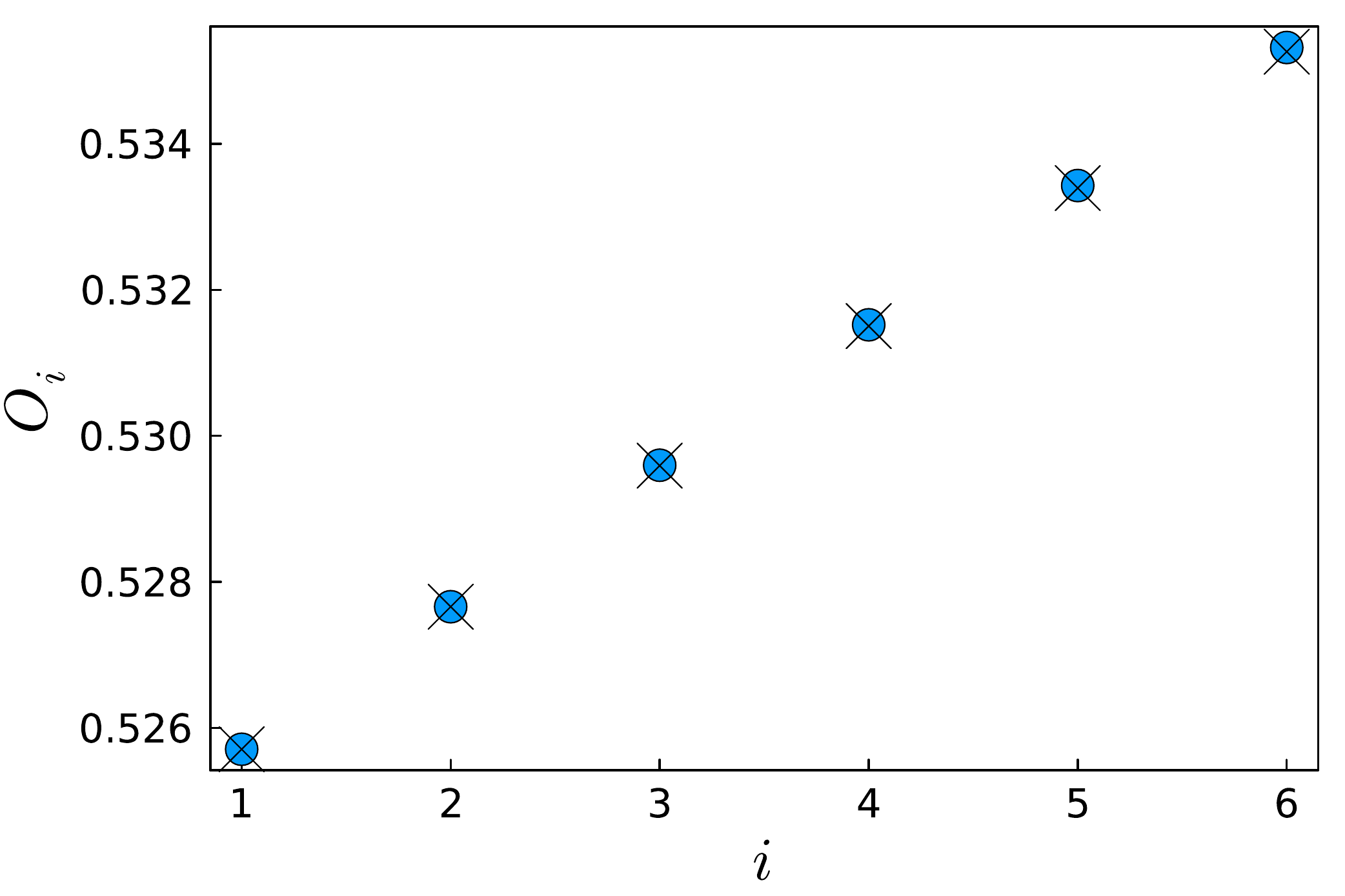}
\end{minipage}
\begin{minipage}[t]{0.45\linewidth}
    \includegraphics[width=7cm,clip]{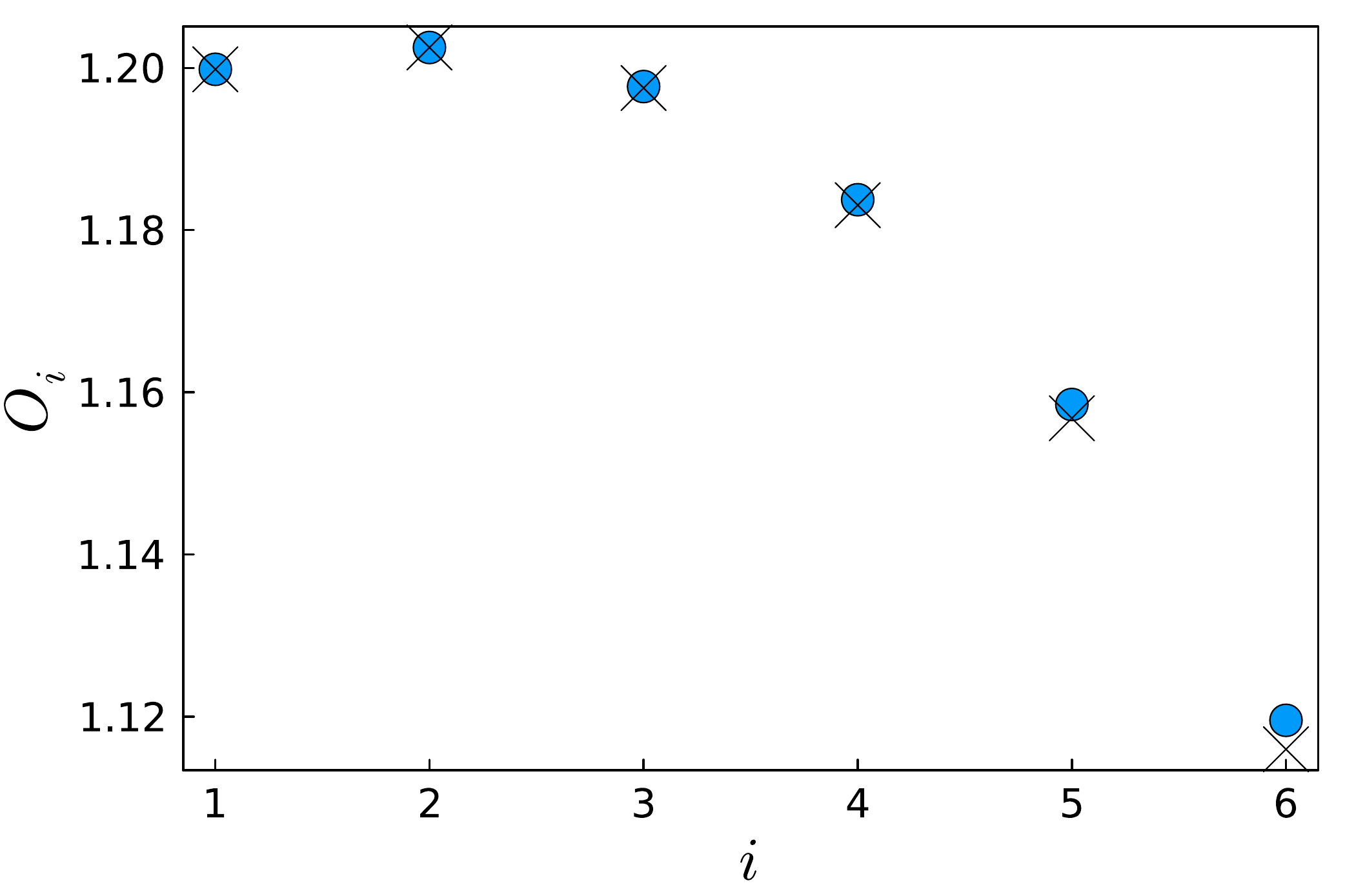}
    \includegraphics[width=7cm,clip]{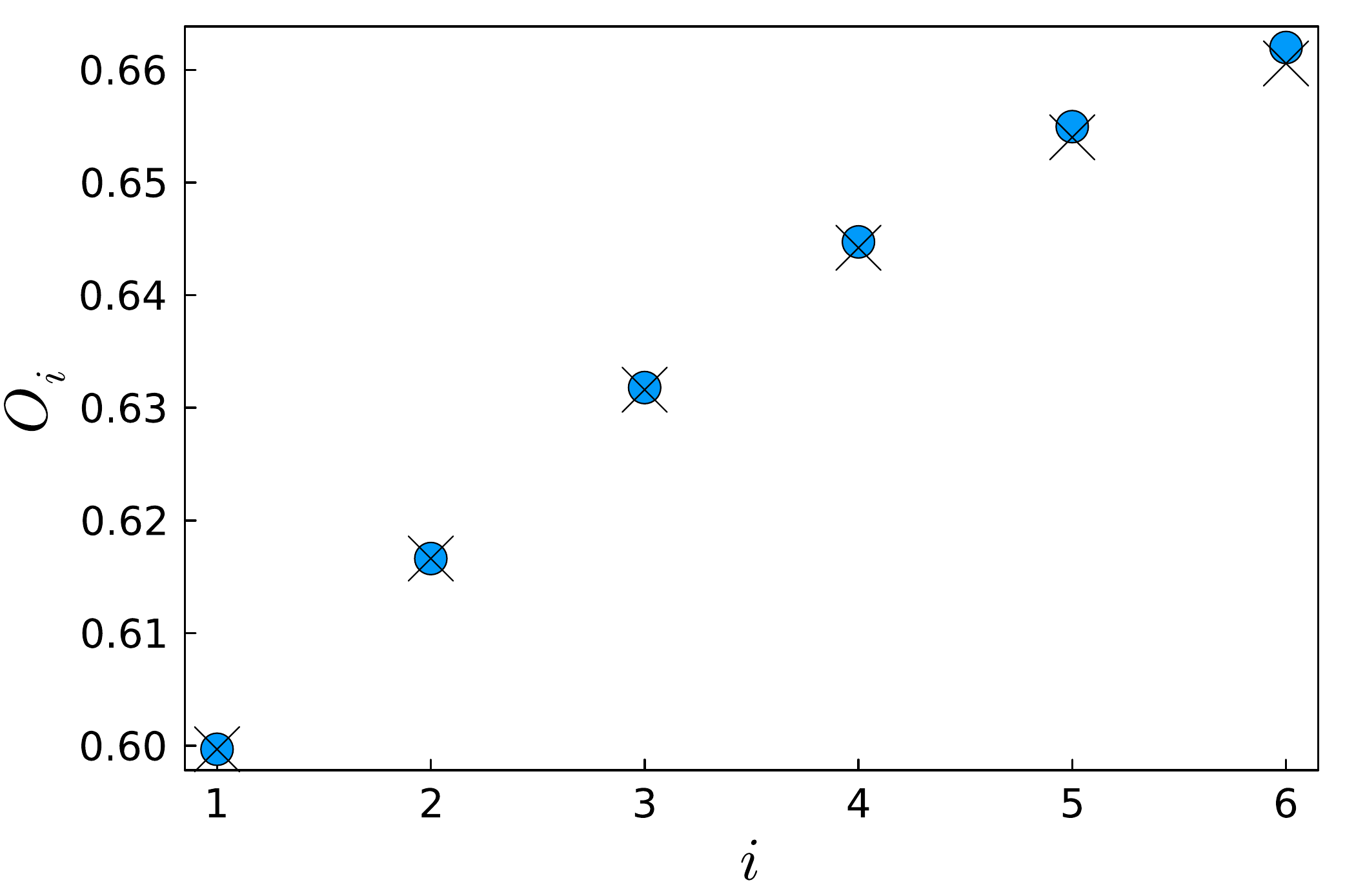}
    \includegraphics[width=7cm,clip]{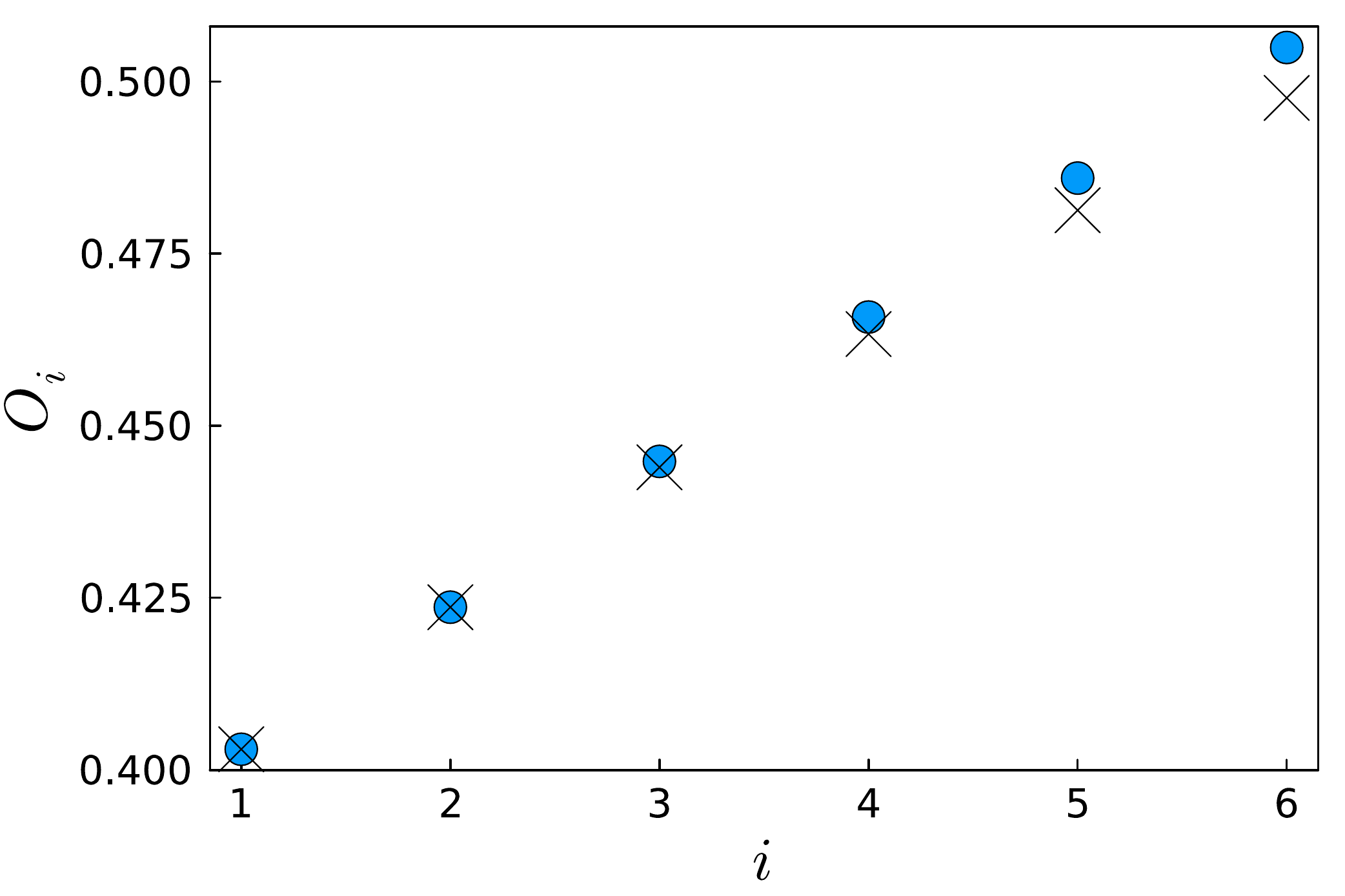}
\end{minipage}
\end{tabular}
    \caption{\rd{(Color online) The overlap centrality (blue circles) $O_i$ and the normalized eigenvector centrality (black crosses) determined by the stationary distribution obtained by the exact diagonalization as a function of agent $i$ for $L=N=6$. 
    The parameters are set as follows; $p=0.1$ (left) and $p=10.0$ (right); $\beta=1$ (top), $\beta=0$ (center), and $\beta=-1$ (bottom), respectively. The eigenvector centrality is normalized so that the two data points of $i=1$ and $i=2$ in the ranking axis are perfectly matched. 
    The overlap centrality gets closer to the eigenvector centrality for $p=0.1$ compared to the case of $p=10.0$.}}
\label{OCandEC}
\end{figure*}


\section{Concluding remarks}\label{CR}
In this paper, we have proposed a stochastic process without both detailed balance and permutation symmetry, which is inspired by the supplanting phenomenon of Japanese macaques. 
\rd{We have introduced a new type of centrality, which we call overlap centrality, to characterize a rank-dependent correlation of agents' positions. Then we have found that the new model of interacting ranked agents exhibits the unexpected linearity \eqref{eq:overlap_centrality_estimate2} with respect to agent's rank $i$ in overlap centrality $O_i$ at small supplanting condition, which could be regarded as a new type of collective phenomenon.
Precisely speaking, when $\mathcal{B}_1\neq\mathcal{B}_2$, the perfect correlation corresponding to the linearity mentioned above appears between $O_i$ and rank $i$ in the regime of weak supplanting limit $p\to+0$.}
Even for small but non-zero $p$, concrete analysis by exact diagonalization shows that $\phi$ is very close to the perfect correlation in the case of the Potts energy if $\beta$ is far from $0$.

\rd{One might ask about the meaning of perfect correlation. So far, we do not have clear answer this question. However, one might have speculation that the perfect correlation implies divergence of characteristic scale in ranking. From this viewpoint, it could be intriguing to quantify such a characteristic scale as future studies.}


Another problem on the singular behaviors of $\phi$ around $\beta = 0$ is to identify the effects which essentially cause those singular behaviors. Compared to the equilibrium Potts model, the model with supplanting process does not have permutation symmetry in terms of agents, and also not have detailed balance. 
In our derivation of perfect correlation, broken permutation symmetry is one of essential parts, but we are not aware of the effects from broken detailed balance. Concerning this point, one can consider an equilibrium model keeping with detailed balance without permutation symmetry by, for example, 
an energy function  $\sum_{i,j}J_{i,j}\delta(x_i,x_j)$, where $J_{i,j}$ is a function of agents $i,j$ such as $i\times j$. 
If one could prove that there does not exist the singularity of $\phi$ for such general equilibrium models, one could presumably expect that 
both permutation symmetry breaking and broken detailed balance are essential for causing the singularity. Such a motivation has been raised for broken rotational symmetry observed in active matter \cite{V-Tasaki}. Indeed, it has been proven for a general class of systems having potentials
dependent on position and velocity that rotational symmetry cannot be broken in equilibrium. This implies that the observed phase transitions associating with broken rotational symmetry in active matters are caused purely by nonequilibrium effects such as broken detailed balance.
We will need somewhat similar ideas.
We remark that the term \textit{permutation symmetry} has been used in this paper in various manners depending on the quantity to which the term is applied. See the list of the various ways of the term in Table \ref{tab:polymorphism}.

\begin{table*}
    \centering
    \begin{tabular}{c|c|cc}
        & Class of Quantity & 
        \multicolumn{2}{c}{
        \begin{tabular}{c}
        Definition for the quantity \\ to be \textit{permutation symmetric}
        \end{tabular}
        }
        \\ \hline
        (a) & function $f(i_1,i_2,\ldots,i_n)$ depending on ranks $i_1,i_2,\ldots,i_n$ 
        & $f(\sigma(i_1),\sigma(i_2),\ldots,\sigma(i_n))=f(i_1,i_2,\ldots,i_n)$ 
        & \multirow{5}{*}{for any $\sigma\in\mathfrak{S}_N$}
        \\ 
        (b) & function $f(\bvec{x})$  depending on agents' configuration $\bvec{x}\in X^N$ & $f(\sigma(\bvec{x}))=f(\bvec{x})$ &  \\
        (c) & state vector $\ket{\psi}\in H_X^{\otimes N}$ & $\PermOp{\sigma}\ket{\psi}=\ket{\psi}$ & \\
        (d) & linear operator $\widehat{A}$ from $H_X^{\otimes N}$ to itself & $\PermOp{\sigma}\widehat{A}\PermOp{\sigma}^{-1}=\widehat{A}$ & \\
    \end{tabular}
    \caption{
    How to use the term \textit{permutation symmetry}  depending on the classes of quantity.
    For example of each class, the case (a) is applied to  \eqref{eq:CompBOp}, (b) to \eqref{sym}, (c) to \eqref{eq:Pcan_sym}, and (d) to \eqref{eq:commuT0}.}
    \label{tab:polymorphism}
\end{table*}

Let us briefly discuss the obtained results in the context of linear response theory. If one defines a susceptibility of $M$ as $\chi(p):={\dif M}/{\dif p}$ in terms of $p$, 
we can obtain 
\begin{align}
    \chi(0)&=  \frac{N-1}{2N}(\mathcal{B}_1+\mathcal{B}_2) + \frac{(N-1)(N-2)}{2N}\mathcal{B}_3 \notag \\
    &\qquad + \partial_p C(\beta, 0) \mathcal{B}_0.
\end{align}
Note that the coefficient $\partial_p C(\beta, 0)$ of $\mathcal{B}_0$ is described as
\begin{align}\label{bfinal}
    -\binom{N}{2} \sum_{\bm{x} \in X^N} \bra{\bm{x}} \BOp{i_0,i_1}\ket{\Pcan(\beta)}
\end{align}
for some $1 \le i_0 < i_1 \le N$. By a similar computation to \eqref{eq:CompBOp}, we can find that \eqref{bfinal} does not depend on $i_0$, $i_1$, and $p$.
So far, it is not obvious for us to 
connect those quantities with the other known quantities, which remains as an open problem.

Whereas our above discussions focus on the case of small $p \ll 1$, 
we move onto the case of any $p$. The result of exact diagonalization implies that perfect correlation with $\phi=1$ would hold for negative $\beta$ sufficiently far from $0$.
This behavior is of interest in that the strong correlation effect originating from hard core repulsion may stabilize perfect correlation with $\phi=1$.
However, 
our strategy of the perturbation with respect to $p$ is unavailable for not small $p\not\ll 1$.
For this reason, it is not clear whether the perfect correlation is derived by use of permutation symmetry
in a similar way to the case of $p\ll 1$.
In order to tackle this problem, the situation with restricted values of $L$ and $N$ is to be considered. As an example, let us take the case satisfying $N=L+1$ with the Potts energy.
In this case, if the repulsion is sufficiently strong, then each site is occupied by at least an agent, and there is a single site occupied by two agents.
Then, one of the two agents occupying the site can hop in accordance with the equilibrium dynamics, while all of the other agents cannot hop effectively.
This leads to reduction of the transition of states and could help us to analyze the overlap centrality of this model.
Note that the above discussion is based on the Potts energy, and is not necessarily applied to the other case with general energy form. 
It remains for future work to perform further numerical calculation for energy functions other than the Potts energy  as well as to explore analytic methods for general $p$.

Let us briefly discuss the possibility of phase transition lines 
in parameter space $(\beta,p)$ for non-zero $p$. 
First, the equilibrium phase transition point $\beta=\beta_c$ with $p=0$ for the case of Potts energy might extend toward non-zero $p$ as 
a nonequilibrium phase transition line $\beta=\beta_c(p)$ where $M$ shows singular jump.
Second, the point $\beta=0$ with $p\to + 0$ where correlation coefficient shows singularity might also extend toward non-zero $p$ 
as another nonequilibrium phase transition line $\beta=\beta_0(p)$ where $\phi$ shows a singular jump.
However, since our analysis is limited to the 
parameter region close to $p=0$, it is necessary to perform large scale numerical simulations or develop another analytical methods in order to capture the true limiting behavior of large system size for non-zero $p$. This remains to be an intriguing future study.

From a mathematical viewpoint, {\it beta decomposition} of the transition matrix could be useful in more general context. This expression gives the exact lowest order of transition matrix in terms of $p$. Nevertheless, the concrete expression of it looks complicated and then it would not be straightforward to use the expression for a given purpose. Further, it should be noted that the minor modification of supplanting process prevents us to obtain the beta decomposition. In this sense, the present version of the supplanting process can be regarded as a specially tractable case. 
It remains as an open question how beta decomposition can be applied to calculation of the other quantities and  whether one can find other tractable cases in this direction.

Let us remark on the direct relation between eigenvector centrality and strength centrality.
Unifying two relations among  centralities discussed in Section \ref{relation_between_centralities}, one can also have another conclusion in the context of network theory.
Consider a complete graph whose edges are weighted by the same value with small fluctuation.
Then the neighbor matrix $\mathcal{R}$ whose entries are weights of edges has a decomposition $\tilde{\mathcal{R}}^{(1)} + \tilde{\mathcal{R}}^{(2)}$ similar to \eqref{eq:decomp_with_fluctuation},
and the graph holds a linear relation similar to \eqref{ov-centrality} between the strength centrality and the eigenvector centrality.

Lastly, we mention the results in this paper relevant to behaviors of members in a group of primate species.
Ranking has been widely known to be one characteristic structure which primates species has when they live as a group \cite{primatebook}.
Indeed, it has been found that ranking affects the spatial location of the members in a group 
\cite{rank_location1,rank_location2,pecentrality1,pecentrality2}
and the distance between the members \cite{rank_distance1,rank_distance2}.
The overlap centrality observed in the model proposed in this paper
might shed light on how one could estimate the ranking structures of such a group 
and its environmental condition.
\\
\begin{acknowledgements}
  The authors thank T.\ Yamaguchi for discussing his field study about Japanese macaques
  and telling us the relevant references to our work.  
\end{acknowledgements}


\onecolumngrid
\appendix


\section{Derivation of two decompositions of transition matrix}\label{MATRIX}

In this section, we would like to explain, in detail, the derivations of two decompositions of transition matrix $\TransOp{}{}$, which we call \textit{supplanting decomposition} and \textit{beta decomposition} as mentioned in Section \ref{sct:Decomp}.
To do this, we give a detailed account of our models in terms of operators.
Hereafter, in order to make this Appendix self-contained,
many definitions in the main text are repeated.

\section*{Notations}

For a set $\bm{S}$, the number of elements is denoted as $\# \bm{S}$.
The Kronecker's delta $\delta(i,j)$ is defined as
\begin{equation}
	\delta(i,j) = \begin{dcases}
		1 & (\text{if $i = j$}) \\
		0 & (\text{if $i \neq j$}). \\
	\end{dcases}
\end{equation}

We consider any vector space appearing in this section as complex vector space.
We denote the set of complex numbers as $\mathbb{C}$, 
and the set of real numbers as $\mathbb{R}$.

\sctZero{Describing the model} \label{sct:Model}
In this subsection, as a preliminary, we introduce basic concepts which are necessary to explain the model.

\sctOne{Agents and states} \label{sct:Agents}
We consider $N \ge 2$ agents labeled by $1, 2, \dots, N$.
They have a total ordering (i.e.\ linear dominance) meaning that if $1 \le i < j \le N$, agent $i$ is higher than agent $j$.
The set of all agents is written as 
\begin{equation}
	[N] \coloneqq \{1, 2, \dots, N\}.
\end{equation}

The agents lie in the lattice $X = \Z / L \Z = \{0, 1, \dots, L-1\}$ with $L \ge 3$.
For an agent $i$, the position is written as $x_i \in X$.
Thus, the configuration space of agents is described by
\begin{equation}
	X^N = \left\{ \bm{x} = (x_1, x_2, \dots, x_N) \mid x_i \in X\right\}.
\end{equation}

\sctOne{Hopping map} \label{sct:Movements}
We define a map $\Shift{i}{\dP} \colon X^N \to X^N$ (resp.\ $\Shift{i}{\dM} \colon X^N \to X^N$) of configurations as
the increment (resp.\ decrement) of position of agent $i$ under periodic boundary condition.
Explicitly, we have
\begin{align}
	\Shift{i}{\dP}(x_1, x_2, \dots, x_i, \dots, x_N) &\coloneqq 
	(x_1, x_2, \dots, x_i + 1, \dots, x_N), \\
	\Shift{i}{\dM}(x_1, x_2, \dots, x_i, \dots, x_N) &\coloneqq 
	(x_1, x_2, \dots, x_i - 1, \dots, x_N).
\end{align}

\sctOne{The permutation-invariant energy function}\label{sct:Energy}
We take an energy function $E(\bm{x}) = E(x_1, x_2, \dots, x_N)$ which is symmetric in the following sense:
\begin{equation}
    E(x_1, x_2, \dots, x_N) = E(x_{\sigma(1)}, x_{\sigma(2)}, \dots, x_{\sigma(N)})
\end{equation}
for any permutation $\sigma \in \mathfrak{S}_N$ of $N$ elements in $[N]$.
The symmetry will be used in 
Appendix \ref{sct:PermOperators} and \ref{sct:Permutation},
and is essential in 
\eqref{eq:PermCanVec} and \eqref{eq:CompBOp}.
For example, in \eqref{eq:PottsEnergy}, 
we take the {\it normalized Potts model}
\begin{align}
	E(\bm{x}) &= -\frac{(L-1)\log(L-1)} {N(L-2)} \sum_{1 \le i, j \le N} \delta(x_i, x_j)
\end{align}
as an energy function.

To simplify the notation, we introduce the $(i, d)$-th difference of the energy function $E$
for $1 \le i \le N$ and $d \in \{+, -\}$:
\begin{align}
	\DiffEnergy{i}{d} (\bm{x}) &\coloneqq E(\Shift{i}{d} \bm{x}) - E(\bm{x}).
\end{align}

\sctOne{Describing the hopping by equilibrium dynamics}\label{sct:DescWOSuppl}
First we describe the hopping by equilibrium dynamics as follows.

\begin{enumerate}[{(i)}]
	\item Determine an agent $1 \le i \le N$ with equal probability $1/N$.
	\item Choose a direction $d = \dP$ or $\dM$ with equal probability $1/2$.
	\item Decide whether the agent $i$ stays or hops to the direction $d$.
	\begin{itemize}
		\item The probability that the agent $i$ hops to direction $d$ is 
		\begin{equation}
			\frac{1}{1 + \exp(\beta \DiffEnergy{i}{d}(\bm{x}))}.
		\end{equation}
		\item 
		The probability that the agent $i$ stays is 
		\begin{equation}
			1 - \frac{1}{1 + \exp(\beta \DiffEnergy{i}{d}(\bm{x}))}= 
			\frac{1}{1 + \exp(-\beta \DiffEnergy{i}{d}(\bm{x}))}.
		\end{equation}
	\end{itemize}
\end{enumerate}

For a graphical explanation, see Fig.\ \ref{fig:Tree0}.

\begin{figure}[h]
  \centering
  \caption{Probability tree of hopping by equilibrium dynamics}
  %
  \includegraphics[clip]{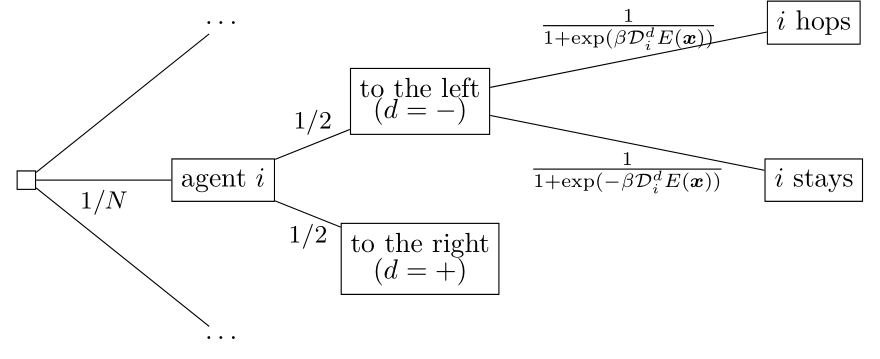}
  
  \label{fig:Tree0}
\end{figure}

\sctOne{Describing the hopping by supplanting process}\label{sct:DescSuppl}
For a given configuration $\bm{x} \in X^N$, an agent $i \in [N]$, and a direction $d \in \{+, -\}$, we define the set $S(\bm{x}, i, d) \subseteq [N]$ as follows: 
\begin{align}
	S(\bm{x}, i, \dPM ) 
	&\coloneqq \left\{ i < j \le N \mid x_j = x_i \pm 1 \right\}. 
\end{align}
In other words, $S(\bm{x}, i, d)$ is defined as a set of agents $j \in [N]$ satisfying the two conditions: (1) $j \in S(\bm{x}, i, d)$ is lower than $i$ (in other words $j >i$), and (2) sits on the site $x_i + 1$ (when $d = +$) or $x_i - 1$ (when $d=-$).

Let us fix the supplanting rate $0 \le p < \infty$. 
After the step (iii) of the first hopping,
we define the {\it supplanting process} as described below.

\begin{enumerate}[{(i)}]
	\setcounter{enumi}{3}
	\item One of the following events occurs in the probabilities described below.
	\begin{itemize}
		\item The agent $j \in S(\bm{x}, i, d)$ hops
		with the probability 
		\begin{equation}
			\frac{p}{1 + p\#S(\bm{x}, i, d)}.
		\end{equation}
		\item In a probability 
		\begin{equation}
			\frac{1}{1 + p\#S(\bm{x}, i, d)},
		\end{equation}
		no one hops.
	\end{itemize}
	\item If a hopping occurs in (iv), choose a direction $d' = \dP$ or $\dM$ of the hopping of the agent $j$ with uniform probability $1/2$.
\end{enumerate}

Note that if $p=0$ or $S(\bm{x}, i, d)$ is empty, then no supplanting occurs.
The following diagram in Fig.\ \ref{fig:Tree} describes the probability tree after the first hopping of the agent $i$.

\begin{figure}[h]
  \centering
  \caption{Probability tree of supplanting process after the first hopping}
  %
  \includegraphics[clip]{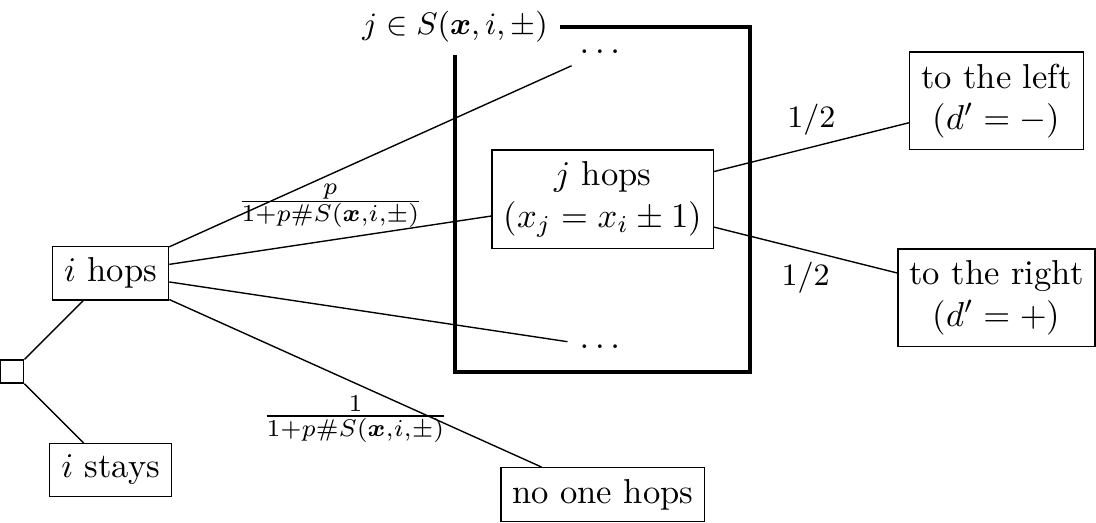}
  \label{fig:Tree}
\end{figure}

\sctZero{Transition matrices}\label{sct:TransMatr}
In this subsection, we write down the transition matrices as linear operators on a state space.

\sctOne{State vector spaces}\label{sct:Vectors}
The state space $H_X$ is the complex vector space with a basis 
\begin{equation}
	|0\rangle, |1\rangle, \dots, |{L-1}\rangle
\end{equation}
corresponding to sites $0, 1, \dots, L-1 \in X$, respectively.


Similarly the multi-state space $H_{X^N}$ is 
the complex vector space with a basis $| \bm{x} \rangle$ corresponding to configurations $\bm{x} \in X^N$.
It is identified to the tensor space $H^{\otimes N}_X$ with a map
\begin{equation}
	H_{X^N} \to H_X^{\otimes N} \coloneqq \underbrace{H_X \otimes H_X \otimes \dots \otimes H_X}_N \quad ;
	\quad |\bm{x}\rangle \mapsto |x_1\rangle \otimes |x_2\rangle \otimes \dots \otimes |x_N\rangle.
\end{equation}
With this identification, we freely use tensor notation to write down operators.

We introduce the inner product $\braket{\cdot | \cdot}$ on $H_X$ such that the basis $(|x\rangle)_{x \in X}$ is orthonormal.  This induces an inner product on $H_{X}^{\otimes N}$, and the basis $(|\bm{x} \rangle)_{\bm{x} \in X^N}$ on $H_{X}^{\otimes N}$ is orthonormal.
We use the same symbol $\langle \cdot | \cdot \rangle$ to write this induced inner product on $H_{X}^{\otimes N}$.

\sctOne{Coefficients of operators}
An operator $\AbstOp \colon H_{X}^{\otimes N} \to H_{X}^{\otimes N}$ is described by its coefficients.
The $(\bm{x}, \bm{y})$-th coefficient of $\AbstOp$ is written by 
$\langle \bm{y} | \AbstOp | \bm{x} \rangle$. 
In other words, we can write
\begin{equation}\label{eq:CoeffDescr}
	\AbstOp |\bm{x} \rangle = \sum_{\bm{y} \in X^N} |\bm{y}\rangle 
	\langle \bm{y} | \AbstOp | \bm{x} \rangle.
\end{equation}

\sctOne{Shift operators}\label{sct:ShiftOperators}
As pieces of transition matrices, we define notations of shift operators.

The shift operators $\ShiftOp{i}{\dP}$ and $\ShiftOp{i}{\dM}$ on $H_{X}^{\otimes N}$ 
for $1 \le i \le N$ are defined as
\begin{align}
	\ShiftOp{i}{\dP}|\bm{x} \rangle &\coloneqq 
	|\Shift{i}{\dP} \bm{x} \rangle \\
	&= |x_1\rangle \otimes |x_2\rangle \dots \otimes |x_i + 1\rangle \otimes \dots \otimes |x_N\rangle, \\
	\ShiftOp{i}{\dM}|\bm{x} \rangle &\coloneqq 
	|\Shift{i}{\dM}\bm{x} \rangle \\
	&= |x_1\rangle \otimes |x_2\rangle \dots \otimes |x_i - 1\rangle \otimes \dots \otimes |x_N\rangle,
\end{align}
for $\bm{x} = (x_1, x_2, \dots,  x_N) \in X^N$.
In terms of coefficients, the followings hold:
\begin{align}
	\langle \bm{y} | \ShiftOp{i}{\dP} | \bm{x} \rangle &=
	\begin{dcases}
		1 & (\text{if $\bm{y} = \Shift{i}{\dP}\bm{x}$}) \\
		0 & (\text{otherwise}),
	\end{dcases} \label{eq:PlusCoeff} \\
	\langle \bm{y} | \ShiftOp{i}{\dM} | \bm{x} \rangle &=
	\begin{dcases}
		1 & (\text{if $\bm{y} = \Shift{i}{\dM}\bm{x}$}) \\
		0 & (\text{otherwise}).
	\end{dcases} \label{eq:MinusCoeff}
\end{align}

\sctOne{Projection operators}\label{sct:ProjOperators}
We denote the identity operator on the state space $H_X$ as $\id_{H}$.
Then we can write the identity operator on the multi-state space $H_X^{\otimes N}$ as $\id_H^{\otimes N}$.

The projection operator $\CheckOp{i}{y}$ on $H_{X}^{\otimes N}$ for $1 \le i \le N$ and $y \in X$ is defined as
\begin{equation}
	\CheckOp{i}{y} |\bm{x} \rangle \coloneqq
	\begin{dcases}
		| \bm{x} \rangle & (\text{if $x_i = y$}) \\
		0 & (\text{otherwise}).
	\end{dcases}
\end{equation}
This projection operator checks whether $x_i = y$ or not; if $x_i = y$ return $|\bm{x}\rangle$, otherwise return the zero vector.
For a configuration $\bm{x} = (x_1, x_2, \dots, x_N) \in X^N$,
we define
\begin{equation}
	\CheckOp{}{\bm{x}} \coloneqq \prod_{1 \le j \le N} \CheckOp{j}{x_j}.
\end{equation}
It satisfies
\begin{equation}
	\CheckOp{}{\bm{y}} |\bm{x}\rangle = \begin{dcases}
		|\bm{x}\rangle & (\text{if $\bm{x} =\bm{y}$}) \\
		0 & (\text{otherwise}).
	\end{dcases}
\end{equation}
and, in terms of coefficients, 
\begin{align}
	\langle \bm{y} | \CheckOp{}{\bm{z}} | \bm{x} \rangle &=
	\begin{dcases}
		1 & (\text{if $\bm{y} = \bm{x} = \bm{z}$}) \\
		0 & (\text{otherwise}).
	\end{dcases} \label{eq:CheckCoeff}
\end{align}
The operator $\CheckOp{}{\bm{z}}$ checks whether $\bm{x} = \bm{z}$.

We define some other projection operators: first, for $x \in X$ and $\emptyset \neq \bm{S} \subseteq [N]$,
\begin{align}
	\CheckOp{\bm{S}}{} &\coloneqq
	\begin{dcases}
	\sum_{x \in X} \prod_{i \in \bm{S}} \CheckOp{i}{x} & (\text{if $\bm{S} \neq \emptyset$}) \\
	\id^{\otimes N}_H & (\text{if $\bm{S} = \emptyset$}).
	\end{dcases}
\end{align}
The operator $\CheckOp{\bm{S}}{}$ checks whether all of the agents in $\bm{S}$ sit on the same site. Explicitly,
\begin{align}
    \CheckOp{\bm{S}}{} |\bm{x} \rangle = 
    \begin{dcases}
		|\bm{x}\rangle & (\text{if $x_i = x_j$ for any $i, j \in \bm{S}$}) \\
		0 & (\text{otherwise}).
	\end{dcases}
\end{align}
If $\bm{S}$ is empty, it checks nothing; it is the identity operator.

Another projection operator $\StrictCheckOp{i; \bm{S}}{}$ for $1 \le i \le N$ and $\bm{S} \subseteq \{i+1, i+2, \dots, N\}$ 
checks two conditions: (1) $x_j = x_i$ if $j \in \bm{S}$, and (2) $x_{j'} \neq x_i$ if $j' \not\in \bm{S}$ and $j'>i$.
Explicitly, it is defined as
\begin{align}\label{eq:DefUpsilon}
    \StrictCheckOp{i; \bm{S}}{} | \bm{x} \rangle 
    &\coloneqq
    \begin{dcases}
        |\bm{x}\rangle & (\text{if 
        $x_j = x_i $ for $j\in \bm{S}$,
        and $x_i \not\in \{x_{j'}\}_{\substack{j' > i \\ j' \not\in \bm{S}}}$}) \\
        0 & (\text{otherwise}).
    \end{dcases}
\end{align}
In terms of $\CheckOp{i}{x}$ and $\CheckOp{\bm{S}}{}$, this operator is written as
\begin{align}
	\StrictCheckOp{i; \bm{S}}{} 
	&= \begin{dcases}
	\sum_{x \in X} \CheckOp{i}{x} \cdot 
 	\prod_{j \in \bm{S}} \CheckOp{j}{x} \cdot
 	\prod_{\substack{j' > i \\ j' \not\in \bm{S}}} (\id_{H}^{\otimes N} - \CheckOp{j'}{x}) & (\text{if $\bm{S} \neq \emptyset$ and $\bm{S} \neq \{i+1, i+2, \dots, N\}$}) \\
 	\sum_{x \in X} \CheckOp{i}{x} \cdot 
 	\prod_{\substack{j' > i}} (\id_{H}^{\otimes N} - \CheckOp{j'}{x}) & (\text{if $\bm{S} = \emptyset$ and $i \le N-1$}) \\
 	\sum_{x \in X} \CheckOp{i}{x} \cdot 
 	\prod_{j \in \bm{S}} \CheckOp{j}{x} & (\text{if $\bm{S} \neq \emptyset$ and $\bm{S} = \{i+1, i+2, \dots, N\}$}) \\
 	\sum_{x \in X} \CheckOp{i}{x} & (\text{if $\bm{S} = \emptyset$ and $i=N$})
	\end{dcases}\\
	&= \sum_{\substack{\bm{S} \subseteq \bm{S}' \subseteq \{i+1, \dots, N\}}}
	(-1)^{\# \bm{S}' - \# \bm{S}} \CheckOp{\bm{S}'}{}.
\end{align}

A relation between these operators is
\begin{align}
	\CheckOp{\{i\} \cup \bm{S}}{} &= \sum_{\substack{\bm{S} \subseteq \bm{S}' \subseteq \{i+1, \dots, N\}}} 
	\StrictCheckOp{i; \bm{S}}{}. \label{eq:SummationOfTildeDelta}
\end{align}

\sctOne{Permutation operators} \label{sct:PermOperators}
For a permutation $\sigma \in \mathfrak{S}_N$ of $[N] = \{1, 2, \dots, N\}$,
we define
\begin{equation}\label{eq:PermutationCoordinates}
    \sigma(\bm{x}) \coloneqq (x_{\sigma(1)}, x_{\sigma(2)}, \dots, x_{\sigma(n)})
\end{equation}
for $\bm{x} = (x_1, x_2, \dots, x_n) \in X^N$.
Let $\PermOp{\sigma}$ be the corresponding matrix to $\sigma^{-1}$.
Explicitly,
\begin{align}
	\PermOp{\sigma} |(x_1, x_2, \dots, x_n)\rangle &=
	|\sigma^{-1}(x_1, x_2, \dots, x_n)\rangle \\
	&=
	|(x_{\sigma^{-1}(1)}, x_{\sigma^{-1}(2)}, \dots, x_{\sigma^{-1}(N)})\rangle.
\end{align}
For an operator $\AbstOp$ on $H_X^{\otimes N}$,
the daggered symbol $\AbstOp^\dag$ denotes the Hermitian conjugate of $\AbstOp$. Then we have 
$\PermOp{\sigma}^\dagger = \PermOp{\sigma}^{-1} = \PermOp{\sigma^{-1}}$.
We also have some commutation relations as follows:
\begin{align}
	\ShiftOp{j}{\dPM} \PermOp{\sigma} &=
	\PermOp{\sigma}\ShiftOp{\sigma(j)}{\dPM}, \label{eq:PermShiftCommutation}\\
    \CheckOp{j}{x} \PermOp{\sigma} &=
	\PermOp{\sigma}\CheckOp{\sigma(j)}{x}, \\
	\CheckOp{\bm{S}}{} \PermOp{\sigma} &=
	\PermOp{\sigma} \CheckOp{\sigma(\bm{S})}{}.
\end{align}
We remark that $\StrictCheckOp{i; \bm{S}}{}$ has a factor $\prod_{\substack{j > i \\ j \not\in \bm{S}}} (\id_H^{\otimes N} - \CheckOp{x}{j}) $.
Since this factor involves an ordering $j > i$, 
it does not satisfy simple relation in terms of permutation operators.
A key procedure in our calculation is to rewrite an operator $\StrictCheckOp{i; \bm{S}}{}$ 
in terms of $\CheckOp{\{i\} \cup \bm{S}}{}$. This enables us to use the commutation relations in terms of permutations.

Besides, since the energy function depends only on the numbers of agents on each site,
the energy function does not depend on the labeling of the agents. In other terms,
\begin{align}
	E(\Shift{\sigma(i)}{d} 
	(x_{\sigma^{-1}(1)}, x_{\sigma^{-1}(2)}, \dots, x_{\sigma^{-1}(N)}) )
	%
	%
	&= E(\Shift{i}{d} (x_1,x_2, \dots, x_N) ). \label{eq:EnergyPermutation}
\end{align}

\sctOne{Other commutation relations of operators} \label{sct:Relations}
The shift operators (resp.\ the projection operators) are commutative at each other,
while a shift operator and a projection operator does not necessarily commute. Explicitly, it holds that
\begin{align}
	\CheckOp{i}{x}\ShiftOp{j}{\dP} &= 
	\begin{dcases}
		\ShiftOp{j}{\dP}\CheckOp{i}{x} & (\text{if $i \neq j$}) \\
		\ShiftOp{j}{\dP}\CheckOp{i}{x-1} & (\text{if $i = j$}),
	\end{dcases} \\
	\CheckOp{i}{x}\ShiftOp{j}{\dM} &= 
	\begin{dcases}
		\ShiftOp{j}{\dM}\CheckOp{i}{x} & (\text{if $i \neq j$}) \\
		\ShiftOp{j}{\dM}\CheckOp{i}{x+1} & (\text{if $i = j$}).
	\end{dcases}
\end{align}

\sctOne{Transition matrix for the first hopping}\label{sct:TransFirst}
According to the description in Appendix \ref{sct:DescWOSuppl},
the transition matrix $\TransOp{0}{}$ for the first hopping can be described as
\begin{equation}
	\TransOp{0}{} \coloneqq 
	\sum_{1 \le i \le N} \frac{1}{N} \sum_{d = \dPM} \frac{1}{2} \sum_{\bm{x} \in X^N}
	\Big(
		\frac{\ShiftOp{i}{d}}{1 + \exp(\beta \DiffEnergy{i}{d}(\bm{x}))} +
		\frac{\id_H^{\otimes N}}{1 + \exp(-\beta \DiffEnergy{i}{d}(\bm{x}))}
	\Big)\CheckOp{}{\bm{x}}.
\end{equation}
For the latter use, we introduce some symbols.  First, 
for $\bm{x} \in X^N$, $1 \le i \le N$, and $d \in \{+, -\}$,
we write the coefficient of $\TransOp{0}{}$ as
\begin{equation}
	c_\beta(\bm{x}, i, d) \coloneqq
	\frac{1}{2N(1 + \exp(\beta \DiffEnergy{i}{d}(\bm{x})))}.
\end{equation}
Then we have
\begin{equation}
	\frac{1}{2N} - c_\beta(\bm{x}, i, d) 
	= \frac{1}{2N(1 + \exp(-\beta \DiffEnergy{i}{d}(\bm{x})))}.
\end{equation}
Since $\exp(x) > 0$ for any real $x \in \mathbb{R}$, we can estimate
\begin{equation}\label{eq:EstimCBeta}
    0 < c_\beta(\bm{x}, i, d) \le \frac{1}{2N}
\end{equation}
for any $\beta, E, \bm{x}, i, d, N$, and $L$.
We define a set of operators: for $d \in \{+, -\}$,
\begin{align}
	\TransOp{0, \text{move}}{i, d} &\coloneqq
	\sum_{\bm{x} \in X^N}
	c_\beta(\bm{x}, i, d) \ShiftOp{i}{d} \CheckOp{}{\bm{x}}, \\
	\TransOp{0, \text{stay}}{i, d} &\coloneqq
	\sum_{\bm{x} \in X^N}
	\left( \frac{1}{2N} - c_\beta(\bm{x}, i, d)  \right) \CheckOp{}{\bm{x}}, \\
	\TransOp{0, \text{move}}{i} &\coloneqq 
	\TransOp{0, \text{move}}{i, \dP} + \TransOp{0, \text{move}}{i, \dM}, \\
	\TransOp{0, \text{stay}}{i} &\coloneqq 
	\TransOp{0, \text{stay}}{i, \dP} + \TransOp{0, \text{stay}}{i, \dM}.
\end{align}
Then we have
\begin{align}
	\TransOp{0}{} &= \sum_{1 \le i \le N} (\TransOp{0, \text{move}}{i} + 
	\TransOp{0, \text{stay}}{i}) \label{eq:T0Last1} \\
	&= \sum_{\substack{1 \le i \le N \\ d = \dPM}} (\TransOp{0, \text{move}}{i, d} + 
	\TransOp{0, \text{stay}}{i, d}). \label{eq:T0Last2}
\end{align}
Note that, using relations in Appendix \ref{sct:PermOperators}, we have
\begin{align}
	\PermOp{\sigma}^\dag \TransOp{0, \text{move}}{i, d} \PermOp{\sigma} 
	&= \TransOp{0, \text{move}}{\sigma(i), d}, 
	\label{eq:T0Permutation} \\
	\PermOp{\sigma}^\dag \TransOp{0, \text{stay}}{i, d} \PermOp{\sigma} 
	&= \TransOp{0, \text{stay}}{\sigma(i), d}, \\
	\PermOp{\sigma}^\dag \TransOp{0}{} \PermOp{\sigma} 
	&= \TransOp{0}{}.
	\label{eq:TotalT0Permutation}
\end{align}

\sctOne{Transition matrix for supplanting process}\label{sct:TransSuppl}
According to Appendix \ref{sct:DescSuppl},
in order to describe the transition matrix $\TransOp{}{}$ for the supplanting process,
it suffices to replace $\ShiftOp{i}{d} \CheckOp{}{\bm{x}}$ appearing in $\TransOp{{0, \text{move}}}{i}$ by
\begin{align}
	&\sum_{j \in S(\bm{x}, i, d)} \frac{1}{\# S(\bm{x}, i, d)} 
	\sum_{d' = \dPM} \frac{1}{2}
	\Big(
		\frac{p \# S(\bm{x}, i, d)}{1 + p \#S(\bm{x}, i, d)} \ShiftOp{j}{d'}
		+ \frac{1}{1 + p \#S(\bm{x}, i, d)} \id_H^{\otimes N}
	\Big) \ShiftOp{i}{d} \CheckOp{}{\bm{x}} \\
	= &\sum_{\substack{j \in S(\bm{x}, i, d) \\ d' = \dPM}} 
	\frac{1}{2\# S(\bm{x}, i, d)} 
	\Big(
		\frac{p \# S(\bm{x}, i, d)}{1 + p \#S(\bm{x}, i, d)} \ShiftOp{j}{d'}
		+ \frac{1}{1 + p \#S(\bm{x}, i, d)} \id_H^{\otimes N}
	\Big) \ShiftOp{i}{d} \CheckOp{}{\bm{x}} \\
	= &\Bigg( \ShiftOp{i}{d} + \sum_{\substack{j \in S(\bm{x}, i, d) \\ d' = \dPM}} 
	\frac{p}{2(1 + p \#S(\bm{x}, i, d))} (\ShiftOp{j}{d'} - \id^{\otimes N}_H) 
	\ShiftOp{i}{d} \Bigg) \CheckOp{}{\bm{x}}.
\end{align}
Thus, by \eqref{eq:T0Last2}, the transition matrix $\TransOp{}{}$ is written as
\begin{align}
	\TransOp{}{} &= 
	\sum_{\substack{1 \le i \le N \\ d = \dPM}} \sum_{\bm{x} \in X^N} 
	\Big[ c_\beta(\bm{x}, i, d)
	\Big( 
	\ShiftOp{i}{d}
	+ \sum_{\substack{j \in S(\bm{x}, i, d) \\ d' = \dPM}} 
	\frac{p}{2(1 + p \#S(\bm{x}, i, d))} (\ShiftOp{j}{d'} - \id^{\otimes N}_H) \ShiftOp{i}{d}
	\Big) \CheckOp{}{\bm{x}} + \TransOp{0, \text{stay}}{i,d} 
	\Big] \\
	&= 
	\TransOp{0}{} + \sum_{\bm{x} \in X^N} 
	\sum_{\substack{1 \le i \le N \\ d = \dPM}} 
	\sum_{\substack{j \in S(\bm{x}, i, d) \\ d' = \dPM}} 
	\frac{p c_\beta(\bm{x}, i, d)}{2(1 + p \#S(\bm{x}, i, d))} 
	(\ShiftOp{j}{d'} - \id^{\otimes N}_H) 
	\ShiftOp{i}{d} \CheckOp{}{\bm{x}}.
	\label{eq:TFinal}
\end{align}

\sctOne{Coefficients of two operators} \label{sct:CoeffSupplMoves}
From the descriptions of operators $\ShiftOp{i}{d}, \CheckOp{}{\bm{x}}$
in \eqref{eq:PlusCoeff}, \eqref{eq:MinusCoeff}, and \eqref{eq:CheckCoeff}, we obtain
\begin{align}
	\langle \bm{y} | \ShiftOp{i}{d} \CheckOp{}{\bm{z}} | \bm{x} \rangle &=
	\begin{dcases}
		1 & (\text{if $\bm{x} = \bm{z}$ and $\bm{y} = \Shift{i}{d} \bm{z}$}) \\
		0 & (\text{otherwise}),
	\end{dcases} \label{eq:CoeffOfCompositions1} \\
	\langle \bm{y} | \ShiftOp{j}{d'} \ShiftOp{i}{d} \CheckOp{}{\bm{z}} | \bm{x} \rangle &=
	\begin{dcases}
		1 & (\text{if $\bm{x} = \bm{z}$ and $\bm{y} = \Shift{j}{d'} \Shift{i}{d} \bm{z}$}) \\
		0 & (\text{otherwise}).
	\end{dcases} \label{eq:CoeffOfCompositions2}
\end{align}
Using this and \eqref{eq:T0Last2}, the coefficients of $\TransOp{0}{}$ is
\begin{align}
	\langle \bm{y} | \TransOp{0}{} | \bm{x} \rangle &=
	\begin{dcases}
		c_\beta(\bm{x}, i, d) & (\text{if $\bm{y} = \Shift{i}{d}\bm{x}$}) \\
		1-\sum_{\substack{1 \le i \le N \\ d =\dPM}}c_\beta(\bm{x}, i, d) & (\text{if $\bm{y} = \bm{x}$}) \\
		0 & (\text{otherwise}).
	\end{dcases}
\end{align}
Similarly by \eqref{eq:TFinal}, the coefficients of $\TransOp{}{}$ is
\begin{align}
	\langle \bm{y} | \TransOp{}{} | \bm{x} \rangle &=
	\begin{dcases}
		\frac{pc_\beta(\bm{x}, i, d)}{2(1+np)}
		& (\text{if $\bm{y} = \Shift{j}{d'} \Shift{i}{d} \bm{x}$ and $j \in S(\bm{x}, i, d)$}) \\
		\frac{2pc_\beta(\bm{x}, i, d)}{2(1+np)} 
		& (\text{if $\bm{y} = \Shift{i}{d} \bm{x}$}) \\
		1-\sum_{\substack{1 \le i \le N \\ d =\dPM}}c_\beta(\bm{x}, i, d) & (\text{if $\bm{y} = \bm{x}$}) \\
		0 & (\text{otherwise}).
	\end{dcases}
	\label{eq:CoeffOfTOp}
\end{align}

\sctOne{The broken permutation symmetry of transition matrix} \label{sct:BrokenSymmetry}
In this subsection, we explain an example of the broken permutation symmetry of transition matrix $\TransOp{}{}$ by using commutative relations.
Recall that, for an operator $\AbstOp$ on $H_X^{\otimes N}$,
the daggered symbol $\AbstOp^\dag$ denotes the Hermitian conjugate of $\AbstOp$.

By \eqref{eq:CoeffOfTOp}, for $i \in [N]$ and $j \in S(\bm{x}, i, d)$, we have
\begin{equation}
    \langle \Shift{j}{d'} \Shift{i}{d} \bm{x} | \TransOp{}{} | \bm{x} \rangle = \frac{pc_\beta(\bm{x}, i, d)}{2(1+ p \#S(\bm{x}, i, d))}.
\end{equation}
On the other hand, for a transposition $\sigma = (i \; j) \in \mathfrak{S}_N$, one has
\begin{align}
    \langle \Shift{j}{d'} \Shift{i}{d} \bm{x} | \PermOp{\sigma}^\dag \TransOp{}{} \PermOp{\sigma} | \bm{x} \rangle &= 
    \langle \bm{x} | \ShiftOp{i}{d, \dag} \ShiftOp{j}{d', \dag} \PermOp{\sigma}^\dag \TransOp{}{} \PermOp{\sigma} | \bm{x} \rangle
    \\
    &= \langle \bm{x} | (\PermOp{\sigma} \ShiftOp{j}{d'}\ShiftOp{i}{d})^\dag \TransOp{}{} \PermOp{\sigma} | \bm{x} \rangle.
\end{align} 
Using the commutative relation \eqref{eq:PermShiftCommutation}, one can see
\begin{align}
    \langle \bm{x} | (\PermOp{\sigma} \ShiftOp{j}{d'}\ShiftOp{i}{d})^\dag \TransOp{}{} \PermOp{\sigma} | \bm{x} \rangle 
    &= \langle \bm{x} | ( \ShiftOp{\sigma(j)}{d'}\ShiftOp{\sigma(i)}{d} \PermOp{\sigma})^\dag \TransOp{}{} \PermOp{\sigma} | \bm{x} \rangle \\
    &= \langle \bm{x} | ( \ShiftOp{i}{d'}\ShiftOp{j}{d} \PermOp{\sigma})^\dag \TransOp{}{} \PermOp{\sigma} | \bm{x} \rangle \\
    &= \langle \bm{x} | \PermOp{\sigma}^\dag \ShiftOp{j}{d, \dag} \ShiftOp{i}{d', \dag} \TransOp{}{} \PermOp{\sigma} | \bm{x} \rangle.
\end{align} 
Proceeding with the transformation using \eqref{eq:PermutationCoordinates}, it follows that
\begin{align}
    \langle \bm{x} | \PermOp{\sigma}^\dag \ShiftOp{j}{d, \dag} \ShiftOp{i}{d', \dag} \TransOp{}{} \PermOp{\sigma} | \bm{x} \rangle
    &= \langle \sigma^{-1}(\bm{x}) | \ShiftOp{j}{d, \dag} \ShiftOp{i}{d', \dag} \TransOp{}{} | \sigma^{-1}(\bm{x}) \rangle \\
    &= \langle \Shift{i}{d'} \Shift{j}{d} \sigma^{-1}(\bm{x}) |
    \TransOp{}{} | \sigma^{-1}(\bm{x}) \rangle.
\end{align}
Since $j \in S(\bm{x}, i, d)$, we have $j > i$,
hence $i \not\in S(\bm{x}, j, d)$.
With \eqref{eq:CoeffOfTOp}, we obtain
\begin{equation}
    \langle \Shift{i}{d'} \Shift{j}{d} \sigma^{-1}(\bm{x}) |
    \TransOp{}{} | \sigma^{-1}(\bm{x}) \rangle = 0.
\end{equation}
Thus, since $c_{\beta}(\bm{x}, i, d) \neq 0$, 
we obtain the broken permutation symmetry of $\TransOp{}{}$ with respect to the permutation of components: 
\begin{equation}\label{eq:BrokenSymmetryOfTOp}
    \langle \Shift{j}{d'} \Shift{i}{d} \bm{x} | \PermOp{\sigma}^\dag \TransOp{}{} \PermOp{\sigma} | \bm{x} \rangle \neq
    \langle \Shift{j}{d'} \Shift{i}{d} \bm{x} | \TransOp{}{} | \bm{x} \rangle.
\end{equation}
This shows that, for any energy function $E(\bm{x})$ as in Appendix \ref{sct:Energy}
and any transposition $\sigma = (i \; j) \in \mathfrak{S}_N$, we find that \begin{equation}\label{eq:NotCommutesTransOp}
\PermOp{\sigma}^\dag \TransOp{}{} \PermOp{\sigma} \neq \TransOp{}{}.
\end{equation}

In fact, for any non-trivial permutation $\sigma \in \mathfrak{S}_N$, we can show \eqref{eq:NotCommutesTransOp}.
Explicitly, for a non-trivial permutation $\sigma \in \mathfrak{S}_N$, we can find the following four data $(i,j, \bm{x}, d)$ satisfying the following two conditions: (1) $i < j$ with $\sigma(i) > \sigma(j)$ and (2) $\bm{x} \in X^N$ and $d = +$ or $-$ such that $j \in S(\bm{x}, i, d)$.
Then, in parallel, the above argument works to show \eqref{eq:BrokenSymmetryOfTOp} and thus \eqref{eq:NotCommutesTransOp} for the permutation $\sigma$.

\sctZero{The supplanting decompotision}\label{sct:SDecomp}
From here to the next section, we introduce two decompositions of the operator $\TransOp{}{}$.
The idea of our first decomposition of $\TransOp{}{}$ is
to split the sums by the number $\# S(\bm{x}, i, d)$.

\sctOne{The $n$-th term} \label{sct:SthTerm}
First we note that a coefficient in \eqref{eq:TFinal},
\begin{equation}
	\frac{pc_\beta(\bm{x}, i, d)}{2(1 + p \#S(\bm{x}, i, d))},
\end{equation}
does not depend on $j \in S(\bm{x}, i, d)$ and $d' \in \{+, -\}$.
Reminding of \eqref{eq:TFinal}, equivalently,
\begin{equation}
	\TransOp{}{} - \TransOp{0}{} = 
	\sum_{\bm{x} \in X^N} \sum_{\substack{1 \le i \le N \\ d =\dPM}} 
	\sum_{\substack{j \in S(\bm{x}, i, d) \\ d' = \dPM}} 
	\frac{p c_\beta(\bm{x}, i, d)}{2(1 + p \#S(\bm{x}, i, d))} 
	(\ShiftOp{j}{d'} - \id^{\otimes N}_H)\ShiftOp{i}{d} \CheckOp{}{\bm{x}},
\end{equation}
we define the operators
\begin{equation}\label{eq:DefnthTerm}
	\TransOp{n}{} \coloneqq \sum_{\substack{\bm{x}, i, d}} 
	\delta(\# S(\bm{x}, i, d), n)
	\sum_{\substack{j \in S(\bm{x}, i, d) \\ d' = \dPM}} 
	\frac{p c_\beta(\bm{x}, i, d)}{2(1 + np)} 
	(\ShiftOp{j}{d'} - \id^{\otimes N}_H)\ShiftOp{i}{d} \CheckOp{}{\bm{x}},
\end{equation}
where $1 \le n \le N-1$. Then we obtain \begin{equation}\label{eq:SDecomp}
	\TransOp{}{} = \TransOp{0}{} + \TransOp{1}{} + \TransOp{2}{} + \dots + \TransOp{N-1}{}.
\end{equation}

\sctOne{Coefficients of $n$-th terms}\label{sct:nthCoeff}
Next we describe coefficients of $\TransOp{n}{}$ as in \eqref{eq:CoeffDescr}.
Using \eqref{eq:CoeffOfCompositions1} and \eqref{eq:CoeffOfCompositions2}, we have
\begin{align}
	\langle \bm{y} | \TransOp{n}{} | \bm{x} \rangle &=
	\begin{dcases}
		\frac{p c_\beta(\bm{x}, i, d)}{2(1 + np)} 
		& (\text{if $\bm{y} = \Shift{j}{d'}\Shift{i}{d} \bm{x}$ 
		and $\#S(\bm{x}, i, d) = n$}) \\
		\frac{-2np c_\beta(\bm{x}, i, d)}{2(1 + np)} 
		& (\text{if $\bm{y} = \Shift{i}{d} \bm{x}$ 
		and $\#S(\bm{x}, i, d) = n$}) \\
		0 & (\text{otherwise}).
	\end{dcases} \label{eq:nthCoeff}
\end{align}
By the estimate \eqref{eq:EstimCBeta}, using $n < N$ and $1+np \ge 1$, we see that
\begin{equation}
    0 \le |\langle \bm{y} | \TransOp{n}{} | \bm{x} \rangle|
    \le \frac{np}{2N(1+np)} < \frac{p}{2}.
\end{equation}
This shows that the nonzero coefficients of $\TransOp{n}{}$ are estimated as $\mathcal{O}(p)$ as $p \to +0$, independent of $n$.
See Table \ref{table:SDecompCoeff} for the coefficients 
$\langle \bm{y} | \TransOp{n}{} | \bm{x} \rangle$
if $\bm{y} = \Shift{j}{d'}\Shift{i}{d} \bm{x}$.

\begin{table}[h]
\renewcommand{\arraystretch}{1.5}
\caption{A part of $(\bm{x}, \bm{y})$-th coefficients of $\TransOp{}{} - \TransOp{0}{}$ 
and $\TransOp{n}{}$ with $\bm{y} = \Shift{j}{d'}\Shift{i}{d} \bm{x}$.}
\centering
\begin{tabular}{c||c||c|c|c|c|c|c}
	$\#S(\bm{x}, i, d)$ & $\TransOp{}{} - \TransOp{0}{}$ 
	& $\TransOp{1}{}$ & $\TransOp{2}{}$ & $\TransOp{3}{}$ & $\TransOp{4}{}$ 
	& $\cdots$ & $\TransOp{n}{}$ \\
	\hline
	$1$ & $\frac{pc_\beta}{2(1+p)}$ & $\frac{pc_\beta}{2(1+p)}$ 
	& $0$ & $0$ & $0$ & $\cdots$ & $0$ \\
	\hline
	$2$ & $\frac{pc_\beta}{2(1+2p)}$ & $0$ 
	& $\frac{pc_\beta}{2(1+2p)}$ & $0$ & $0$ & $\cdots$ & $0$ \\
	\hline
	$3$ & $\frac{pc_\beta}{2(1+3p)}$ & $0$ & $0$ 
	& $\frac{pc_\beta}{2(1+3p)}$ & $0$ & $\cdots$ & $0$ \\
	\hline
	$4$ & $\frac{pc_\beta}{2(1+4p)}$ & $0$ & $0$ 
	& $0$ & $\frac{pc_\beta}{2(1+4p)}$ & $\cdots$ & $0$ \\
	\hline
	$n$ & $\frac{pc_\beta}{2(1+np)}$ & $0$ & $0$ 
	& $0$ & $0$ & $\cdots$ & $\frac{pc_\beta}{2(1+np)}$
\end{tabular}
\label{table:SDecompCoeff}
\end{table}

\sctOne{Another description of $n$-th terms}\label{sct:AnotherDescr}
Here we give another description {\eqref{eq:AnotherTn}} of $n$-th terms.
We change the ordering of summations in \eqref{eq:DefnthTerm} as follows:
we first choose $1 \le i_0 \le N$,
an $n$-elements set $\bm{S} = \{i_1 < \dots < i_{n}\} \subset [N]$
with $i_0 <i_1$,
and then configurations $\bm{x} = (x_1, x_2, \dots, x_N)$ so that $S(\bm{x}, i_0, d) = \bm{S}$.
Then explicitly we reach
\begingroup
 \allowdisplaybreaks
\begin{align}
	\TransOp{n}{} &= \sum_{\substack{1 \le i_0 \le N}}
	\sum_{\bm{S} = \{i_0 < i_1 < \dots < i_{n} \le N\}} 
	\sum_{\substack{\bm{x} \\ S(\bm{x}, i_0, d) = \bm{S} \\ d = \dPM}}
	\frac{p c_\beta(\bm{x}, i_0, d)}{2(1 + np)} 
	\sum_{\substack{j \in {S}(\bm{x}, i_0, d) \\ d' = \dPM}} 
	(\ShiftOp{j}{d'} - \id^{\otimes N}_H)\ShiftOp{i_0}{d} \CheckOp{}{\bm{x}} \\
	&= \frac{p}{2(1 + np)} \sum_{\substack{1 \le i_0 \le N}}
	\sum_{\bm{S}} 
	\sum_{\substack{\bm{x} \\ S(\bm{x}, i_0, d) = \bm{S}  \\ d = \dPM}}
	c_\beta(\bm{x}, i_0, d)
	\Bigg[\sum_{\substack{j \in \bm{S} \\ d' = \dPM}} 
	(\ShiftOp{j}{d'} - \id^{\otimes N}_H) \Bigg]
	\ShiftOp{i_0}{d} \CheckOp{}{\bm{x}} \\
	&= \frac{p}{2(1 + np)} \sum_{\substack{1 \le i_0 \le N}} \sum_{\bm{S}}
	\Bigg[\sum_{\substack{j \in \bm{S} \\ d' = \dPM}} 
	(\ShiftOp{j}{d'} - \id^{\otimes N}_H) \Bigg]
	\sum_{\substack{\bm{x} \\ S(\bm{x}, i_0, d) = \bm{S}  \\ d = \dPM}}
	c_\beta(\bm{x}, i_0, d)
	\ShiftOp{i_0}{d} \CheckOp{}{\bm{x}}.
\end{align}\endgroup

Using $\StrictCheckOp{i_0; \bm{S}}{}$ defined in \eqref{eq:DefUpsilon}, we can write
\begin{align}
	\sum_{\substack{\bm{x} \\ S(\bm{x}, i_0, d) = \bm{S}  \\ d = \dPM}}
	c_\beta(\bm{x}, i_0, d)
	\ShiftOp{i_0}{d} \CheckOp{}{\bm{x}} &= 
	\StrictCheckOp{i_0; \bm{S}}{}
	\sum_{d = \dPM} \sum_{\substack{\bm{x}}}
	c_\beta(\bm{x}, i_0, d) \ShiftOp{i_0}{d} \CheckOp{}{\bm{x}} \\
	&= \StrictCheckOp{i_0; \bm{S}}{} \TransOp{0, \text{move}}{{i}_0}.
\end{align}

By putting $\SubOp{\emptyset} = 0$ and  
\begin{align}
	\SubOp{\bm{S}} &\coloneqq \sum_{\substack{j \in \bm{S} \\ d' = \dPM}} 
	(\ShiftOp{j}{d'} - \id^{\otimes N}_H) \\
	&= \sum_{\substack{j \in \bm{S}}} 
	\left( \ShiftOp{j}{\dP} + \ShiftOp{j}{\dM} - 2\id^{\otimes N}_H \right), \label{eq:SubOp}
\end{align}
we can write
\begin{equation} \label{eq:AnotherTn}
	\TransOp{n}{} = \frac{p}{2(1+np)}\sum_{\substack{1 \le i_0 \le N}}
	\sum_{\substack{\bm{S} \subseteq \{i_0 + 1, \dots, N\}}} \delta(\# \bm{S}, n)
	\SubOp{\bm{S}} \StrictCheckOp{i_0; \bm{S}}{} \TransOp{0, \text{move}}{{i}_0}.
\end{equation}

\sctZero{The beta decomposition}\label{sct:BetaDecomp}
In this section, we introduce the second decomposition \eqref{eq:BetaDecomp} of 
$\TransOp{}{}$
which we call the \textit{beta decomposition}.
The coefficients of $m$-th term in the beta decomposition is estimated as $\mathcal{O}(p^{m})$ 
(see \eqref{eq:BetaDecompEstimate}).

\sctOne{Definition of beta terms}\label{sct:BetaTerms}
Here we give a definition of $m$-th term $\BetaTerm{m}$ of beta decomposition
by using $\TransOp{n}{}$ which appear in the supplanting decomposition (see \eqref{eq:DefnthTerm}).
For $1 \le m \le N-1$, we define
\begin{align} \label{eq:DefBetaTerms}
	\BetaTerm{m} &\coloneqq \frac{(-1)^{m-1} B(m, 1+1/p)}{p} 
	\sum_{m \le n \le N-1} \binom{n-1}{m-1} (1+np) \TransOp{n}{}.
\end{align}
Here $B(a,b)$ is the beta function:
\begin{align}
	B(a, b) &\coloneqq \frac{\Gamma(a) \Gamma(b)}{\Gamma(a+b)}\\
	&= \int_{0}^{1} x^{a-1} (1-x)^{b-1} dx. 
\end{align}
{In particular, we have
\begin{align}
    B \left( m, 1+\frac{1}{p} \right) &= 
    \frac{\Gamma(m) \Gamma(1+1/p)}{\Gamma(m+1+1/p)} \\
    &= \frac{(m-1)!}{(1+1/p)(2+1/p) \dots (m+1/p)} \\
    &= \frac{(m-1)!p^{m}}{(1+p)(1+2p) \dots (1+mp)}.
\end{align}}
By these descriptions, one can write
\begin{align}
	\BetaTerm{m} 
	&= \frac{(-1)^{m-1} (m-1)! p^{m-1}}{(1+p)(1+2p) \dots (1+mp)} 
	\sum_{m \le n \le N-1} \binom{n-1}{m-1} (1+np)\TransOp{n}{}. \label{eq:AnotherBetaTerm}
\end{align}

\sctOne{The beta decomposition}\label{sct:PrfBetaDecomp}
We can prove the following decomposition:
\begin{equation} \label{eq:BetaDecomp}
	\TransOp{}{} = \TransOp{0}{} + \sum_{1 \le m \le N-1} \BetaTerm{m}.
\end{equation}
In fact,
\begin{align}
	\sum_{1 \le m \le N-1} \BetaTerm{m} 
	&= \sum_{1 \le m \le N-1} \frac{(-1)^{m-1} B(1+1/p, m)}{p}
	\sum_{m \le n \le N-1} \binom{n-1}{m-1} (1+np) \TransOp{n}{} \\
	&= \sum_{1 \le n \le N-1} \left[ 
	\sum_{1 \le m \le n} (-1)^{m-1} \frac{B(1+1/p, m)}{p} \binom{n-1}{m-1} \right] 
	(1+np) \TransOp{n}{}. 
\end{align}
Then the following lemma is enough to show the desired decomposition \eqref{eq:BetaDecomp}.

\begin{lem}\label{lem:num}
	\begin{equation}
		\sum_{1 \le m \le n} \frac{(-1)^{m-1}B(1+1/p, m)}{p} \binom{n-1}{m-1} = \frac{1}{1 + np}.
	\end{equation}
\end{lem}
\begin{proof}[Sketch of proof] We only sketch the proof. With two generating series of exponential type
\begin{align}
	F_0(t) &\coloneqq \sum_{m \ge 0} a_m \frac{t^m}{m!}, \\
	F_1(t) &\coloneqq F_0(t) e^{t}, \label{eq:Fsum}
\end{align}
where the coefficients of $F_0(t)$ is
\begin{align}
	a_m &\coloneqq \frac{(-1)^{m}B(1+1/p, m+1)}{p} \\
	&= \frac{(-1)^m m! p^{m}}{(1+p)(1+2p)\dots (1+(m+1)p)},
\end{align}
it is enough to prove
\begin{equation} \label{eq:CoeffOfF}
	F_1(t) = \sum_{r \ge 0} \frac{1}{1 + (r+1)p} \frac{t^{r}}{r!}.
\end{equation}
Since $a_0 = F_0(0) = {(1+p)^{-1}}$ and 
\begin{align}
	&(1+(m+2)p) a_{m+1} = (-1)\cdot (m+1) p a_{m} \\
	\iff &(m+1) a_{m+1} + (m+1) a_m = -\left( 1+ \frac{1}{p} \right) a_{m+1}
\end{align}
for $m \ge 0$, we have a differential equation for $F_0(t)$
\begin{align}
	t \left(\frac{dF_0}{dt}(t) + F_0(t) \right) &= -\frac{1}{p} ({(1+p)}F_0(t) - 1).
\end{align}
Thus $F_1(t) = F_0(t) e^{t}$ satisfies
\begin{align}
	t\frac{dF_1}{dt}(t) &= - {\left(1 + \frac{1}{p}\right)} F_1(t) + \frac{1}{p} e^t.
\end{align}
Now the equation \eqref{eq:CoeffOfF} can be proved by comparison of coefficients.
\end{proof}

\sctOne{Coefficients of beta terms}\label{sct:CoeffBetaTerms}
With our description \eqref{eq:nthCoeff} of coefficients of $\TransOp{n}{}$,
we can write down the coefficients of $\BetaTerm{m}$.

For $\bm{y} = \Shift{j}{d'} \Shift{i}{d} \bm{x}$ or
$\Shift{i}{d} \bm{x}$,
the coefficients $\langle \bm{y} | \TransOp{n}{} | \bm{x} \rangle$ of $\TransOp{n}{}$ is
zero when $n \neq \#S(\bm{x}, i, d)$. Hence we have
\begin{align}
	\langle \bm{y} | \BetaTerm{m} | \bm{x} \rangle &= 
	\begin{dcases}
		\binom{n-1}{m-1} \frac{(-1)^{m+1} B(1+1/p, m)}{2} c_\beta(\bm{x}, i, d)
		& (\text{if $\bm{y} = \Shift{j}{d'} \Shift{i}{d} \bm{x}$ 
		and $\#S(\bm{x}, i, d) = n \ge m$}) \\
		-2n \binom{n-1}{m-1} \frac{(-1)^{m+1} B(1+1/p, m)}{2} c_\beta(\bm{x}, i, d)
	    & (\text{if $\bm{y} = \Shift{i}{d} \bm{x}$ 
		and $\#S(\bm{x}, i, d) = n \ge m$}) \\
		0 & (\text{otherwise}).
	\end{dcases}
\end{align}
Using the description
\begin{equation}
	B(m, 1+1/p) =
	\frac{(m-1)! p^{m}}{(1+p)(1+2p) \dots (1+mp)},
\end{equation}
we can write those coefficients in another form:
\begin{align}
	\langle \bm{y} | \BetaTerm{m} | \bm{x} \rangle &=
	\begin{dcases}
		\binom{n-1}{m-1} \frac{(-1)^{m+1} (m-1)! p^{m}}{2(1+p)(1+2p) \dots (1+mp)} 
		c_\beta(\bm{x}, i, d) &
    	(\text{if $\bm{y} = \Shift{j}{d'} \Shift{i}{d} \bm{x}$ 
		and $\#S(\bm{x}, i, d) = n \ge m$}) \\
		-2n \binom{n-1}{m-1} \frac{(-1)^{m+1} (m-1)! p^{m}}{2(1+p)(1+2p) \dots (1+mp)} 
		c_\beta(\bm{x}, i, d) &
		(\text{if $\bm{y} = \Shift{i}{d} \bm{x}$ 
		and $\#S(\bm{x}, i, d) = n \ge m$}) \\
		0 & (\text{otherwise}).
	\end{dcases} \label{eq:CoeffBetaTerms}
\end{align}
This allows us to estimate the coefficients of $\BetaTerm{m}$.
With $p \to 0$ and fixing other parameters $\beta, N, L$, and $E$, we have
\begin{equation}
	\langle \bm{y} | \BetaTerm{m} | \bm{x} \rangle = \mathcal{O}(p^{m}).
\end{equation}
In particular, we have
\begin{equation} \label{eq:BetaDecompEstimate}
	\TransOp{}{} = \TransOp{0}{} 
	+ \BetaTerm{1} + \dots + \BetaTerm{m} + \mathcal{O}(p^{m+1})
\end{equation}
for $1 \le m \le N-1$.
See Table \ref{table:BetaDecompCoeff} for the coefficients 
$\langle \bm{y} | \BetaTerm{m} | \bm{x} \rangle$
if $\bm{y} = \Shift{j}{d'} \Shift{i}{d} \bm{x}$.

\begin{table}[h]
\renewcommand{\arraystretch}{1.8}
\caption{A part of $(\bm{x}, \bm{y})$-th coefficients of $\TransOp{}{} - \TransOp{0}{}$ and 
$\BetaTerm{m}$ with $\bm{y} = \Shift{j}{d'} \Shift{i}{d} \bm{x}$.}
\centering
\begin{tabular}{c||c||c|c|c|c|c|c}
	$\#S(\bm{x}, i, d)$ & $\TransOp{}{} - \TransOp{0}{}$ 
	& $\BetaTerm{1}$ & $\BetaTerm{2}$ & $\BetaTerm{3}$ & $\BetaTerm{4}$ & $\cdots$ & $\BetaTerm{m}$ \\
	\hline
	$1$ & $\frac{pc_\beta}{2(1+p)}$ 
	& $\binom{0}{0}\frac{pc_\beta}{2(1+p)}$ & $0$ & $0$ & $0$ & $\cdots$ & $0$  \\
	\hline
	$2$ & $\frac{pc_\beta}{2(1+2p)}$ 
	& $\binom{1}{0}\frac{pc_\beta}{2(1+p)}$ & $\binom{1}{1} \frac{(-1) 1!p^2c_\beta}{2(1+p)(1+2p)}$ 
	& $0$ & $0$ & $\cdots$ & $0$ \\
	\hline
	$3$ & $\frac{pc_\beta}{2(1+3p)}$ 
	& $\binom{2}{0}\frac{pc_\beta}{2(1+p)}$ & $\binom{2}{1}\frac{(-1)1!p^2c_\beta}{2(1+p)(1+2p)}$ & 
	$\binom{2}{2} \frac{(-1)^22!p^3c_\beta}{2(1+p)\dots(1+3p)}$ & $0$ & $\cdots$ & $0$ \\
	\hline
	$4$ & $\frac{pc_\beta}{2(1+4p)}$ 
	& $\binom{3}{0}\frac{pc_\beta}{2(1+p)}$ & $\binom{3}{1}\frac{(-1)1!p^2c_\beta}{2(1+p)(1+2p)}$ & 
	$\binom{3}{2}\frac{(-1)^2 2! p^3c_\beta}{2(1+p)\dots(1+3p)}$ 
	& $\binom{3}{3}\frac{(-1)^3 3! p^4c_\beta}{2(1+p)\dots (1+4p)}$ & $\cdots$ & $0$ \\
	\hline
	$m$ & $\frac{pc_\beta}{2(1+mp)}$
	& $\binom{m-1}{0}\frac{pc_\beta}{2(1+p)}$ & $\binom{m-1}{1}\frac{(-1) 1!p^2c_\beta}{2(1+p)(1+2p)}$ & 
	$\binom{m-1}{2}\frac{(-1)^2 2! p^3c_\beta}{2(1+p)\dots(1+3p)}$ 
	& $\binom{m-1}{3}\frac{(-1)^3 3! p^4c_\beta}{2(1+p)\dots (1+4p)}$ & $\cdots$ & 
	$\binom{m-1}{m-1}\frac{(-1)^{m-1}p^{m}c_\beta}{2(1+p)\dots (1+mp)}$
\end{tabular}
\label{table:BetaDecompCoeff}
\end{table}

\sctOne{Another decomposition of beta terms}\label{sct:AnotherBetaDecomp}
Substituting \eqref{eq:AnotherTn} into \eqref{eq:AnotherBetaTerm},
we can give a combinatorial decomposition of beta terms: for $m \ge 1$,

\begin{align}
	\BetaTerm{m} &= 
	\sum_{m \le n \le N-1} \frac{(-1)^{m+1} (m-1)! p^{m-1} (1+np) }{(1+p)(1+2p) \dots (1+mp)} 
	\binom{n-1}{m-1} \frac{p}{2(1+np)}\sum_{\substack{1 \le i_0 \le N}}
	\sum_{\substack{\bm{S} \subseteq \{i_0 + 1, \dots, N\}}}
	\delta(\# \bm{S}, n) \SubOp{\bm{S}} \StrictCheckOp{i_0; \bm{S}}{} \TransOp{0, \text{move}}{{i}_0} \\
	&= \frac{(-1)^{m+1} (m-1)! p^{m}}{2(1+p)(1+2p) \dots (1+mp)}
	\sum_{\substack{1 \le i_0 \le N}} 
	\bigg[ \sum_{\substack{\bm{S} \subseteq \{i_0 + 1, \dots, N\} \\ \#\bm{S} \ge m}}
	\binom{\# \bm{S}-1}{m-1} \SubOp{\bm{S}} \StrictCheckOp{i_0; \bm{S}}{} 
	\bigg] \TransOp{0, \text{move}}{{i}_0}. \label{eq:BetaDecompAnother1}
\end{align}
To obtain a more convenient formula of  \eqref{eq:BetaDecompAnother1},
we use the following combinatorial equation: for $\bm{S} \subseteq \{i_0 + 1, \dots, N\}$ with $\#\bm{S} \ge m$,
\begin{equation}
	\binom{\# \bm{S}-1}{m-1} \SubOp{\bm{S}} =
	\sum_{\substack{\bm{S}' \subseteq \bm{S} }} \delta(\#\bm{S}', m) \SubOp{\bm{S}'}.
\end{equation}
Then we can perform the following transformation
of the part surrounded by square brackets in \eqref{eq:BetaDecompAnother1}:
\begingroup
\allowdisplaybreaks
\begin{align}
	&\qquad \sum_{\substack{\bm{S} \subseteq \{i_0 + 1, \dots, N\} \\ \#\bm{S} \ge m}}
	\binom{\# \bm{S}-1}{m-1} \SubOp{\bm{S}} \StrictCheckOp{i_0; \bm{S}}{} \\
	&= \sum_{\substack{\bm{S} \subseteq \{i_0 + 1, \dots, N\} \\ \#\bm{S} \ge m}}
	\sum_{\substack{\bm{S}' \subseteq \bm{S} }} 
	\delta(\#\bm{S}', m) \SubOp{\bm{S}'}
	\StrictCheckOp{i_0; \bm{S}}{} \\
	&= \sum_{\substack{\bm{S}' \subseteq \{i_0 + 1, \dots, N\}}} \sum_{\substack{\bm{S} \subseteq \{i_0 + 1, \dots, N\} \\ \bm{S} \supseteq \bm{S}'}}
	\delta(\#\bm{S}', m)
	\SubOp{\bm{S}'}
	\StrictCheckOp{i_0; \bm{S}}{} \\
	&= \sum_{\substack{\bm{S}' \subseteq \{i_0 + 1, \dots, N\}}} \delta(\#\bm{S}', m) \SubOp{\bm{S}'}
	\sum_{\substack{\bm{S} \subseteq \{i_0 + 1, \dots, N\} \\ \bm{S} \supseteq \bm{S}'}} 
	\StrictCheckOp{i_0; \bm{S}}{} \\
	&= \sum_{\substack{\bm{S}' \subseteq \{i_0 + 1, \dots, N\}}} 
	\delta(\#\bm{S}', m) \SubOp{\bm{S}'}
	\CheckOp{\{i_0\} \cup \bm{S}'}{},
\end{align}
\endgroup
where we used \eqref{eq:SummationOfTildeDelta} in the last equality.
Hence, we obtain
\begin{align}
	\BetaTerm{m} &= \frac{(-1)^{m+1} (m-1)! p^{m}}{2(1+p)(1+2p) \dots (1+mp)} 
	\sum_{\substack{1 \le i_0 \le N}} 
	\bigg[ \sum_{\substack{\bm{S}' \subseteq \{i_0 + 1, \dots, N\}}} 
	\delta(\#\bm{S}', m) \SubOp{\bm{S}'}
	\CheckOp{\{i_0\} \cup \bm{S}'}{} 
	\bigg] \TransOp{0, \text{move}}{{i}_0} \\
	&= \frac{(-1)^{m+1} (m-1)! p^{m}}{2(1+p)(1+2p) \dots (1+mp)} 
	\sum_{\substack{1 \le i_0 < i_1 < \dots < i_m \le N }}
	\SubOp{\{i_1, \dots, i_m\}} \CheckOp{\{i_0, i_1, \dots, i_m\}}{} 
	\TransOp{0, \text{move}}{{i}_0}.
	\label{eq:UmDescr}
\end{align}

For example when $m=1$, we have
\begin{align}
	\BetaTerm{1} &= \frac{p}{2(1+p)} \sum_{1 \le i_0 < i_1 \le N} \SubOp{\{i_1\}} 
	\CheckOp{\{i_0, i_1\}}{} \TransOp{0, \text{move}}{i_0} \\
	&= \frac{p}{2(1+p)} \sum_{1 \le i_0 < i_1 \le N}
	(\ShiftOp{i_1}{\dP} + \ShiftOp{i_1}{\dM} - 2 \id^{\otimes N}_H) 
	\CheckOp{\{i_0, i_1\}}{} \TransOp{0, \text{move}}{i_0}.
\end{align}
When $m = N-1$, we have
\begin{align}
	\BetaTerm{N-1} &= \frac{(-1)^{N} N! p^{N-1}}{2(1+p)(1+2p) \dots (1+(N-1)p)} 
	\SubOp{\{2,3,\dots, N\}} 
	\CheckOp{\{1,2,\dots, N\}}{} \TransOp{0, \text{move}}{1}.
\end{align}
For an agent $i_0$ and a subset $\bm{S} \subseteq [N] \setminus \{i_0\}$,
let us define 
\begin{align}
	\BetaTerm{i_0; \bm{S}} \coloneqq 
	\frac{(-1)^{\#\bm{S}} \#\bm{S}! p^{\#\bm{S}-1}}{2(1+p)(1+2p) \dots (1+(\#\bm{S}-1)p)} 
	\SubOp{\bm{S}} \CheckOp{\{i_0\} \cup \bm{S}}{} \TransOp{0, \text{move}}{i_0}.
\end{align}
Then we can rewrite \eqref{eq:UmDescr} as
\begin{align}
	\BetaTerm{m} = \sum_{1 \le i_0 \le N} 
	\sum_{\substack{\bm{S} \subseteq \{i_0+1, \dots, N\}}}
	\delta(\#\bm{S}, m) 
	\BetaTerm{i_0; \bm{S}}. \label{eq:UmDecomp}
\end{align}
Note that we can define $\BetaTerm{i_0; \bm{S}}$ even if
$\bm{S} \not\subseteq \{i_0 + 1, \dots, N\}$.

\begin{rem}[The origin of the beta decomposition]\label{rem:Addendum}
	With the operators $\BetaTerm{i_0; \bm{S}}$, 
 we can define a variant model of supplanting process.
	For a subset $\bm{S}_0 \subseteq \{1, \dots, N\}$, we define
	\begin{equation}
		\TransOp{\bm{S}_0}{} \coloneqq \sum_{i_0 \in \bm{S}_0}
		\sum_{\substack{\bm{S}' \subseteq \{i_0+1, \dots, N\} \cap \bm{S}_0}} 
		\BetaTerm{i_0; \bm{S}'}. \label{eq:TSDecomp}
	\end{equation}
	This operator $\TransOp{\bm{S}_0}{}$ represents the model 
	where supplanting process only occurs on pairs of agents $i, j$ 
	with $i, j \in \bm{S}_0$.
	
	Conversely, we can define $\BetaTerm{i_0; \bm{S}}$ from these variants $\TransOp{\bm{S}_0}{}$ 
	by inclusion-exclusion principle:
	\begin{equation}
		\BetaTerm{i_0; \bm{S}_0} = \sum_{\substack{\bm{S}' \subsetneq \{i_0+1, \dots, N\} \cap
		 \bm{S}_0}} (-1)^{\#\bm{S}' - \#\bm{S}_0}\TransOp{\{i_0\} \cup \bm{S}'}{}.
	\end{equation}
	The beta decomposition was originally derived from this point of view.
\end{rem}

\sctOne{Commutation relation between beta terms and permutation operators} \label{sct:Permutation}

Though it is not used in the main text, we write down a commutation relation between beta terms and permutation operators.

Using relations in Section \ref{sct:PermOperators}, we have the following forms (the third equation is equivalent to \eqref{eq:T0Permutation}):
{\begin{align}
	\SubOp{\bm{S}} \PermOp{\sigma} &= \PermOp{\sigma}\SubOp{\sigma(\bm{S})}, \\
	\CheckOp{\{i_0\} \cup \bm{S}}{} \PermOp{\sigma} &=
	\PermOp{\sigma}\CheckOp{\{\sigma(i_0)\} \cup \sigma(\bm{S})}{}, \\
	\TransOp{0, \text{move}}{i_0} \PermOp{\sigma} &=
	\PermOp{\sigma}\TransOp{0, \text{move}}{\sigma(i_0)}.
\end{align}}

With these three relations, we obtain that
\begin{equation} \label{eq:PermUm}
	\BetaTerm{i_0 ; \bm{S}} \PermOp{\sigma} = 
	\PermOp{\sigma} \BetaTerm{\sigma(i_0) ; \sigma(\bm{S})}.
\end{equation}
This enables us to investigate the commutation relation between $\BetaTerm{m}$ and permutation operators with using \eqref{eq:UmDecomp} or between $\TransOp{\bvec{S}_0}{}$ and permutation operators with using \eqref{eq:TSDecomp}.

\section{Eigenvector centrality and overlap centrality}\label{ECENTER}

In this section we see a relation between the eigenvector centrality and the overlap centrality of the neighbor matrix. 
Let us recall the definitions of their centrality in our context.
The eigenvector centrality is defined as a normalized eigenvector with the maximum eigenvalue of the neighbor matrix $\mathcal{R}$. Such an eigenvector exists uniquely and its components can be taken to be real and positive because of the Perron--Frobenius theorem \cite{Bonacich, Newman}. 
The overlap centrality \eqref{eq:definition_of_overlap_centrality} is the expectation value of how many agents are at the same site with a given agent.

First we describe our settings. We take a state $\ket{P} \in H_X^{\otimes N}$ corresponding to a probability distribution $P(\bvec{x})$.
$\ket{P}$ does not have to be a stationary state of a certain stochastic process.
The neighbor matrix $\mathcal{R} = (r_{ij})_{1 \le i, j \le N}$ of $\ket{P}$ is defined similar to \eqref{eq:definition_of_neighbor_matrix}, explicitly
\begin{equation}
    r_{ij} = \sum_{\bvec{x}}\delta(x_i,x_j)\braket{\bvec{x}|P}.
\end{equation}
We consider a decomposition of the state $\ket{P}$
\begin{align} \label{eq:state_decomp}
    \ket{P}
    = \ket{P_1}+\ket{P_2},
\end{align}
and define matrices $\mathcal{R}^{(\ell)} = (r^{(\ell)}_{ij})$ as
\begin{align}
    r_{ij}^{(\ell)} &= \sum_{\bvec{x}}\delta(x_i,x_j)\braket{\bvec{x}|P_\ell}
\end{align}
for $\ell=1,2$. By definition, we also have a decomposition of the neighbor matrix $\mathcal{R}$
\begin{equation}
    \mathcal{R}
    =
    {\mathcal{R}}^{(1)}+{\mathcal{R}}^{(2)}.
\end{equation}
We assume that the decomposition \eqref{eq:state_decomp} 
satisfies the following two conditions.

\begin{enumerate}[{(i)}]
    \item $\ket{P_1}$ is symmetric under permutations: that is, $\PermOp{\sigma}\ket{P_1}=\ket{P_1}$ for any $\sigma\in\mathfrak{S}_N$.  From this assumption there is a constant $c$ with $r_{ij}^{(1)} = c$ for any $i \neq j$.
    We assume that $c \neq 0$.
    \item The off-diagonal entries $r_{ij}^{(2)}$ of $\mathcal{R}^{(2)}$ are sufficiently smaller than $\abs{c}$: that is, $|r_{ij}^{(2)}| \ll |c|$ for any $i \neq j$. \label{cnd:small}
\end{enumerate}

We modify the matrices $\mathcal{R}^{(1)}, \mathcal{R}^{(2)}$ to $ \tilde{\mathcal{R}}^{(1)} = (\tilde{r}_{ij}^{(1)}),  \tilde{\mathcal{R}}^{(2)} = (\tilde{r}_{ij}^{(2)})$ 
such that all of the diagonal entries of $\tilde{\mathcal{R}}^{(2)}$ are zero and the following decomposition of $\mathcal{R}$ holds:
\begin{equation}\label{eq:decomp_with_fluctuation}
    \mathcal{R}
    =
    \tilde{\mathcal{R}}^{(1)}+\tilde{\mathcal{R}}^{(2)}.
\end{equation}
Considering all of the diagonal entries of $\mathcal{R}$ are one, we can take 
\begin{align}
    \tilde{r}_{ij}^{(1)} = \begin{dcases}
        1 & \text{(if $i=j$)} \\
        r_{ij}^{(1)}=c & \text{(if $i\neq j$)}
    \end{dcases} \quad \text{ and } \quad
    \tilde{r}_{ij}^{(2)} = \begin{dcases}
        0 & \text{(if $i=j$)} \\
        r_{ij}^{(2)} & \text{(if $i\neq j$)}.
    \end{dcases}
\end{align}
From the assumption \eqref{cnd:small}, the matrix $\tilde{\mathcal{R}}^{(2)}$ can be regarded as a perturbative part in $\mathcal{R}$.

The eigenvalue problem of $\tilde{\mathcal{R}}^{(1)}$ can be solved easily: eigenvalues of $\tilde{\mathcal{R}}^{(1)}$ are $1+(N-1)c$ and $1-c$, and their corresponding eigenspaces are $\C \bvec{v}_0$ and the orthogonal complement $(\C \bvec{v}_0)^\perp$, respectively, where $\bvec{v}_0\coloneqq \frac{1}{\sqrt{N}}(1,1,\ldots,1)^T\in\C^N$. In particular, if $\tilde{\mathcal{R}}^{(2)}$ vanishes, then the eigenvector centrality of $\mathcal{R}=\tilde{\mathcal{R}}^{(1)}$ is the vector $\bvec{v}_0$.

Let us use a first-order perturbation theory to calculate the eigenvector centrality $\bv{V}$ of $\mathcal{R}$. Here we introduce an orthonormal system $(\bvec{u}_i)_{i=1}^{N-1}$ of the vector space $(\C \bvec{v}_0)^\perp$. According to the Rayleigh--Schr\"{o}dinger type perturbation theory, we have
\begin{align}
    \bv{V}
    &\propto \bvec{v}_0 + \sum_{i=1}^{N-1}\frac{\bvec{u}_i \bvec{u}_i^\dag\tilde{\mathcal{R}}^{(2)}\bvec{v}_0}{1+(N-1)c-(1-c)}+\mathcal{O}(|\tilde{r}^{(2)}/c|^2)\\
    &= \bvec{v}_0 + \frac{(1-\bvec{v}_0 \bvec{v}_0^\dag)\tilde{\mathcal{R}}^{(2)}\bvec{v}_0}{Nc}+\mathcal{O}(|\tilde{r}^{(2)}/c|^2)\\
    &= \frac{1}{Nc}\mathcal{R}\bvec{v}_0-\bka{\frac{\bvec{v}_0^\dag\tilde{\mathcal{R}}^{(2)}\bvec{v}_0}{Nc}+\frac{1-c}{Nc}}\bvec{v}_0+\mathcal{O}(|\tilde{r}^{(2)}/c|^2)
\end{align}
up to the first order of $|\tilde{r}^{(2)}/c|\coloneqq \max_{1\le i,j\le N}|\tilde{r}_{ij}^{(2)}/c|$.
Since $(\mathcal{R}\bvec{v}_0)_i=\frac{1}{\sqrt{N}}\sum_{j=1}^N r_{ij}=\frac{1}{\sqrt{N}}(O_i-1)$, we obtain
\begin{align}\label{eq:relation_between_V_and_O}
    \bv{V}
    \propto \frac{1}{N^{3/2}c}\bv{O}-\gamma\times(1,1,\ldots,1)^T
    +\mathcal{O}(|\tilde{r}^{(2)}/c|^2),
\end{align}
where $\bv{O}=(O_i)_{i=1}^N$ is the vector consisting of the overlap centrality, and $\gamma$ is a constant which is expressed by
\begin{align}
    \gamma
    = \frac{1}{\sqrt{N}}\bka{\frac{\bvec{v}_0^\dag\tilde{\mathcal{R}}^{(2)}\bvec{v}_0}{Nc}+\frac{2-c}{Nc}}.
\end{align} 

Thus we find that the eigenvector centrality and the overlap centrality are equal up to multiplying by a constant and adding a vector in $\C(1,1,\ldots, 1)^T$. Note that the higher order terms of this perturbation can have a nontrivial $\bv{O}$-dependence, but are ignored in the approximation.

Let us apply this result to the case of the stationary state $\ket{P(\beta,p)}$ of the transition matrix $\TransOp{}{}(\beta,p)$.
Suppose that $p$ is small enough to be able to perform the perturbation expansion \eqref{ptheory}.
Under this assumption, we give the decomposition of $\ket{P(\beta,p)}$ as follows:
\begin{align}
    \ket{P(\beta,p)}
    &= \ket{P_1(\beta, p)}+\ket{P_2(\beta,p)},
    \label{eq:Pto2Vectors}\\
    \ket{P_1(\beta,p)}
    &= C(\beta,p)\ket{P_\text{can}(\beta)}, \\
    \ket{P_2(\beta,p)}
    &= C(\beta,p)\sum_{n=1}^\infty \bka{\GOp (\TransOp{}{}-\TransOp{0}{})}^n\ket{\Pcan(\beta)}.
\end{align}
Let us check that this decomposition satisfies the above condition (i) and (ii).
\begin{enumerate}[(i)]
    \item From \eqref{eq:Pcan_sym}, $\ket{P_1(\beta,p)}$ is symmetric under permutations. Since $\beta, E(\bvec{x})\in\R$, we have
\begin{align}
    c
    = \dfrac{C(\beta,p)}{Z_N(\beta)}\sum_{\bvec{x}\in X^N}\delta(x_1,x_2)e^{-\beta E(\bvec{x})}
    \neq 0.
\end{align}
    \item $\TransOp{}{}-\TransOp{0}{}=\mathcal{O}(p)$ leads to $\tilde{r}_{ij}^{(2)}=\mathcal{O}(p)$,  
    and $c$ is of order of unity in terms of $p$, which follows that $\abs{\tilde{r}^{(2)}/c}=\mathcal{O}(p)$.
\end{enumerate}
Hence, the relation \eqref{eq:relation_between_V_and_O} can be applied in the present case, and we obtain the expression \eqref{ov-centrality}, in which the term of $\mathcal{O}(\abs{\tilde{r}^{(2)}/c}^2)$ in \eqref{eq:relation_between_V_and_O} is replaced 
with $\mathcal{O}(p^2)$. 
Note that the manner of such a decomposition \eqref{eq:Pto2Vectors} is not unique. In particular, one can choose the coefficient of $\ket{\Pcan(\beta)}$ in $\ket{P_1(\beta,p)}$ as any form of $1+\mathcal{O}(p)$, though the values of $c$ and $\gamma$ depend on the manner.

\twocolumngrid

\end{document}